\documentclass[11pt,reqno]{amsart}
\usepackage{amssymb,mathrsfs,graphicx,subfigure, enumerate}
\usepackage{amsmath,amsfonts,amssymb,amscd,amsthm,bbm}
\usepackage{graphicx,colortbl,here}
\usepackage{extpfeil}
\numberwithin{figure}{section}
\numberwithin{table}{section}

\usepackage{kotex}
\usepackage[utf8]{inputenc}
\usepackage[T1]{fontenc}
\usepackage{booktabs}
\usepackage{tabularx, makecell}

\def\mi {{\mathrm i}}

\def \rpsi_i {|\psi_i \rangle}
\def \lpsi_i {\langle \psi_i|}
\def \lrpsi_i{\langle \psi_i | \psi_i \rangle}
\def \rpsi_k {|\psi_k \rangle}
\def \lpsi_k {\langle \psi_k|}
\def \lrpsi_k{\langle \psi_k | \psi_k \rangle}

\newcommand{\bbr}{\mathbb R}

\newcommand{\bbs}{\mathbb S}

\newcommand{\kp}{\kappa}

\newcommand{\veps}{\varepsilon}

\newcommand{\ba}{\begin{aligned}}
\newcommand{\ea}{\end{aligned}}

\newcommand{\be}{\begin{equation}}
\newcommand{\ee}{\end{equation}}

\newcommand{\bbsm}{        {      \bbs^{m}  }    }

\newtheorem{theorem}{Theorem}[section]
\newtheorem{lemma}{Lemma}[section]

\newtheorem{proposition}{Proposition}[section]
\newtheorem{remark}{Remark}[section]
\newtheorem{definition}{Definition}[section]

\begin{document}


\title[]{Dimension-dependent order parameter selection in higher-order Kuramoto dynamics on spheres: continuous versus quantized regimes}

 \author[H. Huh]{Hyungjin Huh}
\address[H. Huh]{\newline Department of Mathematics, \newline Chung-Ang University, Seoul 06974, Republic of Korea}
\email{huh@cau.ac.kr}

\author[D. Kim]{Dohyun Kim}
\address[D. Kim]{\newline Department of Mathematics Education and Institute of Pure and Applied Mathematics, \newline Sungkyunkwan University, Seoul 03063, Republic of Korea, \newline School of Computational Sciences, \newline  Korea Institute for Advanced Study,   Seoul 02455, Republic of Korea}
\email{dohyunkim@skku.edu}

\thanks{\textbf{Acknowledgment.}
The work of H. Huh was supported by  the National Research Foundation of Korea (NRF) grant funded by the Korea government (RS-2025-16067086)  and the work of D. Kim was supported by the National Research Foundation of Korea(NRF) grant funded by the Korea government(MSIT)(RS-2024-00454452) and by the Visiting Professorship at Korea Institute for Advanced Study.}

\begin{abstract}
We study a high-dimensional Kuramoto model with attractive pairwise coupling and a repulsive higher-order effect.  Although the pairwise interaction favors synchronization,  the higher-order interaction prevents complete synchronization and selects an intermediate level of coherence corresponding  to the balanced equilibria. For the unit sphere of dimension at least two, we provide  an explicit basin of attraction for the balanced equilibria. We also classify nonzero-mean equilibria and show that all non-balanced equilibria  are linearly unstable. This indicates that balanced states are the only natural stable candidates in the higher dimensions. 

In contrast, on the circle, the same model reduces to a higher-order Kuramoto-type model which exhibits a qualitatively different selection mechanism that depends fundamentally on the dimension. In this case,   population imbalance at finite $N$ prevents exact stationary balance and  instead generates a common angular drift producing phase-locked states. We construct basins of attraction for two-cluster locked states, in which the phases split into two groups, and identify the corresponding finite-$N$ selected values of the order parameters. This phenomenon is called continuous-versus-quantized order-parameter selection.  We further demonstrate, through a root-selection mechanism, why such two-cluster locked states are  typically observed in most simulations. Finally, we show that two-cluster locked states with a large population imbalance are linearly unstable, which explains why only certain locked states are dynamically robust.

\end{abstract}

\keywords{Kuramoto model, high-dimensional Kuramoto model, higher-order interaction, finite-size corrections, phase-locked states, balanced states}

\makeatletter
\@namedef{subjclassname@2020}{%
  \textup{2020} Mathematics Subject Classification}
\makeatother

\subjclass[2020]{34D20, 34D05, 34D06, 82C22.}  

\date{\today}

\maketitle

\tableofcontents

\section{Introduction}
\setcounter{equation}{0}

Synchronization of weakly coupled phase oscillators has long served as a central paradigm for collective behavior in nonlinear dynamics. The Kuramoto model \cite{Ku} provides one of the simplest and most significant mathematically tractable frameworks for studying emergent dynamics. In the classical Kuramoto model, attractive pairwise coupling promotes phase synchronization and is naturally measured by the order parameter. The basin of attraction  for synchronization, the role of the frequency distribution, finite-size effects and many variants of the model have been extensively studied in the literature.  

\subsection{Higher-order interactions} 
Recently, there has been growing  interest in nonlinear dynamics involving higher-order or non-pairwise interactions. This development is part of the broader study of networks beyond pairwise interactions where the dynamics might depend not only on edges but also on group interactions encoded by hypergraphs or simplicial complexes. In oscillator networks, such higher-order couplings are known to generate phenomena that are absent  in the classical pairwise Kuramoto model including multistability, abrupt transitions, quasiperiodic behavior, cluster formation and explosive synchronization.

In particular, three-body and simplicial Kuramoto-type interactions have been shown to produce rich phase diagrams. The authors in \cite{TA} studied multistable attractors generated by three-body interactions. Later in \cite{SA}, the authors demonstrated abrupt desynchronization and extensive multistability in oscillator systems with higher-order interactions  (see Section \ref{sec:literature}).


\subsection{High-dimensional  model} 
The purpose of this paper is to show that attractive pairwise interactions do not necessarily guarantee complete synchronization once higher-order interactions are present.  We study the following high-dimensional model with higher-order interactions \cite{HK,Ko}. 
\begin{equation} \label{main}
 \dot x_i  = \kp_1(x_c-\langle x_i,x_c\rangle x_i) + \kp_2\langle x_i,x_c\rangle ( x_c - \langle x_i,x_c\rangle x_i)
\end{equation}
with the average $x_c$ and initial data $x_i^0$ on the unit sphere:
\begin{equation} \label{init}
x_c:= \frac1N \sum_{k=1}^N x_k,\quad x_i(0) = x_i^0 \in \bbs^m, \quad m\geq 2, \quad  i\in [N]:=\{1,\cdots,N\}.
\end{equation}
Here $\langle \cdot, \cdot \rangle$ is an inner product in $\Bbb{R}^{m+1}$.
 The first term with $\kp_1$ is the usual   mean-field pairwise interaction, and the second term with $\kp_2$ is a higher-order interaction whose strength depends on the alignment of $x_i$ with the mean field $x_c$.   
  
 We mainly focus on the regime
\begin{equation*} 
\kp_1>0,\quad \kp_1+\kp_2<0,\quad \alpha:= \frac{\kp_1}{-\kp_2}\in (0,1).
\end{equation*}
In this regime, the attractive pairwise interaction is counteracted by the repulsive  higher-order interactions. In particular, highly synchronized oscillators experience a negative effective coupling so that complete synchronization might lose its attracting character. Precisely, if we write \eqref{main} as 
\[
\dot x_i = -\kp_2(\alpha - \langle x_i,x_c\rangle)(x_c - \langle x_i,x_c\rangle x_i),
\]
then near a completely synchronized state, where $\langle x_i,x_c\rangle \approx 1$,  the  effective coupling $-\kp_2(\alpha - \langle x_i,x_c\rangle)$ becomes negative. In addition, oscillators are attracted toward the mean-field when their projection onto it is below $\alpha$, whereas they are repelled when this projection is above $\alpha$. This sign-changing structure  provides the main mechanism underlying the emergence of intermediate coherence.

\subsection{Main results for the high-dimensional model}
Our first main result concerns the high-dimensional sphere $\bbs^m$ with $m\geq2$.  We construct an explicit basin of attraction for balanced equilibria satisfying
\begin{equation} \label{balanced}
\|x_c(t)\|^2 = \alpha.
\end{equation}
We call a state satisfying \eqref{balanced}   a balanced state. For the initial data in this basin,  we indeed show that $\langle x_i,x_c\rangle(t)\to \alpha$ for all $i\in [N]$ which in turn implies \eqref{balanced}. The proof is based on an energy estimate for the deviations $\langle x_i,x_c\rangle -\alpha$ combined with a bootstrap argument controlling the pairwise distances from a reference balanced configuration. The key coercivity estimate relies on the availability of at least two transverse directions, which is precisely why the assumption $m\geq2$ is essential. The role of the transverse direction is crucial in the construction of balanced states. Once the mean direction is fixed, a balanced configuration requires the particles to have the prescribed projection onto the mean-field while their remaining components must cancel each other in the transverse space. In dimensions $m\geq2$, this transverse space has enough freedom to arrange these components with zero sum. This geometric cancellation is what allows the system to realize a stationary state with intermediate coherence. See Section \ref{sec:3} for detailed arguments.

In fact, in our numerical simulations, balanced states are robustly observed for   generic initial configurations. 
 Our basin result alone does not explain why   balanced equilibria appear to be the dynamically relevant stable states.
  We then complement our convergence result with a classification and linear stability analysis of equilibria. For nonzero mean-field (or average), equilibria split into three types: (i) pure bipolar equilibria, (ii) pure balanced equilibria and (iii) mixed polar-balanced equilibria. We show that every nonzero-mean equilibrium that is not a pure balanced equilibrium is linearly unstable. Thus, among nonzero-mean equilibria, pure balanced equilibria are the only candidates not ruled out by linear instability under the regime $\kp_1>0$ and $\kp_1+\kp_2<0$. 

On the other hand, although it lies somewhat outside the scope of the main results, we include three additional convergence estimates for \eqref{main} (desynchronizability for $\kp_1<0$, non-desynchronizability for $\kp_1>0$ and bipolar-synchronizability for $\kp_1<0$ and $\kp_1+\kp_2>0$) to further develop the analysis initiated in our previous work \cite{HK} (see Appendix \ref{sec:app.C}).

\subsection{One-dimensional model}

The circle case $\bbs^1$ is not merely the one-dimensional version of the sphere model; it is a distinct case in which the mechanism of balanced state selection for $m\geq2$ breaks down and is instead replaced by a  finite-size mechanism. Thus, $\bbs^1$ plays a distinguished role.  Under the ansatz $x_i = (\cos\theta_i, \sin\theta_i)$, the model \eqref{main} on the sphere   reduces exactly to a Kuramoto-type  model with higher-order interactions  \cite{SA,TA,WZX,ZSBPL} 
\begin{equation} \label{kura}
  \dot \theta_i = \frac{\kp_1}{N}\sum_{k=1}^N \sin (\theta_k - \theta_i) + \frac{\kp_2}{2N^2} \sum_{j,k=1}^N \sin(\theta_j+\theta_k - 2\theta_i)
\end{equation}
with initial data
\begin{equation} \label{kurainit}
\theta_i(0) = \theta_i^0,\quad i\in [N].
\end{equation}

\subsection{Main results for the one-dimensional model}

\subsubsection{Equilibrium vs phase-locked states}
At the level of the order parameter, the mechanism which selects an intermediate coherence level is analogous to that on $\bbs^m$ with $m\geq2$; however, the resulting asymptotic states are fundamentally different. In higher dimensions, the transverse latitude has positive dimension (see Figure \ref{fig2-1}), so the oscillators can redistribute themselves along transverse directions and balance the mean-field while remaining stationary. On $\bbs^1$, by contrast, the corresponding latitude consists of only two points. Hence, a stationary balanced configuration requires the  two points to be occupied by equal numbers of oscillators. This integer constraint is generally not possible for odd $N$. The resulting population imbalance cannot be compensated geometrically and is instead converted into a common angular drift leading to phase-locked states, not a stationary equilibrium. See Sections \ref{sec:6.1} and \ref{sec:6.2}.

\subsubsection{Two-cluster locked states}

Among a rich family of phase-locked configurations, we show that two-cluster locked states, where all phases split into two groups, are observed generically by using a root-selection mechanism.  If we write \eqref{kura} in a mean-field form 
\[
\dot \theta_i = \kp_1 R\sin (\Psi-\theta_i ) + \frac{\kp_2R^2}{2} \sin (2(\Psi-\theta_i)) =: f_R(\theta_i -\Psi),\quad\frac1N \sum_{j=1}^N e^{\mi\theta_j} =:  Re^{\mi \Psi},
\]
then phase-locked states are determined by the locking equation $f_R(\phi_i) = \Omega$ where 
\[
\phi_i := \theta_i - \Psi,\quad f_R(u) := -\kp_2 R(R\cos u - \alpha) \sin u.
\]
 Since $\sum_{i=1}^N \sin \phi_i=0$ holds and  $f_R$ is a trigonometric polynomial  of degree two, we conclude that the scalar equation $f_R(\phi_i) = \Omega$ has exactly four roots in general (see Figure \ref{fig6-2}). Since the stability of these roots alternates along the circle, precisely two of the four roots are attracting, while the remaining two are repelling. Consequently, the dynamics  selects two attracting phase locations, providing a dynamical justification for the generic emergence of two-cluster locked states. This mechanism  emerges from the combination of self-consistency, the geometry of the locking equation and root stability. See  Section \ref{sec:6.3}.

In addition, we construct an explicit basin of attraction for two-cluster locked states.  See Theorem \ref{T5.2}.


\subsubsection{Finite-$N$ selection of the order parameter}
Suppose that the phases split into two clusters:
\[
\theta_i = \begin{cases}
\theta_A,\quad i\in A,\\
\theta_B,\quad i\in B, 
\end{cases}
\]
where $A\sqcup B = [N]$ and $(|A|,|B|)=:(m,\ell)$. Let  $\Delta:=\theta_A-\theta_B$  denote  the difference of the two phases and  let $k:=|m-\ell|$  denote  the population imbalance of the two clusters. Then, we see that $\Delta$ satisfies
\[
\dot \Delta = -\sin \Delta ( \kp_1 + 2\kp_2 pq + \kp_2(p^2 + q^2)\cos \Delta), \quad p:=\frac mN,\quad q:=\frac\ell N
\]
which gives
\begin{equation} \label{A-30}
 \kp_1 + 2\kp_2 pq + \kp_2(p^2 + q^2)\cos \Delta=0.
\end{equation}
In addition, the order parameter becomes 
\begin{equation} \label{A-31}
R^2 = |pe^{\mi \theta_A} + qe^{\mi \theta_B}|^2 = p^2 + q^2 + 2pq\cos \Delta.
\end{equation}
By using \eqref{A-30} and \eqref{A-31}, we derive
\[
R^2   = \alpha + \frac{2(1-\alpha)k^2}{N^2+k^2}=:R_{N,k}^2
\]
which gives a finite-$N$ selection to the squared order parameter.  Thus, the values $R_{N,k}^2$ are not arbitrary values associated with rotating profiles; rather,  they are selected by the self-consistency condition for two-cluster locked states with a prescribed population imbalance.  When $k=0$, we recover the value $\alpha$, which is frequently (or generically) observed in $\bbs^m$ with $m\geq2$. However, in $\bbs^1$, the value $\alpha$ can occur only under an equal split and is not available for odd $N$. Hence, it is considerably less likely to be observed.   Note that the population imbalance  appears as a common angular velocity $\Omega$ which is explicitly given as 
\[
\Omega = -\kp_2(1-\alpha)^\frac32 \frac{(1-\eta^2)\sqrt{\alpha+\eta^2}}{(1+\eta^2)^2} \eta,\quad \eta := \frac kN.
\]
Hence, whenever $\eta\neq0$, one has $\Omega\neq0$, so the corresponding locked state is genuinely rotating rather than stationary.   See Section \ref{sec:5.2}.

\subsubsection{Nonuniform selection of the $R_{N,k}$ branches}

Although a complete global classification of all possible locked profiles is not attempted here, the argument above identifies the robust mechanism responsible for the observed $R_{N,k}$ branches. In addition, we also analyze that not all $R_{N,k}$-branches are equally likely to be observed: sufficiently large population imbalance leads to linear instability of the corresponding branch. Precisely, we denote $\phi_+:= \theta_A-\Psi$ and $\phi_- := \theta_B-\Psi$. Then, the eigenvalues for $\phi_\pm$ are given as $\lambda_\pm := f_R'(\phi_\pm)$. After some calculation, we see that $\lambda_+$ is represented as 
\[
\lambda_+ = \kp_2 (1-\alpha) \frac{P_\alpha(\eta)}{(1+\eta^2)^2},  \quad P_\alpha(\eta) = \alpha(1-3\eta+\eta^2+\eta^3)-4\eta^3.
\]
 Since $\kp_2<0$, the condition $\lambda_+<0$ is equivalent to $P_\alpha(\eta)>0$. 
Since $P_\alpha(0)=\alpha>0$, $P_\alpha(1)=-4$ and   $P_\alpha(\cdot)$ is strictly decreasing on $[0,1]$ for $0<\alpha<1$, there exists a unique root $\eta_c=\eta_c(\alpha)\in (0,1)$ such that $P_\alpha(\eta_c)=0$ (see Figure \ref{fig6-3}). Hence, if 
\[
\eta<\eta_c, \quad \textup{i.e.,}\quad k<N\eta_c,
\]
then $\lambda_+$ is  negative. Similarly, we see that $\lambda_-<0$. Hence, the two-cluster locked state is linearly stable. Consequently, large-imbalance branches are excluded by linear instability,  whereas small-imbalance branches are stable and remain close to the balanced value $\alpha$. See Section \ref{sec:smallk}.

 Figure \ref{fig:organization} schematically summarizes the organization of phase-locked states.
 
 \begin{figure}[H]
\centering
\mbox{
\subfigure{  
\includegraphics[width=0.5\textwidth]{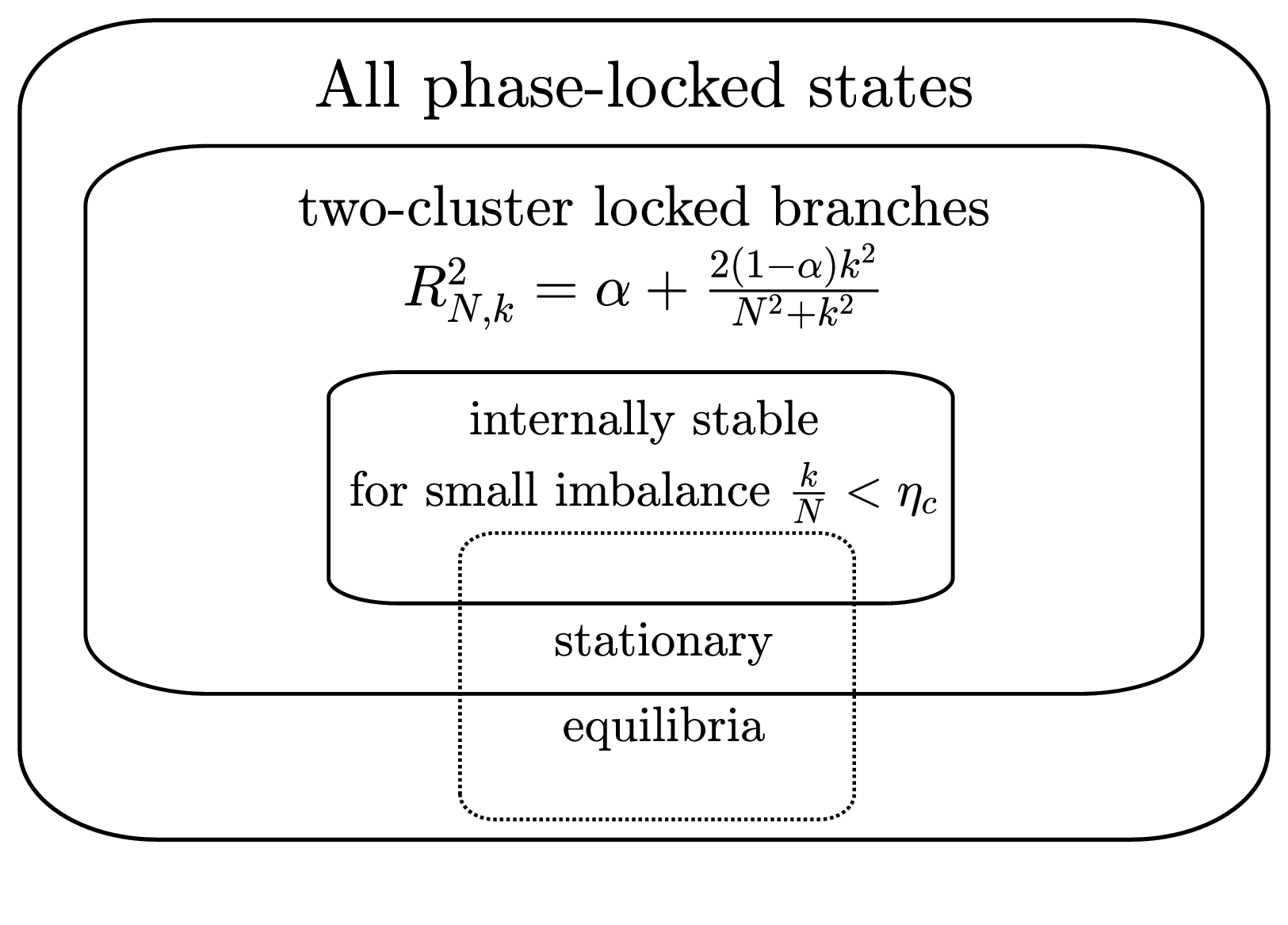}}  
}
\caption{ Organization of phase-locked states: Stationary equilibria are phase-locked states with zero angular velocity. The two-cluster branches form a distinguished subset in which small-imbalance branches satisfy the internal stability condition. For even $N$, the branch $k=0$ is simultaneously stationary, balanced and internally stable.} \label{fig:organization}
\end{figure}

\subsection{Unified picture of the main results}

The main results identify  a dimension-dependent mechanism of the order parameter selection. In the regime
\[
\kp_1>0,\quad \kp_1+\kp_2<0,\quad \alpha=\frac{\kp_1}{-\kp_2}\in (0,1),
\]
attractive pairwise interactions do not lead to complete synchronization. Instead, the dynamics selects an intermediate level of coherence whose
structure depends essentially on the dimension of the underlying sphere. For $\mathbb S^m$ with $m\geq2$, the order parameter converges to the continuously selected value
\[
R_\infty^2=\alpha.
\]
However, on $\mathbb S^1$, the transverse geometry permits only two admissible locations and the finite population imbalance produces the
discrete family
\[
R_\infty^2 = R_{N,k}^2 = \alpha+\frac{2(1-\alpha)k^2}{N^2+k^2}.
\]
Thus, the same competition between pairwise attraction and higher-order repulsion leads to \emph{continuous order-parameter selection} on
higher-dimensional spheres and \emph{quantized order-parameter selection} on the circle. The main results are summarized in Table~\ref{tab:dimension-dichotomy}.

The contrast between  $\bbs^1$ and $\bbs^m$ with $m\geq2$  can be justified as follows. On $\bbs^1$, finite-$N$ population imbalance cannot be removed by redistributing transverse components, because the fixed-projection level set consists of only two points. On the other hand, in higher dimensions, the transverse space has enough room to absorb the population imbalance geometrically which allows the system to converge to a stationary balanced state. Hence, the same higher-order interaction produces finite-$N$ phase locking on $\bbs^1$ and stationary balanced configuration on higher-dimensional spheres.

For the pairwise model, corresponding to $\kp_2=0$, the dynamics is a gradient flow with the mean-field energy. In the attractive regime, the dynamics consequently selects complete synchronization through a mechanism that is essentially insensitive to the dimension of the sphere. In particular, the equilibrium condition
\[
x_c = \langle x_i,x_c\rangle x_i \quad \Longrightarrow \quad \|x_c\|^2 = \langle x_i,x_c\rangle^2
\]
  yields that all $x_i$ and $x_c$ are parallel, except for $x_c=0$. Hence, this condition only requires each agent to be parallel or antiparallel to the mean-field and does not involve the geometry of the transverse space. On the other hand, the higher-order interaction changes this picture qualitatively. In the regime $\kp_1>0$ and $\kp_1+\kp_2<0$, the equilibrium condition factorizes into 
\[
(\alpha-\langle x_i,x_c\rangle) (x_c - \langle x_i,x_c\rangle x_i)=0.
\]
Hence, an agent might be stationary without being polar if $\langle x_i,x_c\rangle = \alpha$. This condition activates a fixed-projection latitude whose normalized transverse component belongs to $\bbs^{m-1}$. The resulting self-consistency problem requires a zero-sum configuration on this transverse sphere. 

The geometric origin of this dimension-dependent dichotomy can be formulated in terms of fixed-projection level sets. For $p\in\mathbb S^m$ and $c\in(-1,1)$,
\[
\mathcal L_c(p) := \{x\in\mathbb S^m:\langle x,p\rangle=c\} \cong \mathbb S^{m-1}.
\]
Thus, on $\mathbb S^1$, the admissible transverse set is $\mathbb S^0$ which consists of only two points.  Hence, exact cancellation requires an integer population constraint. However, for $m\geq2$, the transverse level set has positive dimension and allows continuous redistribution of the transverse components. This geometric mechanism is formalized in Proposition~\ref{prop:fixed-projection}.

\begin{table}[t]
\centering
\caption{Dimension-dependent coherence selection covered by our convergence results under the regime
$\kp_1>0$ and $\kp_1+\kp_2<0$.}
\label{tab:dimension-dichotomy}
\renewcommand{\arraystretch}{1.35}
\small
\begin{tabularx}{0.96\textwidth}{
    >{\centering\arraybackslash}p{0.14\textwidth}
    >{\raggedright\arraybackslash}p{0.25\textwidth}
    >{\centering\arraybackslash}p{0.27\textwidth}
    >{\raggedright\arraybackslash}X}
\toprule
State space
& Limiting state
& Order-parameter selection
& Geometric mechanism
\\
\midrule

$\mathbb S^m$ ($m\geq2$)
&
Stationary balanced equilibrium
&
$\displaystyle
R_\infty^2=\alpha
$
\newline
\newline
&
Continuous geometric balancing on a positive-dimensional level set
\\[3mm]

$\mathbb S^1$
&
Two-cluster locked state (in general)
&
$\displaystyle
R_\infty^2
=
\alpha+
\frac{2(1-\alpha)k^2}{N^2+k^2}
$
\newline
\newline
&
Discrete population imbalance on a two-point level set
\\

\bottomrule
\end{tabularx}
\end{table}
 
Finally, Figure \ref{fig:summary} summarizes the dimension-dependent selection mechanism underlying our main results. 

 \begin{figure}[H]
\centering
\mbox{
\subfigure{  
\includegraphics[width=1\textwidth]{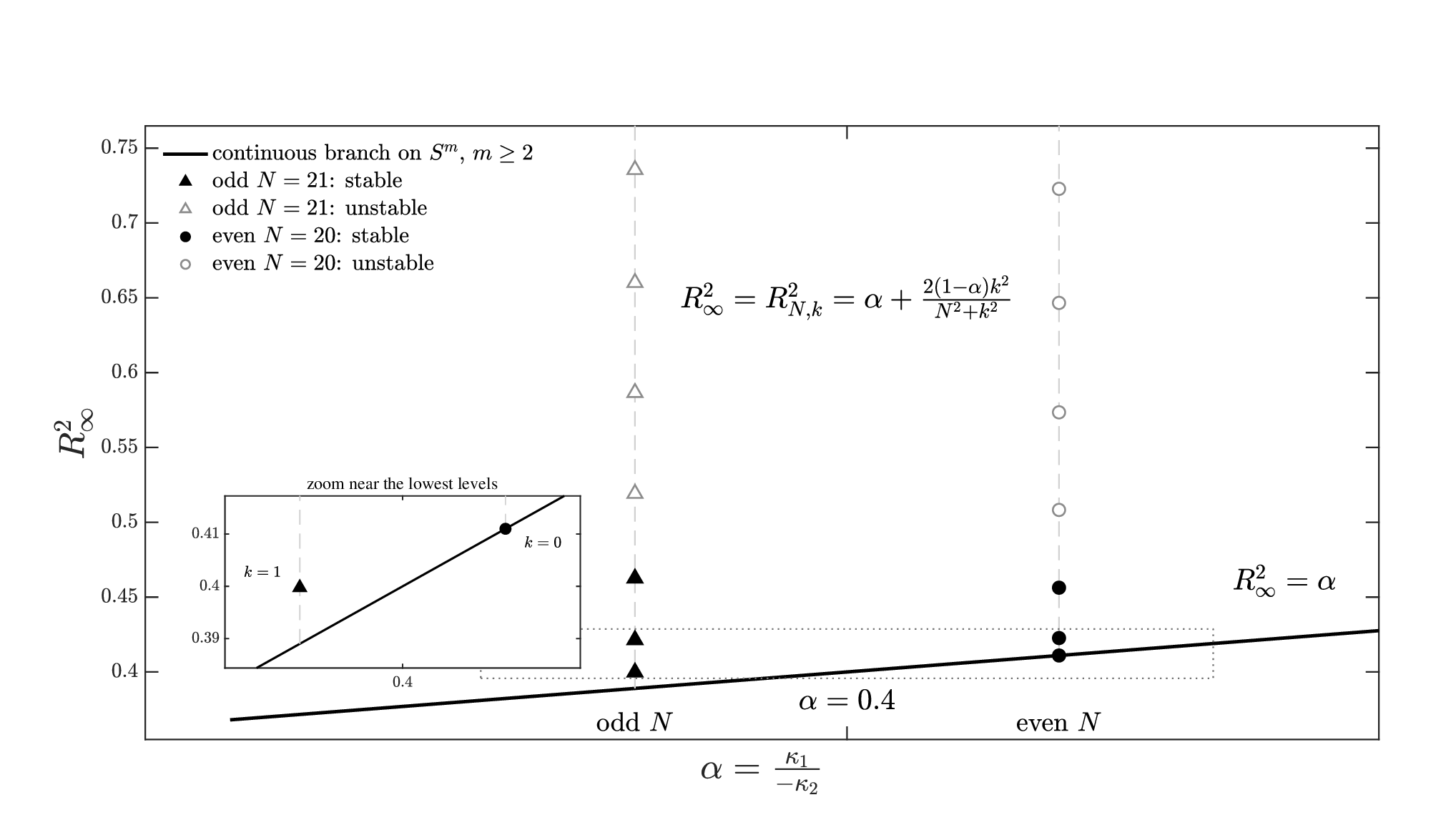}\label{fig:3-1}}
}
\caption{ Continuous versus quantized order-parameter selection: The solid curve represents the continuous selection $R_\infty^2 = \alpha$ on $\bbs^m$ with $m\geq2$, whereas the circle $\bbs^1$ admits discrete finite-$N$ levels $R_{N,k}^2$. Filled and open markers denote stable and unstable levels, respectively. The zoomed inset highlights the parity effect near the lowest levels: $k=0$ is admissible for even $N$, while odd $N$ starts from $k=1$.}  \label{fig:summary}
\end{figure}

 The rest of the paper is organized as follows. In Section \ref{sec:2}, we introduce the basic definitions and derive several identities that will be frequently used throughout the paper. We also review previous results relevant to the present work.  In Section \ref{sec:3}, we provide an explicit basin of attraction for   balanced equilibria on  $\bbs^m$ with $m\geq2$. In Section \ref{sec:4}, we classify   equilibria on the sphere  and establish the  linear instability of all non-balanced nonzero-mean equilibria. In Section \ref{sec:5}, we turn to the model on $\bbs^1$ which takes the form of a Kuramoto-type model with higher-order interactions,  and establish convergence to two-cluster locked states so that the squared order parameter converges to $R_{N,k}^2$. In Section \ref{sec:6}, we classify stationary equilibria and phase-locked states,  analyze the stability of the locked branches and  explain the selection mechanism of the finite-$N$ values $R_{N,k}$. In Section \ref{sec:7}, we perform numerical simulations that support our theoretical results and provide qualitative insights beyond the analytical results. In Appendices \ref{sec:app.A} and \ref{sec:app.B}, we provide the proofs of Lemma \ref{L3.1} and Theorem \ref{T5.2}, respectively. In Appendix \ref{sec:app.C}, we present complementary dynamics of \eqref{main}.

\section{Preliminaries}\label{sec:2}
\setcounter{equation}{0}

\subsection{Basic notation and identities}
First, we see that the unit sphere $\bbs^m$ is positively invariant along \eqref{main}--\eqref{init}, since the vector field in \eqref{main} is tangent to the unit sphere. Since the proof is straightforward, we omit it. 
\begin{lemma}
Let $\{x_i\}_{i=1}^N$ be a solution to \eqref{main}--\eqref{init}. Then, $x_i(t) \in \bbs^m$ for all $t\geq0$ and $i\in [N]$. 
\end{lemma}

For \eqref{main}, we introduce the order parameter, the pairwise correlations, the projections onto the mean field, their deviations from $\alpha$ and the corresponding quadratic energy:
\begin{align} \label{functions}
\begin{aligned}
&R(t):= \|x_c(t)\|,\quad h_{ij}(t) :=  \langle x_i(t),x_j(t)\rangle,\quad g_{ij}(t):= 1- h_{ij}(t),\\
&q_i(t) := \langle x_i(t),x_c(t)\rangle,\quad  u_i(t) := q_i(t) - \alpha, \quad E(t) := \frac12\sum_{i=1}^N u_i(t)^2.
\end{aligned}
\end{align}
Similarly for \eqref{kura}, we define the order parameter $R$ and common phase $\Psi$:
\[
z=\frac1N \sum_{j=1}^N e^{\mi \theta_j} =  Re^{\mi \Psi},
\]
where we use the same notation $R$ for the order parameter in both \eqref{main} and \eqref{kura}, as the meaning will be clear from the context. 

Throughout the main analysis, unless otherwise stated, we assume
\[
\kp_1>0,\quad \kp_1+\kp_2<0,\quad \alpha=\frac{\kp_1}{-\kp_2}\in (0,1),\quad K:=-\kp_2>0.
\]

\begin{lemma} \label{lem:identity}
Let $\{x_i\}_{i=1}^N$ be a solution to \eqref{main}--\eqref{init}. Then, we have
\begin{align*}
&\textup{(i)}~~  \frac12\frac{d}{dt}R^2 =\frac{-\kp_2}{N}\sum_{i=1}^N  (\alpha-q_i)(R^2-q_i^2). \\
&\textup{(ii)}~~\frac{d}{dt} (1-h_{ij})  = \kp_2 (u_i-u_j)^2  -\kp_2 \bigl[\alpha(u_i+u_j)+u_i^2+u_j^2\bigr](1-h_{ij}). \\
&\textup{(iii)}~~\frac{dE}{dt} =  -\frac{K\alpha}{2N} \sum_{i,j=1}^N g_{ij}(u_i + u_j)^2 +  \frac{K}{2N}\sum_{i,j=1}^N (u_i + u_j)(u_i- u_j)^2 \\
&\hspace{1.7cm} - \frac{K}{2N}\sum_{i,j=1}^N g_{ij}(u_i + u_j)(u_i^2 + u_j^2),
\end{align*}
where $K:=-\kp_2>0$.
\end{lemma}

 \begin{definition}
 Let $\{x_i\}_{i=1}^N$ be a solution to \eqref{main}--\eqref{init} and $\{\theta_i\}_{i=1}^N$ be a solution to \eqref{kura}--\eqref{kurainit}. 
 \begin{enumerate}
\item We say that system \eqref{main} or \eqref{kura}  exhibits complete synchronization if and only if 
\[
\lim_{t\to\infty} R(t) =1.
\]
\item We say that system \eqref{main} or \eqref{kura} exhibits complete desynchronization if and only if 
\[
\lim_{t\to\infty} R(t) =0. 
\]
\item We say that system \eqref{main} exhibits complete bipolar synchronization if and only if there exists a nonempty proper subset  $A\subset \{1,\cdots,N\}$ such that up to relabeling,
\begin{align*}
&\lim_{t\to\infty} \|x_i(t) - x_j(t)\| =0,\quad \forall (i,j)\in A\times A, \quad  \forall (i,j)\in A^c\times A^c. \\
&\lim_{t\to\infty} \|x_i(t) + x_j(t)\| =0,\quad \forall (i,j)\in A\times A^c.
\end{align*}
\item A solution to system \eqref{main} or \eqref{kura} tends to a balanced state if and only if 
\[
\lim_{t\to\infty}R(t)^2 = \begin{cases}
\vspace{0.3cm}\alpha \quad \textup{for \eqref{main}},  \\
\displaystyle \alpha + \frac{2(1-\alpha)k^2}{N^2+k^2} \quad \textup{for \eqref{kura}}.
\end{cases}
\]
\item If there exists an initial region with a positive  measure with respect to the natural product surface measure on $(\bbs^m)^N$ leading to complete synchronization (or complete desynchronization), then system \eqref{main} is called synchronizable (or desynchronizable). 
\end{enumerate}
 \end{definition}

\subsection{Geometry of fixed-projection states}

The distinction between the circle and  higher-dimensional spheres can be seen at the level of a simple geometric constraint. For $p\in\mathbb S^m$ and $c\in(-1,1)$, we recall the fixed-projection level set
\[
\mathcal L_c(p) = \left\{ x\in\mathbb S^m:\langle x,p\rangle=c \right\}.
\]
The following observation provides the geometric basis for the dimension-dependent selection mechanism studied below.

\begin{proposition}
\label{prop:fixed-projection}
Let $p\in\mathbb S^m$ and $c\in(-1,1)$. Then, every $x\in\mathcal L_c(p)$ can be uniquely represented as
\[
x = cp+\sqrt{1-c^2}\,y, \qquad y\in  \Sigma_p^{m-1}, 
\]
where
\[
\Sigma_p^{m-1} := \bbs^m \cap p^\perp =  \left\{ y\in \bbs^m: \langle y,p\rangle =0 \right\}.
\]
In particular,
\[
\mathcal L_c(p)\cong\mathbb S^{m-1}.
\]
Moreover, let $x_1,\cdots,x_N\in\mathcal L_c(p)$ and write
\[
x_i = cp+\sqrt{1-c^2}\,y_i, \qquad y_i\in\Sigma_p^{m-1}.
\]
Then
\[
\frac1N\sum_{i=1}^N x_i=cp \quad\Longleftrightarrow\quad \sum_{i=1}^N y_i=0.
\]
Consequently, the following dimensional dichotomy holds: 
\begin{enumerate}
\item[(i)]
If $m=1$, then
\[
\Sigma_p^0=\mathbb S^0=\{\nu,-\nu\}
\]
for some unit vector $\nu\in p^\perp$. Hence,
\[
\sum_{i=1}^N y_i=0
\]
is possible if and only if the two points $\nu$ and $-\nu$ are occupied by the same number of particles. In particular, $N$ must be even.

\item[(ii)]
If  $m\geq2$, then $\Sigma^{m-1}_p$ has positive dimension. For every $N\geq2$, there exist $y_1,\cdots,y_N\in\Sigma^{m-1}_p$ satisfying
\[
\sum_{i=1}^N y_i=0.
\]
Moreover, such zero-sum configurations form a continuous family.
\end{enumerate}
\end{proposition}

\begin{proof}
Since $\bbr^{m+1} = \operatorname{span}\{p\}\oplus p^\perp$, every $x\in\mathcal L_c(p)$ can be written uniquely as
\[
x=cp+v, \qquad v\in p^\perp.
\]
Using $\|x\|=1$, we obtain 
\[
1 = \|x\|^2  = c^2+\|v\|^2 \quad \Longrightarrow \quad v= \sqrt{1-c^2}y 
\]
for a unique $y\in\Sigma_p^{m-1}$. This shows the first assertion. 

For a configuration $x_i=cp+\sqrt{1-c^2}\,y_i,$ we have
\[
\frac1N\sum_{i=1}^N x_i = cp + \frac{\sqrt{1-c^2}}{N} \sum_{i=1}^N y_i.
\]
Thus,
\[
\frac1N\sum_{i=1}^N x_i=cp \quad \Longleftrightarrow \quad \sum_{i=1}^N y_i=0.
\]

If $m=1$, then $p^\perp$ is one-dimensional and thus, $\Sigma_p^0=\{\nu,-\nu\}$. Hence, the zero-sum condition holds if and only if the multiplicities of $\nu$ and $-\nu$ coincide. In addition, $N$ should be even. 

If $m\geq2$, choose orthonormal vectors $e_1,e_2\in p^\perp$ and set 
\[
y_j = \cos\frac{2\pi(j-1)}{N}\,e_1 + \sin\frac{2\pi(j-1)}{N}\,e_2, \qquad j=1,\ldots,N.
\]
Then, $\sum_{j=1}^N y_j=0.$ Since the transverse sphere has positive dimension, these configurations can be continuously varied when the zero-sum constraint is preserved. 
\end{proof}

%
%

 \begin{figure}[H]
\centering
\mbox{
\subfigure{  
\includegraphics[width=1\textwidth]{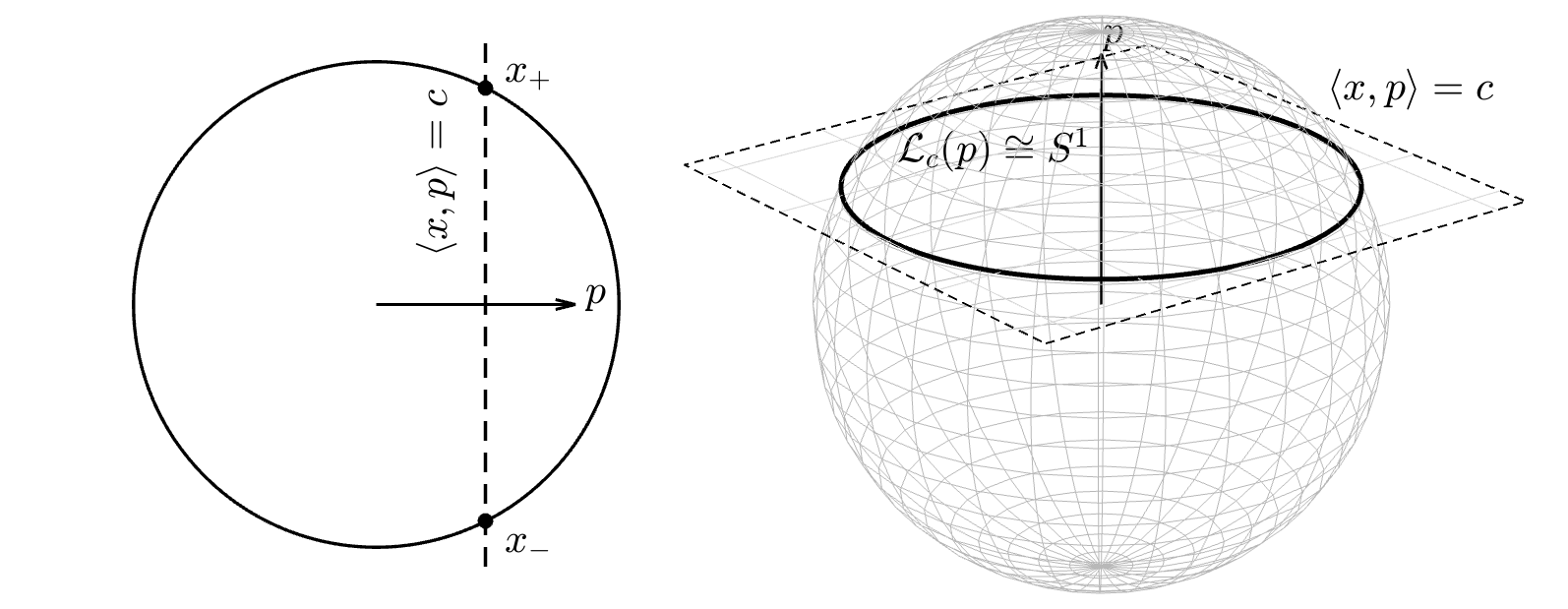}\label{fig:3-1}}
}
\caption{   Fixed-projection geometry on $\bbs^1$ and $\bbs^2$: the admissible set consists of two points on the circle, whereas it forms a positive-dimensional latitude on the higher-dimensional sphere. }  \label{fig2-1}
\end{figure}

\subsection{Kuramoto reduction}
Let $m=1$ and write $x_i = (\cos\theta_i,\sin\theta_i)$. Then, $\theta_i$ satisfies \eqref{kura}:
\begin{equation*} 
\dot \theta_i = \frac{\kp_1}{N}\sum_{k=1}^N \sin (\theta_k - \theta_i) + \frac{\kp_2}{2N^2} \sum_{j,k=1}^N \sin(\theta_j+\theta_k - 2\theta_i).
\end{equation*}

\subsection{Previous results} \label{sec:literature}

The extensions of the Kuramoto model to higher-dimensional unit spheres   replace scalar phases by positions on the unit sphere and reveal dynamical features that have no direct counterpart on $\bbs^1$. In particular, dimension-dependent transitions were discussed in the generalized $D$-dimensional Kuramoto model where odd and even dimensions may exhibit qualitatively distinct routes to synchronization \cite{CGO1}. Higher-order interactions on spheres have also been shown to exhibit nontrivial static configurations, including equally spaced states that remain coherent even when pairwise interactions alone are repulsive \cite{Lo22,Lohe2022}. A related and particularly striking result is that purely contrarian oscillators can synchronize when nonpairwise interactions are introduced \cite{Ko}. We also mention the work \cite{DKM},  which emphasizes the distinction between synchronization scenarios in odd and even dimensions. See also \cite{K21} for another type of higher-order coupling which is reminiscent of the bi-harmonic Kuramoto model. The present work provides a complementary phenomenon: even conformist oscillators with attractive pairwise interactions need not converge to the completely synchronized state when the competing repulsive  higher-order coupling is sufficiently strong.

More broadly, higher-order interactions are now known to fundamentally modify the collective dynamics of coupled oscillators. For instance, they induce abrupt synchronization switching and hysteresis \cite{SA}, stabilize multicluster states \cite{XS2021} and generate multistability among complete synchronization, incoherence, and two-cluster states \cite{LMHN}. It is important to note  that these effects depend not only on the interaction order but also on how group interactions are represented. In \cite{ZLB2}, the same higher-order structure can enhance synchronization when represented as a hypergraph but suppress it when represented as a simplicial complex. Thus, the geometric form of the higher-order coupling is an essential component of the dynamics, rather than a secondary modeling choice.

Recent studies have also highlighted that local linear stability alone does not determine which coherent state is observed from generic initial data. In \cite{ZSBPL}, the authors showed that higher-order interactions can increase the linear stability of an attractor while simultaneously decreasing its basin of attraction. This distinction between local stability and global basin geometry is closely related to the basin-dependent selection of synchronized, desynchronized, bipolar and phase-locked states considered in this work. Complementary numerical results in \cite{MNC} show that weak higher-order interactions may enhance synchronization, although strong higher-order coupling generally suppresses it. These results demonstrate that  non-pairwise interactions cannot be characterized simply as synchronizing or desynchronizing: their effect depends on the coupling strength, interaction structure and initial configuration. See \cite{LMHN2025,LMZL} for systematic derivations of higher-order phase models and \cite{MNGC} for phase chimera states on nonlocal hyperrings. See \cite{KP,MG} for the competition between higher-harmonic interactions, \cite{MNB} for a phase reduction method and \cite{BBK,BGHS} for mathematical approaches for higher-order interactions. We also refer the reader to \cite{Battiston2026,BC,Bi,BBK,BBC} for comprehensive reviews.

In \cite{HK}, the present authors studied the system \eqref{main}--\eqref{init}. In particular, we showed that the line $\kp_1+\kp_2=0$ serves as the critical threshold for synchronizability. The present work further identifies $\kp_1=0$ as the critical line for desynchronizability. Moreover, although \cite{HK} stated that only complete synchronization and complete desynchronization compete in the regime $\kp_1+\kp_2>0$ and $\kp_1<0$, we here show that this regime also admits complete bipolar synchronization. Thus, depending on the initial configuration, complete synchronization, complete desynchronization and complete bipolar synchronization all emerge for the same parameter  
regime. Lastly, we recall the model \eqref{main} with $\kp_2=0$, called the swarm sphere model \cite{Lo09,O1,T,Zhu1}.
Recent studies include \cite{CO} which investigates microscopic instability mechanism underlying macroscopic bursting phenomena, and \cite{DLG} where discontinuous transitions arise from feedback through the global order  parameter. We also refer the reader to \cite{C-H,H-K-R18,M1} for related synchronization scenarios.

\section{Convergence estimates on $\bbs^m$ with $m\geq2$} \label{sec:3} 
\setcounter{equation}{0}

\subsection{Exclusion of complete synchronization and complete desynchronization}
In this subsection, we show that  when $\kp_1>0$ and $\kp_1+\kp_2<0$, complete synchronization and complete desynchronization cannot emerge for generic initial data. 

\begin{theorem} \label{T3.1} 
Suppose that
\[
\kp_1>0,\quad \kp_1+\kp_2<0,\quad \alpha= \frac{\kp_1}{-\kp_2} \in (0,1).
\]
Then, the order parameter cannot converge to  $0$ or $1$ for generic initial data. Precisely, if $R_0>0$, then $R(t)$ cannot converge to zero, and if $R_0<1$, then $R(t)$ cannot converge to 1.
\end{theorem}

\begin{proof}
(i) We first exclude convergence to complete desynchronization when $R_0>0$.   Suppose to the contrary that $R(t)$ converges to zero. We indeed show that this is possible only if $R_0=0$. Since $R(t)\to0$, there exists $T>0$ such that
\[
R(t) <\alpha,\quad t\geq T.
\]
Then, for every $i\in [N]$, we have
\[
q_i(t) = \langle x_i(t),x_c(t)\rangle \leq \|x_c(t)\|<\alpha,\quad \textup{or}\quad \alpha - q_i(t) >0,\quad t\geq T.
\]
Recall the order parameter identity in Lemma \ref{lem:identity}:
\[
\frac12 \frac{d}{dt} R^2 = \frac K N \sum_{i=1}^N (\alpha - q_i) (R^2 - q_i^2)\geq 0,\quad t\geq T.
\]
Hence, $R^2(t)$ is non-decreasing on $t\geq T$. However, by assumption  $R(t)^2\to0$, we have
\[
R(t) = 0,\quad t\geq T.
\]
Then, by uniqueness of an ODE, we should have $R_0=0$. This contradicts the assumption $R_0>0$. \newline

\noindent (ii) Suppose to the contrary that $R(t)$ converges to 1. Then, we have
\[
q_i(t) = \langle x_i(t),x_c(t)\rangle   \to 1,\quad i\in [N].
\]
Since $\alpha<1$, there exists $T>0$ such that
\[
\alpha - q_i(t)<0,\quad i\in [N],\quad t\geq T.
\]
Again, recalling the order parameter identity, 
\[
\frac12 \frac{d}{dt} R^2 = \frac K N \sum_{i=1}^N (\alpha - q_i) (R^2 - q_i^2) \leq0,\quad t\geq T.
\]
Thus, $R(t)^2$ is non-increasing on $t\geq T$. Since $R^2(t)\leq 1$, while $R^2$ is non-increasing on $[t,\infty)$ and converges to $1$, we necessarily have $R(t) \equiv 1$ for $t\geq T$. Again, uniqueness of an ODE gives $R_0=1$. This is  a contradiction.
\end{proof}

\subsection{Basin of attraction for a balanced state}
In this subsection, we provide a basin of attraction leading to the balanced state.

\subsubsection{Construction of an equilibrium for a balanced state}
First, we construct a reference balanced equilibrium. Since $m\geq2$, we can choose three orthonormal vectors in $\bbr^{m+1}$:
\[
p, \, e_1, \,e_2 \in \bbr^{m+1}
\]
and set 
\[
\theta_i := \frac{2\pi(i-1)}{N},\quad y_i:= (\cos\theta_i)e_1 + (\sin\theta_i)e_2.
\]
Then, we have
\[
\|y_i\|=1,\quad \langle y_i,p\rangle =0,\quad \sum_{i=1}^N y_i=0.
\]
Define a target equilibrium
\[
x_i^*:= \sqrt\alpha p + \sqrt{1-\alpha } y_i,\quad i\in [N]
\]
which satisfies 
\[
x_c^* := \frac1N \sum_{i=1}^N x_i^* = \sqrt \alpha p,\quad \|x_c^*\|^2 = \alpha,\quad q_i^*:= \langle x_i^*,x_c^*\rangle =\alpha.
\]
We verify that $X^* :=(x_1^*,\cdots, x_N^*)$ becomes an equilibrium. For later use, we denote
\[
g_{ij}^* := 1-\langle x_i^*,x_j^*\rangle. 
\]
Since $\langle x_i^*,x_j^* \rangle = \alpha + (1-\alpha)\langle y_i,y_j\rangle$, we have
\[
g_{ij}^* = (1-\alpha)(1-\cos(\theta_i - \theta_j)).
\]
Define the distance from the reference pairwise profile
\[
\Delta(t) := \max_{1\leq i,j\leq N} |g_{ij}(t) - g_{ij}^*|.
\]

\subsubsection{Coercivity estimate} 
Our goal is to show that $E(t)$ converges to zero. In order to derive a dissipative estimate for $E$, a coercivity estimate is needed.  Define
\[
\mathcal Q_*(u) := \frac{1}{2N} \sum_{i,j=1}^N g_{ij}^*(u_i + u_j)^2, 
\]
where $u_i = \langle x_i,x_c\rangle - \alpha$ is introduced in \eqref{functions}.
\begin{lemma} \label{L3.1}
Let $\{x_i\}$ be a solution to \eqref{main}--\eqref{init}. Then, for $N\geq3$ and every $u=(u_1,\cdots, u_N)\in \bbr^N$, 
\[
\mathcal Q_*(u) \geq \frac{1-\alpha}{2}\sum_{i=1}^N u_i^2 =:   2\chi_* E,\qquad \chi_* = \frac{1-\alpha}{2}.
\]
\end{lemma}

\begin{proof}
Since the proof is rather lengthy, we postpone it to Appendix \ref{sec:app.A}. 

\end{proof}


\begin{remark}
The coercivity estimate itself is a spectral property of the regular-polygon reference profile. The assumption $m\geq2$ enters through the construction of this profile which requires a two-dimensional plane in $p^\perp$. In Remark \ref{rem:A.1}, we show that $\chi_*$ can be zero in the one-dimensional case. 
\end{remark}

 \subsubsection{Explicit basin of attraction}
 We are now ready to state the main theorem. 

\begin{theorem} \label{T3.2}
Suppose that
\[
m\geq2,\quad N\geq3,\quad  \kp_1>0,\quad \kp_1+\kp_2<0, \quad \alpha= \frac{\kp_1}{-\kp_2}\in (0,1).
\]
Fix $\veps\in (0,\frac{\chi_*}{2})$ with $\chi_* = \frac{1-\alpha}{2}$ and assume that the initial data satisfy 
\begin{align} \label{C-70}
\begin{aligned}
&\textup{(i)}~~E_0 <\frac{\alpha^2(\chi_* - 2\veps)^2}{72},\\
&\textup{(ii)}~~ \Delta_0 +  \frac{4\alpha\sqrt{2E_0} + 8 E_0}{\alpha(\chi_*-2\veps) -6\sqrt{2E_0}} <\veps, 
\end{aligned}
\end{align}
where $E_0=E(0)$ and $\Delta_0=\Delta(0)$.
Let $\{x_i\}$ be a solution to \eqref{main}--\eqref{init}. Then, we have
\[
\lim_{t\to\infty} \langle x_i(t),x_c(t)\rangle = \alpha ,\quad i\in [N].
\]
Consequently, the squared order parameter converges to 
\[
\lim_{t\to\infty} \|x_c(t)\|^2 = \alpha . 
\]
In this situation, there exists a balanced equilibrium $x_i^\infty \in \bbsm$ such that
\[
\lim_{t\to\infty} x_i(t) = x_i^\infty,\quad i\in [N].
\]

\end{theorem}

\begin{proof}
Define a temporal set
\[
\mathcal T:= \{ T>0:~~ \Delta(t) <\veps,\quad t\in [0,T)\}.
\]
Then, due to \eqref{C-70}(ii), the set $\mathcal T$ is nonempty. Thus, $T_* := \sup \mathcal T$ is well-defined. We claim that $T_* = \infty$. Suppose to the contrary that 
\[
T_* <\infty.
\]
Recall from Lemma \ref{lem:identity} that for $K=-\kp_2>0$, 
\begin{align*}
\dot E &= -\frac{K\alpha}{2N} \sum_{i,j=1}^N g_{ij}(u_i + u_j)^2 +  \frac{K}{2N}\sum_{i,j=1}^N (u_i + u_j)(u_i- u_j)^2 \\
&\quad - \frac{K}{2N}\sum_{i,j=1}^N g_{ij}(u_i + u_j)(u_i^2 + u_j^2) \\
& =: \mathcal I_{11} + \mathcal I_{12} + \mathcal I_{13}.
\end{align*}
$\bullet$ (Estimate of $\mathcal I_{11}$): For $t\in [0,T_*)$, we use Lemma \ref{L3.1} to find 
\begin{align*}
\mathcal I_{11} & = -\frac{K\alpha}{2N} \sum_{i,j=1}^N g_{ij}(u_i + u_j)^2 = -\frac{K\alpha}{2N} \sum_{i,j=1}^N g_{ij}^* (u_i + u_j)^2- \frac{K\alpha}{2N} \sum_{i,j=1}^N (g_{ij}-g_{ij}^*) (u_i + u_j)^2 \\
&\leq -   K\alpha \chi_*\sum_{i=1}^N u_i^2 + \frac{K\alpha \veps}{2N} \sum_{i,j=1}^N (u_i + u_j)^2 \leq -K\alpha( \chi_* - 2\veps) \sum_{i=1}^N u_i^2 = -2K\alpha(\chi_*-2\veps) E.
\end{align*}
$\bullet$ (Estimate of $\mathcal I_{12}$): We use
\[
|u_i| \leq \sqrt{2E},\quad \sum_{i,j=1}^N (u_i - u_j)^2 \leq 4N E
\]
to find 
\begin{align*}
\mathcal I_{12} & =  \frac{K}{2N}\sum_{i,j=1}^N (u_i + u_j)(u_i- u_j)^2 \leq 4K\sqrt2 E^\frac32.
\end{align*}
$\bullet$ (Estimate of $\mathcal I_{13}$): Similarly, we get
\begin{align*}
\mathcal I_{13} & = - \frac{K}{2N}\sum_{i,j=1}^N g_{ij}(u_i + u_j)(u_i^2 + u_j^2) \leq 8K\sqrt2 E^\frac32.
\end{align*}
To this end, we have for $t\in [0,T_*)$, 
\begin{equation} \label{exp}
\dot E \leq -2K\alpha(\chi_*-2\veps) E + 12\sqrt2 K E^\frac32 =: -aE + bE^\frac32, 
\end{equation}
where $a  = 2K\alpha(\chi_*-2\veps)$ and $b = 12\sqrt 2 K$. On the other hand, we also recall from Lemma \ref{lem:identity}
\[
\dot g_{ij} = -K(u_i - u_j)^2 + K \Big[  \alpha(u_i + u_j) + u_i^2 + u_j^2    \Big] g_{ij}.
\]
Hence, we observe
\begin{align*}
|\dot g_{ij}| &\leq K(u_i - u_j)^2 + K \Big[  \alpha |u_i + u_j| + u_i^2 + u_j^2 \Big] g_{ij} \\
&\leq 8KE  + 2K ( 2\sqrt2 \alpha \sqrt E + 4E) = 4\sqrt2 K\alpha \sqrt E + 16K E, 
\end{align*}
where we used
\[
|u_i \pm u_j|\leq 2\sqrt{2E},\quad u_i^2 + u_j^2 \leq 4E.
\]
Hence, we have
\[
\Delta(t) \leq \Delta_0 + 4\sqrt2 K \alpha\int_0^t \sqrt{E(s)}ds + 16K\int_0^t E(s) ds. 
\]
Since we assume
\[
E_0 <\left( \frac ab\right)^2 =\frac{\alpha^2(\chi_* - 2\veps)^2}{72},
\]
we have
\[
\int_0^t E(s) ds \leq \frac{E_0}{a-b\sqrt{E_0}},\quad \int_0^t \sqrt{E(s)} ds \leq \frac{2\sqrt{E_0}}{a-b\sqrt{E_0}}.
\]
Thus, we have
\[
\Delta (t)  <\Delta_0 + 4\sqrt2 K\alpha \frac{2\sqrt{E_0}}{a-b\sqrt{E_0}} + 16K \frac{E_0}{a-b\sqrt{E_0}} <\veps
\]
which is precisely \eqref{C-70}(ii). However, by the maximal bootstrap time, we have
\[
\lim_{t\to T_*} \Delta(t) = \veps.
\]
This contradiction gives 
\[
T_* = \infty.
\]
Therefore, it follows from \eqref{exp} that $E(t)$ converges to zero exponentially. For the convergence to equilibrium, we observe 
\[
|\dot x_i(t)|\leq K\sqrt{2E(t)}.
\]
Since $E(t)$ converges to zero exponentially, we conclude that there exists $x_i^\infty$ such that $x_i(t) \to x_i^\infty\in \bbs^m$. 

\end{proof}

 \begin{remark}
 
 The conditions in \eqref{C-70} define an explicit neighborhood of the reference balanced equilibrium $X^*$. In particular, the basin obtained above is nonempty and contains an open set of initial configurations. For every initial configuration in this basin, the solution converges exponentially to the set of balanced equilibria. 

 \end{remark}
 
\section{Stability: Why is the balanced state generic?} \label{sec:4} 
\setcounter{equation}{0}
In this section,  we classify all nonzero-mean equilibria and show that every non-balanced equilibrium is linearly unstable. Combined with the result for the basin of attraction established in Section \ref{sec:3}, this identifies the pure balanced equilibria as the only class of nonzero-mean equilibria not ruled out by linear instability. If we write \eqref{main} as 
\[
\dot x_i = -\kp_2(\alpha - q_i) (x_c - q_ix_i),
\]
then at equilibrium, we have two cases for each $i\in [N]$:
\[
\textup{(i)}~~q_i = \alpha \qquad \textup{(ii)}~~ x_c = q_i x_i.
\]
The first  one corresponds to a latitude particle, whereas the second one implies that $x_i$ is polar, i.e., $x_i = \pm p$. In what follows, we consider nonzero-mean equilibrium, i.e., 
\[
R\neq 0,\quad p:= \frac{x_c}{\|x_c\|}=\frac{x_c}{R}\neq0.
\]
\subsection{Settings}
Define three disjoint subsets of $[N]$:
\begin{align*}
&P_+:=\{ i: x_i = p\},\quad  n_+:= |P_+|, \\
&P_-:= \{i: x_i = -p\},\quad n_-:= |P_-|,  \\
&P_L:=\{ i : \langle x_i,x_c\rangle = \alpha,~~x_i \neq \pm p\},  \quad n_L:= |P_L|.
\end{align*}
Then, we have
\[
P_+ \cup P_-\cup P_L = \{1,\cdots,N\}.
\]
For $i\in P_L$, we decompose $x_i$ into $p$-direction and $p^\perp$-direction:
\[
x_i = \langle x_i,p\rangle p + (x_i - \langle x_i,p\rangle p) = \frac\alpha R p + \left( x_i - \frac\alpha Rp \right).
\]
Here, we see
\[
\left\|x_i - \frac\alpha R p \right\|^2 = 1-\frac{\alpha^2}{R^2},\quad y_i := \frac{x_i - \frac\alpha R p}{ \sqrt{1- \frac{\alpha^2}{R^2}}}.
\]
Then, we represent $x_i$ for $i\in P_L$:
\[
x_i = \frac\alpha R p + \sqrt{1-\frac{\alpha^2}{R^2}}y_i,\quad y_i \in p^\perp,\quad \|y_i\|=1.
\]
Note that $\langle x_i,p\rangle =  \frac\alpha R \leq 1$ and if $i\in P_L$, then since $x_i\neq \pm p$, we should have 
\[
R>\alpha.
\]
We observe
\begin{align*}
NRp& = N x_c = \sum_{i=1}^N x_i =( n_+ - n_-)p + \sum_{i\in P_L} \left( \frac\alpha R p + \sqrt{ 1-\frac{\alpha^2}{R^2} }y_i \right) \\
& = \left(  n_+ - n_- + \frac{n_L\alpha}{R} \right) p +  \sqrt{ 1-\frac{\alpha^2}{R^2}} \sum_{i\in P_L} y_i \\
& =  \left(  n_+ - n_- + \frac{n_L\alpha}{R} \right) p.
\end{align*}
Since the left-hand side has no component in $p^\perp$, comparison of the transverse component yields $\sum_{i\in P_L} y_i = 0$.    This gives the self-consistency condition:
\begin{equation} \label{selfconsistency}
NR^2 - (n_+ - n_-) R - n_L \alpha =0. 
\end{equation}
Now, for the case of $x_c\neq0$, we classify the cases:
\begin{enumerate}  
\item Pure bipolar equilibria:
\[
n_L = 0,\quad n_+ + n_- = N.
\]
\item Mixed polar-balanced equilibria:
\[
n_L>0,\quad n_++ n_->0.
\]
\item Pure balanced equilibria:
\[
n_L= N,\quad n_+=n_-=0. 
\]
\end{enumerate}

For the summary of the stability results, we refer the reader to Table \ref{tab:Sm-stability-short}.

\begin{remark}
When we consider $P_L$, since $x_i\neq \pm p$, we should have $R>\alpha$.  Thus, if $n_L=1$, then $\sum_{i\in P_L } y_i = y_\ell$ and $y_\ell$ should be zero. However, it contradicts $\|y_\ell\|=1$. Hence, $n_L\geq2$. 
\end{remark}

 \begin{figure}[H]
\centering
\mbox{
\subfigure{  
\includegraphics[width=0.5\textwidth]{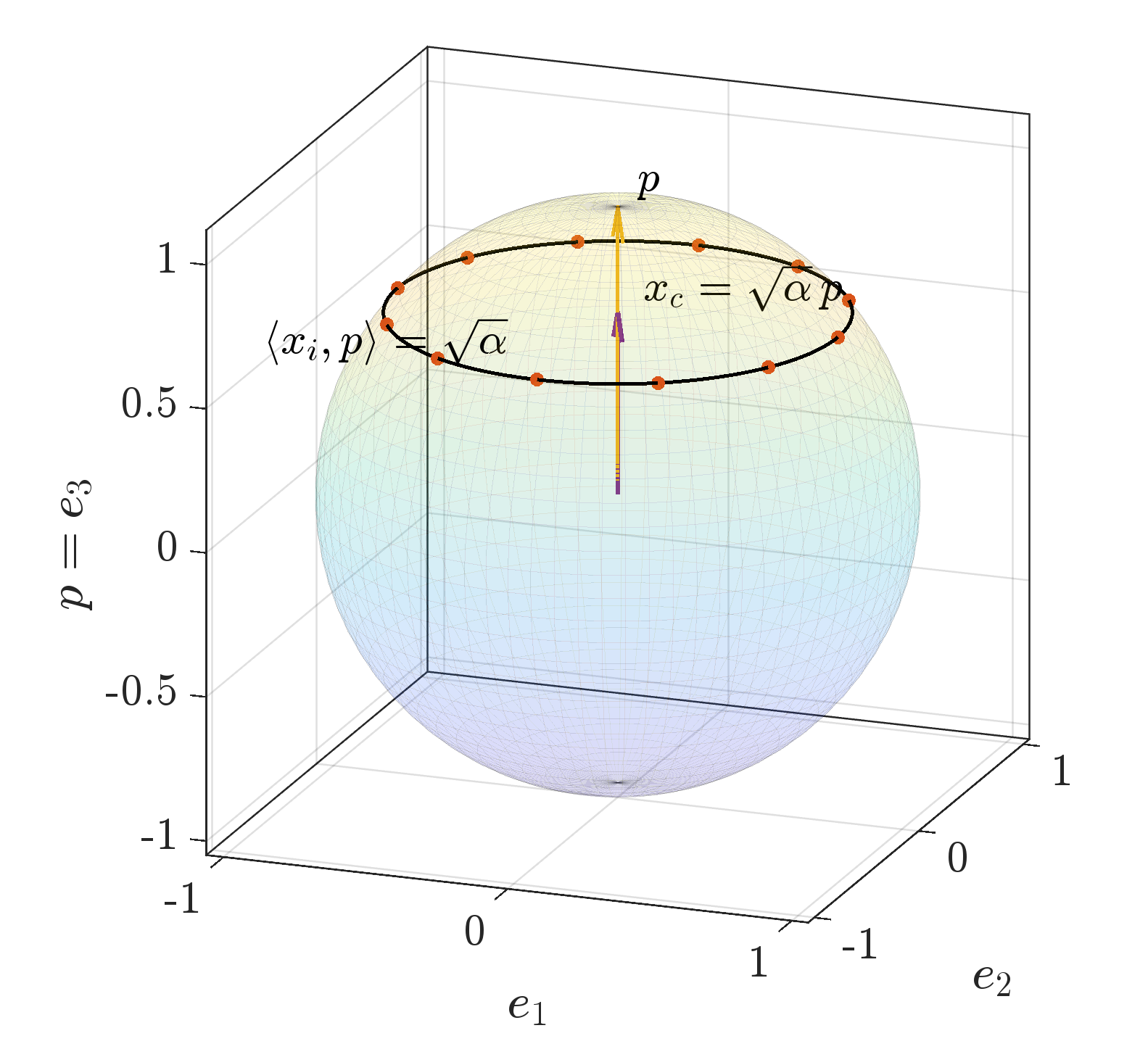}\label{fig:3-1}}
}
\caption{  Balanced state geometry on $\bbs^2$: The particles share the prescribed projection $\langle x_i,p\rangle = \sqrt\alpha$ while their transverse components cancel in the transverse directions, which yield $x_c = \sqrt\alpha p$ and $\|x_c\|^2=\alpha$. }  \label{fig5-1}
\end{figure}

\subsection{Linearization}
Let $X^* = (x_1^*,\cdots,x_N^*)$ be an equilibrium. Consider a tangent perturbation near $x_i^*$:
\[
x_i^\veps = x_i^* + \veps \xi_i + O(\veps^2),\quad \xi_i \in T_{x_i^*} \bbsm.
\]
Denote
\begin{align*} \label{G-10}
\begin{aligned}
&x_c^* := \frac1N \sum_{i=1}^N x_i^*,\quad q_i^* := \langle x_i^*,x_c^*\rangle,\\
& \delta x_i := \xi_i,\quad  \delta x_c := \frac1N \sum_{i=1}^N \xi_i, \quad \delta q_i := \langle \xi_i,x_c^*\rangle + \langle x_i^*,\delta x_c\rangle.
\end{aligned}
\end{align*}
For $\delta q_i$, we have
\begin{align*}
q_i^\veps  &: = \langle x_i^\veps,x_c^\veps\rangle = \langle x_i^* + \veps \xi_i + O(\veps^2) , x_c^* + \veps \delta x_c + O(\veps^2)\rangle \\
& = \langle x_i^*,x_c^*\rangle + \veps (  \langle \xi_i,x_c^*\rangle + \langle x_i^*,\delta x_c\rangle ) + O(\veps^2) \\
& = q_i^* + \veps \delta q_i + O(\veps^2).
\end{align*}
We  denote the vector field of \eqref{main}:
\[
F_i(X) := K(\alpha - q_i) ( x_c - q_ix_i).
\]

\begin{lemma}
The linearization of \eqref{main} is given as
\[
\dot \xi_i = K \Big[  -\delta q_i( x_c^* - q_i^* x_i^*) + (\alpha - q_i^*)(\delta x_c - \delta q_i x_i^* - q_i^*\xi_i)        \Big].
\]
\end{lemma}

\begin{proof}
It suffices to  consider the variation of $F_i(X)$ at $X=X^*$. First, we see
\begin{align*}
\alpha  - q_i^\veps  = \alpha-q_i^* - \veps \delta q_i + O(\veps^2)
\end{align*}
and hence
\[
\delta( \alpha - q_i^\veps) = -\delta q_i.
\]
Second, we observe
\begin{align*}
x_c^\veps - q_i^\veps x_i^\veps & = (x_c^* + \veps \delta x_c + O(\veps^2)) - (q_i^* + \veps \delta q_i + O(\veps^2))(x_i^* + \veps \xi_i + O(\veps^2)) \\
& = x_c^* - q_i^* x_i^*+ \veps ( \delta x_c - \delta q_i x_i^* - q_i^*\xi_i ) + O(\veps^2)
\end{align*}
and hence
\[
\delta (x_c^\veps - q_i^\veps x_i^\veps) = \delta x_c - \delta q_i x_i^* - q_i^*\xi_i .
\]
Therefore, we obtain
\begin{align*}
\delta F_i(X) & = K\Big[ \delta(\alpha-q_i)(x_c^* - q_i^*x_i^*)  +  (\alpha - q_i^*)\delta(x_c-q_ix_i)       \Big] \\
& =  K \Big[  -\delta q_i( x_c^* - q_i^* x_i^*) + (\alpha - q_i^*)(\delta x_c - \delta q_i x_i^* - q_i^*\xi_i)        \Big].
\end{align*}
\end{proof}

In the following subsections, we use the linearization to find a growing mode to verify that some equilibria are linearly unstable. 

\subsection{Bipolar equilibria are linearly unstable}
Consider 
\[
n_L=0
\]
and a bipolar equilibrium
\[
x_i^* = \sigma_i p ,\quad \sigma_i \in \{1,-1\}.
\]
Then, 
\[
q_i^* = \langle x_i^*,x_c^*\rangle = \sigma_i R,\quad R=\frac{n_+-n_-}{N}>0.
\]
Choose any unit vector $z\in p^\perp$ and consider the following transverse perturbations of the form:
\[
\xi_i = \eta_i z.
\]
Since $z\in p^\perp$, we have
\begin{align*}
\delta x_c& =\frac1N \sum_{j=1}^N \xi_j =  \bar \eta z,\quad \bar\eta := \frac1N \sum_{j=1}^N \eta_j,\\
\delta q_i & =  \langle \xi_i,x_c^*\rangle + \langle x_i^*,\delta x_c\rangle = \langle \eta_iz, Rp\rangle + \langle \sigma_i p, \bar\eta z\rangle=0 .
\end{align*}
 In addition, since $x_i^*$ are bipolar, we have
\[
x_c^* - q_i^* x_i^* = 0.
\]
Hence, the linearized equation becomes
\[
\dot\xi_i = K(\alpha - \sigma_i R) (\delta x_c - \sigma_i R\eta_i z) =  K(\alpha - \sigma_i R) (\bar \eta - \sigma_i R\eta_i )z
\]
which gives
\begin{equation} \label{G-30}
\dot \eta_i = K(\alpha - \sigma_i R) (\bar \eta - \sigma_i R \eta_i).
\end{equation}

\subsubsection{Both clusters are nonempty} We suppose that both clusters are nonempty, i.e., 
\[
n_+>0,\quad n_->0.
\]
Consider the two-dimensional subspace
\[
\eta_i =\begin{cases}
u,\quad i\in P_+, \\
v,\quad i\in P_-.
\end{cases}
\]
Then, $\bar \eta = \frac{n_+u + n_-v}{N}$ and \eqref{G-30} reduces to 
\begin{align*}
\dot u & = K(\alpha - R) (\bar \eta - Ru) = K(\alpha-R) \frac{n_-}{N}(u+v), \\
\dot v &=  K(\alpha + R)(\bar \eta + Rv) = K(\alpha+R) \frac{n_+}{N} (u+v).
\end{align*}
We add these two equations to find
\[
\frac{d}{dt} (u+v) = K\left[ (\alpha-R)\frac{n_-}{N} + (\alpha+R)\frac{n_+}{N}\right] (u+v) = K(\alpha + R^2)(u+v), 
\]
where we used $n_+-n_-= NR$. Hence, this bipolar equilibrium with two nonempty clusters has a positive eigenvalue $K(\alpha+R^2)>0$ and is consequently linearly unstable. 

\subsubsection{Complete synchronization} Although we show that complete synchronization does not emerge for generic initial data, for completeness, we perform linear stability argument. Complete synchronization corresponds to 
\[
n_+ = N,\quad n_-=0,\quad R=1.
\]
Consider perturbations with 
\[
\bar\eta =\frac1N \sum_{i=1}^N \eta_i=0.
\]
Then, the linearization equation becomes
\[
\dot \eta_i = K(1-\alpha)\eta_i.
\]
Hence, completely synchronized equilibrium has a positive eigenvalue and is thus linearly unstable.

\subsection{Mixed polar-balanced equilibria are linearly unstable}  \label{sec:2.4}
Now, consider a mixed equilibrium where
\[
n_L>0,\quad n_++n_->0.
\]
For simplicity, we write for $i \in P_L$
\begin{align*}
& x_i ^* = \frac\alpha R p + \sqrt{1-\frac{\alpha^2}{R^2}}y_i =: ap + by_i,  \quad a:= \frac\alpha R,\quad b:= \sqrt{1-a^2}, \\
& y_i \in p^\perp,\quad \|y_i\|=1,\quad \sum_{i\in P_L} y_i=0.
\end{align*}
Recall that since $R>\alpha$, we have $b>0$. For the instability of the mixed equilibrium, we assume that
\[
\textup{dim}p^\perp = m\geq2.
\]
\subsubsection{When a polar cluster has at least two agents}  (i) Suppose first
\[
n_+\geq2.
\]
Choose perturbations supported only on $P_+$ with
\[
\xi_i\begin{cases}
 \in p^\perp,\quad i\in P_+, \\
 =0, \quad i\in P_-\cup P_L, 
 \end{cases} \quad  \sum_{i\in P_+} \xi_i=0.
\]
Then, we have
\[
\delta x_c = \frac1N \sum_{i=1}^N \xi_i = 0,\quad \delta q_i = \langle \xi_i,x_c^*\rangle + \langle x_i^*,\delta x_c\rangle =0,
\]
and in addition, for $i\in P_+$, 
\[
q_i^* = R,\quad x_c^*-q_i^*x_i^* = Rp - Rp =0.
\]
Hence, $\xi_i$ satisfies
\[
\dot \xi_i = KR(R-\alpha)\xi_i.
\]
Since $R>\alpha$, we get a positive eigenvalue and hence the mixed equilibrium with at least two positive pole agents is linearly unstable. 

(ii) Similarly, if $n_-\geq2$, then we follow a similar argument to find
\[
\dot \xi_i = KR(\alpha+R)\xi_i
\]
which also yields a positive eigenvalue. Hence, if either polar cluster has at least two agents, then the mixed equilibrium is linearly unstable. Now, it remains to consider the case where the polar part consists only of singleton (or empty) clusters. \newline

(iii) Assume now that no polar cluster contains more than one agent. Thus, each of $n_+$ and $n_-$ is either 0 or 1, but at least one is nonzero. Define the covariance operator $C:p^\perp \to p^\perp$ of the latitude directions by 
\[
C:= \sum_{i\in P_L} y_i \otimes y_i.
\]
In fact, if $v\in p^\perp$, then $Cv = \sum_{i\in P_L} \langle y_i,v\rangle y_i \in p^\perp$ with $\textup{dim}p^\perp \geq 2$. Since $\textup{tr}C = n_L$, there exists an eigenvector $z\in p^\perp$ with $\|z\|=1$ whose corresponding eigenvalue $\lambda$ satisfies
\begin{equation} \label{G-35}
\lambda<n_L.
\end{equation} 
In fact, if every eigenvalue is equal to or larger than $n_L$, thence since there are at least two eigenvalues on $p^\perp$, the trace is greater than or equal to $2n_L$ which contradicts $\textup{tr}C = n_L$. 

On the other hand for $i\in P_L$, we define 
\[
e_i := bp - ay_i,\quad a=\frac\alpha R,\quad b= \sqrt{1-a^2}.
\]
Then, we have 
\[
\langle e_i,x_i^*\rangle = \langle bp - ay_i, ap + by_i\rangle =0,\quad \Longrightarrow \quad e_i \in T_{x_i^*}\bbsm,\quad \|e_i\|=1.
\]
Now, consider the perturbations of the following form:
\[
\xi_i = \begin{cases}
Uz,\quad i\in P_+, \\
Vz,\quad i \in P_-,\\
W\langle y_i,z\rangle e_i,\quad i\in P_L.
\end{cases}
\]
Note that if $P_+$ or $P_-$ is empty, we omit the variables $U$ or $V$, respectively. Below, we calculate the linearized dynamics for $U,V$ and $W$. First, we observe
\[
\delta x_c = \frac1N \left(  n_+ Uz + n_- V z + W\sum_{i\in P_L} \langle y_i,z\rangle e_i  \right).
\]
We see
\begin{align*}
\sum_{i\in P_L} \langle y_i,z\rangle e_i & = \sum_{i\in P_L} \langle y_i,z\rangle (bp -a y_i) = bp \sum_{i\in P_L} \langle y_i,z\rangle - a \sum_{i\in P_L}   \langle y_i,z\rangle y_i \\
& = - aCz = - a \lambda z
\end{align*}
and hence we obtain 
\[
\delta x_c = Mz,\quad M:= \frac{n_+ U + n_- V - a\lambda W}{N}.
\]
For $i\in P_+$, we have $\delta q_i=0$ and hence
\[
\dot U = K(\alpha -R) (M-RU)
\]
and similarly for $i\in P_-$ we also have 
\[
\dot V = K(\alpha +R)(M+RV).
\]
On the other hand for $i\in P_L$, we observe
\begin{align*}
\delta q_i =& \langle \xi_i,x_c^*\rangle + \langle x_i^*,\delta x_c\rangle = \langle W\langle y_i,z\rangle e_i,Rp\rangle + \langle ap + by_i, Mz\rangle  \\
& =  W R\langle y_i,z\rangle \langle e_i,p\rangle + bM \langle y_i,z\rangle = b(M+RW)\langle y_i,z\rangle.
\end{align*}
In addition, 
\[
x_c^* - q_i^* x_i^* = x_c^* - \alpha x_i^* = Rp - \alpha ( ap +by_i) = bR e_i.
\]
Hence, the linearization equation becomes
\[
\dot W \langle y_i,z\rangle e_i = \dot \xi_i = K\Big[ - b(M+RW) \langle y_i,z\rangle  bRe_i\Big]= -Kb^2 R(M+RW)\langle y_i,z\rangle e_i
\]
which yields
\[
\dot W = - K b^2 R( M+RW).
\]
So far, we obtain 
\begin{equation*}
\begin{cases}
\vspace{0.3cm} \displaystyle \dot U =  K(\alpha -R) \left( \frac{n_+ U + n_-V - a\lambda W}{N} - RU \right), \\
\vspace{0.3cm} \displaystyle \dot V = K(\alpha +R)  \left( \frac{n_+ U + n_-V - a\lambda W}{N} + RV \right), \\
\vspace{0.3cm} \displaystyle \dot W  = - Kb^2 R  \left( \frac{n_+ U + n_-V - a\lambda W}{N}  +R W \right).
\end{cases}
\end{equation*}

\begin{remark}
We emphasize that $\textup{dim}p^\perp \geq2$ is assumed. Hence, this argument does not apply directly to $\bbs^1$. 
\end{remark}

\subsubsection{One positive polar agent and no negative polar agent} Consider 
\[
n_+ = 1,\quad n_-=0,\quad n_L = N-1.
\]
Then, the Jacobian matrix for $U$ and $W$ becomes
\[
K\begin{pmatrix}
\vspace{0.2cm} (\alpha -R) (\frac1N -R ) & -(\alpha -R) \frac{a\lambda}{N} \\
-\frac{b^2R}{N} & -b^2 R( R-\frac{a\lambda}{N}) 
\end{pmatrix}
\]
and its determinant is 
\[
\textup{det}  = K^2(R-\alpha) b^2R^2 \left( \frac1N - R + \frac{a\lambda}{N}\right).
\]
By recalling self-consistency equation \eqref{selfconsistency} with $\alpha = Ra$, 
\[
NR^2 -(n_+- n_-)R - n_L\alpha = NR^2 - R - n_L\alpha=0 \quad \Longrightarrow \quad an_L = NR-1.
\]
Since we choose $\lambda < n_L$, 
\[
\frac1N - R + \frac{a\lambda}{N} <\frac 1N -R + \frac{an_L}{N} = \frac{1}{N}-R + \frac{NR-1}{N}=0
\]
which shows that the determinant is strictly less than zero:
\[
\textup{det}<0.
\]
Hence, the $2\times2$ matrix has two real eigenvalues of opposite signs and hence has a positive eigenvalue. Thus, the equilibrium is linearly unstable.

\subsubsection{One negative polar agent and no positive polar agent} Similarly, we consider 
\[
n_+=0,\quad n_-=1,\quad n_L = N-1
\]
and the Jacobian matrix for $V$ and $W$ becomes
\begin{equation} \label{P-50}
K\begin{pmatrix}
\vspace{0.2cm} (\alpha + R) (\frac1N + R ) & -(\alpha +R) \frac{a\lambda}{N} \\
-\frac{b^2R}{N} & -b^2 R( R-\frac{a\lambda}{N}) 
\end{pmatrix}
\end{equation}
and
\[
\textup{det} = -K^2(\alpha +R)b^2R^2 \left( R + \frac1N - \frac{a\lambda}{N}\right).
\]
Since the self-consistency equation gives
\[
an_L = NR+1
\]
and $\lambda <n_L$ in \eqref{G-35}, we have
\[
R+ \frac1N - \frac{a\lambda}{N} > R+ \frac1N - \frac{an_L}{N} = R+ \frac1N - \frac{NR+1}{N} =0.
\]
Hence, we have
\[
\textup{det}<0.
\]
Thus, the equilibrium is linearly unstable. 

\subsubsection{One positive  and one negative polar agent} We now consider  
\[
n_+=1,\quad n_- = 1,\quad n_L = N-2.
\]
Then, the trace of the  matrix for $U,V$ and $W$ is
\begin{align*}
\frac1K \textup{tr} &= (\alpha-R) \left( \frac{n_+}{N} -R\right) + (\alpha+R) \left( \frac{n_-}{N} +R\right) - b^2R\left( R- \frac{a\lambda}{N}\right) \\
& =  (\alpha-R) \left( \frac{1}{N} -R\right) + (\alpha+R) \left( \frac{1}{N} +R\right) - b^2R\left( R- \frac{a\lambda}{N}\right) \\
& = \frac{2\alpha}{N} + 2R^2 - b^2R^2 + \frac{b^2Ra\lambda}{N}  = \alpha + (1-b^2)R^2 + \frac{b^2Ra\lambda}{N} \\
& =\alpha + a^2 R^2 + \frac{b^2Ra\lambda}{N}>0,
\end{align*}
where we used the self-consistency equation $2\alpha = N\alpha - NR^2$.  Since the trace is strictly positive, at least one eigenvalue has positive real part.  Hence, the corresponding equilibrium is linearly unstable. 

So far,  we provide an explicit basin of attraction leading to the pure balanced equilibria and exclude (1) pure bipolar equilibria and (2) mixed polar-balanced equilibria  which are all linearly unstable. Thus, only pure balanced equilibria can be regarded as a possible candidate for a generic attractor.  The aforementioned argument is summarized as follows.

\begin{theorem}
Assume
\[
N\geq 2,\quad m\geq 2, \quad \kp_1>0,\quad \kp_1+\kp_2<0,\quad \alpha = \frac{\kp_1}{-\kp_2}\in (0,1).
\]
Then, every nonzero-mean equilibrium that is not a  pure balanced equilibrium is linearly unstable. More precisely, 
\begin{enumerate}
\item Every bipolar equilibrium is linearly unstable. 
\item Every mixed polar-balanced equilibrium, i.e., every equilibrium containing at least one polar agent $x_i = \pm p$ and at least one latitude agent satisfying $\langle x_i,x_c\rangle = \alpha$, is linearly unstable. 
\end{enumerate}
Thus, the only nonzero-mean equilibria that are not ruled out by this instability result are the pure balanced equilibria:
\[
x_i^* = \sqrt\alpha p + \sqrt{1-\alpha} y_i,\quad y_i \in p^\perp,\quad \|y_i\|=1,\quad \sum_i y_i =0
\]
for which $\|x_c^*\|^2 = \alpha$. 
\end{theorem}

\begin{remark}
This result justifies the claim that in the regime $\kp_1>0$ and $\kp_1+\kp_2<0$, the non-balanced equilibria do not form generic attractors. These equilibria possess linearly unstable directions and therefore cannot be robust local attractors.  Thus, among nonzero-mean equilibria, pure balanced equilibria are the only class not excluded by the present linear instability analysis.  
\end{remark}

The stationary equilibria are classified by the numbers
\[
(n_+,n_-,n_L),
\]
where $n_+$ and $n_-$ denote the numbers of agents aligned and anti-aligned with the mean direction, and $n_L$ denotes the number of agents in the latitude class $q_i=\alpha$. The stability classification is summarized in Table \ref{tab:Sm-stability-short}. Except for the pure balanced latitude branch $(0,0,N)$, all nonzero mean stationary equilibria are linearly unstable on $S^m$ with  $m\ge2$.

 \begin{table}[ht]  
\centering
\caption{Stationary equilibria and their linear stability on $\bbs^m$, $m\ge2$ in the regime $\kp_1>0$ and $\kp_1+\kp_2<0$} \label{tab:Sm-stability-short}
\renewcommand{\arraystretch}{1.3}
\begin{tabular}{c c c}
\toprule
\textbf{Type} & \textbf{Count condition} & \textbf{Conclusion} \\
\midrule


Complete synchronization
&
$\textup{$n_+=N$} $
&
Unstable
\\

Nontrivial bipolar equilibrium
&
$n_L=0, n_+,n_->0$
&
Unstable
\\

Pure balanced latitude equilibrium
&
$n_L=N$ &
\textit{Stable}
\\

Mixed polar--latitude equilibrium
&
$n_L>0,\quad n_+ + n_->0$
&
Unstable
\\

\quad Internal polar instability
&
$n_+\ge2$ or $n_-\ge2$
&
Unstable
\\

\quad One-sided mixed equilibrium
&
$(n_+,n_-)=(1,0)$ or $(0,1)$
&
Unstable
\\

\quad Two-sided singleton mixed equilibrium
&
$(n_+,n_-)=(1,1)$
&
Unstable
\\

\bottomrule
\end{tabular}
\end{table}


%
%

\section{Convergence estimates on $\bbs^1$} \label{sec:5} 
\setcounter{equation}{0}

 In this section, we study the Kuramoto-type model \eqref{kura} on $\bbs^1$.

\subsection{Exclusion of complete synchronization and complete desynchronization}
In this subsection, we rule out the emergence of complete synchronization and complete desynchronization for generic initial data. In fact, Theorem \ref{T3.1} is also valid for $m=1$. However for the consistency of the paper, we provide the proof by directly analyzing \eqref{kura}. Define the maximal diameter
\[
D(\Theta(t)):= \max_{i,j\in [N]} |\theta_i(t) - \theta_j(t)|,\quad t>0.
\]

\begin{theorem}
Suppose that
\[
\kp_1>0, \quad \kp_1+\kp_2<0,\quad \alpha = \frac{\kp_1}{-\kp_2}\in (0,1).
\]
Then, complete synchronization and complete desynchronization cannot occur for generic initial data. Precisely, if $R_0<1$, then complete synchronization cannot emerge, and if $R_0>0$, then $R(t)$ cannot converge to zero. 

\end{theorem}

\begin{proof} 
(i) First, we show that complete synchronization cannot happen, i.e., $D(\Theta(t))$ cannot converge to zero. Suppose to the contrary that there exists $T>0$ such that 
\[
D(\Theta(t))<\frac\pi2,\quad t\geq T.
\]
For simplicity, we write $D:= D(\Theta)$ and 
\[
\theta_M := \max_{i}\theta_i,\quad \theta_m := \min_i \theta_i,\quad D = \theta_M-\theta_m,\quad \bar\theta := \frac{\theta_M+\theta_m}{2},\quad \eta_i := \theta_i - \bar\theta.
\]
Then, we have
\[
-\frac D2\leq \eta_i \leq \frac D2.
\]
It follows from straightforward calculation that 
\[
\dot D = -2\kp_1\sin\frac D2 \cdot A_1 - \kp_2\sin D \cdot A_2, 
\]
where $A_1$ and $A_2$ are defined as 
\[
A_1 := \frac1N \sum_{k=1}^N \cos\eta_k,\quad A_2 :=  \frac{1}{N^2} \sum_{j,k=1}^N \cos(\eta_j+\eta_k).
\]
In what follows, all differential equalities and inequalities involving maxima are understood to hold almost everywhere in time; the corresponding quantities are locally Lipschitz. First, for $t\geq T$, since $D(t)<\pi$,  we have $\cos\eta_k>0$ and hence we have
\[
0\leq A_1(t) \leq 1,\quad t\geq T.
\]
Second for $A_2$, since we assume $0<D<\frac\pi2$ for $t\geq T$ and $\eta_j + \eta_k \in [-D,D]$, we have
\[
A_2 \geq \cos D>0,\quad t\geq T.
\]
We use $A_1\leq 1$ and $A_2\geq \cos D$ with $\kp_2<0$ to find
\[
\dot D \geq -2\kp_1\sin \frac D2-\kp_2\sin D \cos D = \sin D \left[  -\frac{\kp_1}{\cos \frac D2} -\kp_2\cos D\right],\quad t\geq T.
\]
For $\alpha= \frac{\kp_1}{-\kp_2}\in (0,1)$, we choose small $\delta\in (0,\frac\pi2)$ satisfying
\[
\cos\delta \cos \frac\delta2>\alpha.
\]
Then, for $0<D\leq \delta$, we have
\[
\cos D\geq \cos \delta,\quad \frac{1}{\cos \frac D2} \leq \frac{1}{\cos\frac\delta2}
\]
which gives
\[
-\frac{\kp_1}{\cos \frac D2} -\kp_2\cos D \geq -\frac{\kp_1}{\cos \frac\delta2} -\kp_2\cos \delta =:\gamma_\delta.
\]
Since $\kp_2<0$ we observe
\[
\gamma_\delta :=  (-\kp_2)\cos \delta -\frac{\kp_1}{\cos \frac\delta2} >0 \quad \Longleftrightarrow \quad \cos\delta \cos\frac\delta 2 > \frac{\kp_1}{-\kp_2}=\alpha.
\]
Hence we have 
\[
\gamma_\delta>0.
\]
Therefore, for $0<D<\delta$, we have
\[
\dot D \geq \gamma_\delta \sin D \geq \frac{2\gamma_\delta}{\pi} D
\]
which shows that $D(t)$ cannot converge to zero. This is a contradiction.

(ii) We recall
\[
z= \frac1N \sum_{j=1}^N e^{\mi \theta_j} = Re^{\mi \Psi}.
\]
Then, we observe
\[
\frac1N \sum_{k=1}^N \sin(\theta_k - \theta_i) = R\sin   (\Psi - \theta_i),\quad \frac{1}{N^2} \sum_{j,k=1}^N \sin(\theta_j+\theta_k - 2\theta_i) = R^2\sin (2\Psi-2\theta_i).
\]
Hence, we write \eqref{kura} in a mean-field form:
\[
\dot\theta_i = \kp_1R\sin (\Psi-\theta_i) + \frac{\kp_2}{2}R^2\sin(2\Psi-2\theta_i) 
\]
or it can be written as
\begin{equation} \label{G-20}
\dot\theta_i =  KR\sin x_i ( R\cos x_i -\alpha), 
\end{equation}
where we defined the relative phase:
\[
x_i := \theta_i - \Psi,\quad K =-\kp_2>0,\quad \kp_1 = K\alpha.
\]
Since  $R= \frac1N \sum_{i=1}^N \cos x_i$, we have
\begin{align*}
\dot R &= -\frac1N \sum_{i=1}^N \sin x_i \dot x_i = -\frac1N \sum_{i=1}^N \sin x_i (\dot \theta_i - \dot \Psi) = - \frac1N \sum_{i=1}^N \sin x_i \dot \theta_i \\
& = KR\left[ \frac \alpha N \sum_{i=1}^N \sin^2 x_i - \frac RN \sum_{i=1}^N \sin^2 x_i \cos x_i\right].
\end{align*}
We claim that if $R(t)<\alpha$, then $R(t)$ cannot decrease. For this, since $\cos x_i\leq 1$, we have $\sin^2 x_i \cos x_i \leq \sin^2x_i$. Hence, if $0<R<\alpha$, we have
\[
\frac \alpha N \sum_{i=1}^N \sin^2 x_i - \frac RN \sum_{i=1}^N \sin^2 x_i \cos x_i \geq0.
\]
If $R(T)=0$ at some finite time $T$, then the configuration at $t=T$ is an equilibrium. Thus, uniqueness of a solution implies $R(0)=0$. Thus, $R(t)>0$ for $t\geq0$ whenever $R(0)>0$.  Now, suppose to the contrary that $R(t)$ converges to zero. Then, there exists $T>0$ such that $0<R(t)<\alpha$ for $t>T$. Then, we have
\[
R(t) \geq R(T)>0,\quad t\geq T
\]
which contradicts. 
\end{proof}

\subsection{Basin of attraction for two-cluster locked states} \label{sec:5.2} 
In this subsection, we provide a basin of attraction leading to the balanced state with finite-$N$ correction.

\subsubsection{Construction of phase-locked states}
We decompose $\{1,\cdots,N\}$ as
\[
A\sqcup B = \{1,\cdots,N\},\quad |A| = m,\quad |B|=\ell,\quad m+\ell = N
\]
and without loss of generality, we assume $m\geq \ell$ and denote
\[
 p := \frac mN,\quad q:= \frac \ell N, \quad \rho:= p-q = \frac{m-\ell}{N},\quad m,\ell\geq1.
\]
By restricting the dynamics to the two-cluster configuration, we write 
\[
\theta_i(t) = \begin{cases} 
\theta_A,\quad i\in A, \\
\theta_B,\quad i \in B.
\end{cases}
\]
For the equation of $\theta_A$, we see
\begin{align*}
\dot \theta_A & = -\kp_1 \frac\ell N \sin \Delta  + \frac{\kp_2}{2N^2} ( -2m\ell \sin \Delta - \ell^2 \sin (2\Delta)) \\
&=-\kp_1 q \sin \Delta -\kp_2 pq \sin \Delta - \frac{\kp_2}{2}  q^2 \sin (2\Delta).
\end{align*}
Similarly, the equation for $\theta_B$ becomes
\[
\dot \theta_B = \kp_1p \sin \Delta + \kp_2pq \sin \Delta + \frac{\kp_2}{2} p^2 \sin (2\Delta).
\]
Then, the difference $\Delta=\theta_A - \theta_B$  satisfies
\begin{equation} \label{F-10}
\dot \Delta = - \sin \Delta ( \kp_1+2\kp_2pq + \kp_2(p^2 + q^2)\cos\Delta).
\end{equation}
Since we want nontrivial two-cluster locked profile, $\Delta=0$ (complete synchronization) and $\Delta=\pi$ (antipodal state) are excluded. Hence, a difference for locked phase $\Delta_*\in (0,\pi)$ is determined by 
\[
\cos \Delta_* =  \frac{\alpha - 2pq}{p^2+q^2}  = \frac{2\alpha -1 + \rho^2}{1+\rho^2},\quad p^2 + q^2 = \frac{1+\rho^2}{2},\quad 2pq = \frac{1-\rho^2}{2}.
\]
Define the target profile $e:=(e_1,\cdots,e_N)$ as $e_i := a_*$ for $i\in A$ and $e_i := b_*$ for $i\in B$. Here, $a_*$ and $b_*$ should satisfy $a_* - b_* = \Delta_*$. In addition, by using rotational invariance, we fix the phase gauge by requiring that the average of $e_i$ is zero. Thus, $e_i$ is introduced to satisfy 
\[
e_i = \begin{cases}
a_*  = q\Delta_*,\quad i\in A, \\
b_* = -p \Delta_*,\quad i\in B,
\end{cases}\quad a_* - b_* = \Delta_*,\quad \sum_{i=1}^N e_i =0.
\]
Then, the squared order parameter satisfies 
 \begin{equation} \label{F-12}
R_*^2  = \left| \frac1N \sum_{k=1}^N e^{\mi e_k}    \right|^2 = |pe^{\mi a_*} + qe^{\mi b_*}|^2 =  p^2 + q^2 + 2pq \cos \Delta_* = \alpha + \frac{2(1-\alpha)\rho^2}{1+\rho^2}.
\end{equation}
Or sometimes, $R_*^2$ can be written as $R_{N,k}^2$ to emphasize the imbalance:
\[
R_*^2 = \alpha + \frac{2(1-\alpha)\rho^2} {1+\rho^2} = \alpha + \frac{2(1-\alpha)k^2}{N^2+k^2}=:R_{N,k}^2,\quad \rho = \frac{m-\ell}{N} = \frac kN.
\]
Note that $e$ would not be  an equilibrium; in general, it is a rotating wave profile. In other words, there exists $\Omega_* \in \bbr$ such that
\[
F(e) = \Omega_* \mathbf{1}.
\]
The profile is stationary when $\rho=0$, whereas $\Omega_*\neq0$ for $\rho>0$. 
\subsubsection{Stability parameters and local basin condition} 
We introduce four positive constants:
\begin{align} \label{fourconstants}
\begin{aligned}
\Lambda_A &:= K(1- \alpha ) \frac{ \alpha (1-3\rho + \rho^2+\rho^3) -4\rho^3}{(1+\rho^2)^2}, \\
\Lambda_B &:= K(1- \alpha ) \frac{  \alpha (1 + 3\rho + \rho^2 -\rho^3) + 4\rho^3}{(1+\rho^2)^2}, \\
\Lambda_\Delta &:= 2K(1- \alpha )\frac{ \alpha +\rho^2}{1+\rho^2}, \\
\Lambda_0 & := \min\{ \Lambda_A,\Lambda_B,\Lambda_\Delta\}.
\end{aligned}
\end{align}
Here, $\Lambda_B$ is obviously positive, due to $1-\rho^3\geq0$. Now, to guarantee the positiveness of $\Lambda_A$, we need to assume
\[
 \alpha (1-3\rho + \rho^2 + \rho^3)>4\rho^3
 \]
 which can be achieved by small $k$ (see Section \ref{sec:smallk}). Choose $\delta \in (0,\frac\pi8)$ to satisfy 
\[
 \lambda_\delta: = \Lambda_0 - 12\kp_1\sin(2\delta)-12K \sin (4\delta)>0. 
\]
Since $\Lambda_0>0$ and $\lambda_\delta \to \Lambda_0$ as $\delta\to0$, such a $\delta>0$ exists. 
  Under this assumption, we choose initial data as a slight perturbation of the two-cluster states with a common phase.

\begin{theorem} \label{T5.2}
Suppose that
\[
\kp_1>0,\quad \kp_1+\kp_2<0,\quad \alpha = \frac{\kp_1}{-\kp_2} \in (0,1),\quad \alpha (1-3\rho +\rho^2 + \rho^3)>4\rho^3.
\]
Then, for $\veps \in (0,\frac\delta2)$, if the initial data satisfy
\[
\theta_i^0 = \begin{cases}
\psi + a_* + \xi_i,\quad i\in A, \\
\psi+ b_* + \eta_i,\quad i\in B,
\end{cases},\quad \max \left\{  \max_{i\in A} |\xi_i|,\max_{j\in B} |\eta_j|     \right\} <\veps,
\]
then we have
\[
\max_{1\leq i\leq N} \left| \theta_i(t) -\frac1N\sum_{j=1}^N\theta_j(t) -e_i \right| \leq 4\veps e^{-\lambda_\delta t}.
\]
Consequently, we obtain
\[
\lim_{t\to\infty} ( \theta_i(t)-\theta_j(t)) = \begin{cases} 
0,&i,j\in A\ \textup{or}\ i,j\in B,\\
\Delta_*,&i\in A,\ j\in B
\end{cases}
\]
and 
\[
\lim_{t\to\infty} R(t) ^2 = \alpha + \frac{2(1-\alpha)k^2}{N^2+k^2} = R_{N,k}^2.
\]
In addition, we have 
\[
\lim_{t\to\infty} \dot\theta_i(t) = \Omega_*,\quad i\in [N].
\]
\end{theorem}

\begin{proof}
Since the proof consists of several steps, it is provided in Appendix \ref{sec:app.B}.

 \end{proof}

\section{Stability: Why two-cluster locked state is generic?} \label{sec:6}
\setcounter{equation}{0}

Contrary to the high-dimensional case, an essential feature of the dynamics on $\bbs^1$ is the distinction between stationary equilibria and rotating phase-locked states.  The circle differs qualitatively from the genuinely high-dimensional spheres, because a finite population imbalance cannot be absorbed by a stationary balanced configuration. Thus, the residual imbalance is compensated by a collective rotation rather than by a stationary balance. Consequently, stationary balance imposes a substantially stronger constraint on $\bbs^1$. 

In this section, we classify these two cases and investigate their linear stability properties.

\subsection{Stationary equilibrium} \label{sec:6.1} 
First, we recall from \eqref{G-20} that
\[
\dot \theta_i  = KR\sin \phi_i (R\cos \phi_i-\alpha) =: f_R(\phi_i),\quad f_R(u) = -KR\sin u (\alpha -R\cos u),\quad \phi_i =\theta_i - \Psi.
\]
Then, the stationary equilibrium satisfies $f_R(u)=0$ for all $i\in [N]$. Note that $R=0$ is a stationary equilibrium. Below, we consider $R>0$. 

\subsubsection{Self-consistency equation} For each $i\in [N]$, we have
\[
\sin \phi_ i= 0 \quad \textup{or}\quad \cos \phi_i = \frac{\alpha}{R}
\]
which corresponds to 
\[
\phi_i =0,\quad \pi,\quad \phi_L, \quad -\phi_L,\,\, \mathrm{where}\,\, \cos\phi_L := \frac\alpha R.
\]
For $\pm\phi_L$ to exist, one must have $R\geq\alpha$.  Define
\begin{align*}
&n_+:= |\{ i: \phi_i=0\}|,\quad n_- := |\{ i: \phi_i = \pi\}|, \\
&  n_{L,+} := |\{ i: \phi_i =\phi_L\}|, \quad  n_{L,-} := |\{ i: \phi_i = -\phi_L\}|, \\
 & n_L:=|\{ i: \phi_i = \pm \phi_L\}|,\quad  n_L = n_{L,+}+n_{L,-}.
\end{align*}
Here, self-consistency equation becomes
\begin{equation} \label{D-60}
\frac1N \sum_{i=1}^N e^{\mi \phi_i}=R.
\end{equation}
Then, the imaginary part of \eqref{D-60} becomes
\[
\frac1N (n_{L,+}-n_{L,-})\sin \phi_L=0
\]
which gives for non-trivial case $\sin\phi_L\neq0$,
\[
n_{L,+}=n_{L,-}=:s.
\]
Hence, $n_L=2s$ should be even. On the other hand, the real part of \eqref{D-60} gives
\[
R = \frac1N (n_+-n_- + 2s\cos \phi_L) = \frac1N \left( n_+ - n_- + \frac{2s\alpha}{R}\right)
\]
which is rewritten as
\[
NR^2 - (n_+ - n_-)R - 2s\alpha =0.
\]
Since the Kuramoto model with higher-order interactions coincides with the high-dimensional model for $m=1$, we directly use the stability result for some cases. 

\subsubsection{Pure bipolar equilibria}
In this case, we have
\[
n_L = 0.
\]
Hence, for all $i\in [N]$, all oscillators   lie at one of the two antipodal phases $0$ and $\pi$:
\[
\phi_i = 0, \, \pi.
\]
Then, we use the previous instability results for the high-dimensional case to see that all pure bipolar equilibria are linearly unstable. 

\subsubsection{Pure balanced equilibria} In this case, we have
\[
n_+=n_-=0,\quad n_L=N.
\]
Hence, $N=2s$ should be even. In this case, self-consistency equation gives
\[
R^2=\alpha.
\]

\subsubsection{Mixed polar-latitude equilibria}

What is different between $\bbs^1$ and $\bbs^m$ with $m\geq2$ is that in $\bbs^1$, we have 
\[
\textup{dim} p^\perp =1.
\]
In other words, there are only two unit vectors $\pm \nu$ in $p^\perp$: $y_i = \pm \nu$. Thus, when we consider $P_L$, we only have
\[
x_i = \frac\alpha R p \pm \sqrt{ 1- \frac{\alpha^2}{R^2}}\nu,\quad \sum_{i\in P_L} y_i =0.
\]
Thus, $n_L$ must be even on $\bbs^1$. 



Let us consider mixed equilibrium case in Section \ref{sec:2.4}. 

(i) ($n_+\geq 2$ or $n_-\geq2$):  Suppose that $n_+\geq 2$. Then, we similarly choose a perturbation supported only on $P_+$ with zero sum:
\[
\xi_i = \begin{cases}
\eta_i \nu, \quad i\in P_+, \\
0,\quad \textup{otherwise}, 
\end{cases}\quad \sum_{i\in P_+} \eta_i=0.
\]
Then, by a similar argument, we see that such mixed equilibria are linearly unstable on $\bbs^1$. In addition, exactly the same result holds for $n_-\geq2$. 

(ii) It remains to consider the two-sided singleton case 
\[
n_+=1,\quad n_-=1.
\]
In $\bbs^1$, the operator 
\[
C = \sum_{i\in P_L} y_i \otimes y_i
\]
on $p^\perp$ has the single eigenvalue $\lambda = n_L$. Then, the trace becomes $\alpha + \alpha^2 + \frac{\alpha b^2 n_L}{N}>0$. Hence, it has at least one eigenvalue with positive real part. Hence, the equilibrium is linearly unstable.

\begin{remark}
Then, what fails on $\bbs^1$? In fact, the singleton one-sided polar case is unstable. Previously on $\bbs^m$ with $m\geq2$, $C$ has an eigenvalue $\lambda<n_L$. This is true when $\textup{dim}p^\perp \geq 2$. However on $\bbs^1$, since $\textup{dim}p^\perp = 1$ and $y_i = \pm \nu$, we have
\[
C = n_L \nu\otimes \nu.
\]
Thus, the only eigenvalue is $\lambda = n_L$.  
\end{remark}

(iii) Note that the case
\[
n_+ = 0,\quad n_-=1,\quad n_L = N-1
\]
remains linearly unstable on $\bbs^1$, since it has a positive eigenvalue. Precisely, since $\lambda = n_L$, matrix in \eqref{P-50} has determinant 0 and a positive trace. Hence, the equilibrium has  a positive eigenvalue.

(iv) The exceptional case is
\[
n_+=1,\quad n_-=0,\quad n_L = N-1.
\]
The scalar $R$ satisfies
\[
NR^2 - R - (N-1)\alpha=0.
\]
The matrix has determinant 0 and its nonzero eigenvalue is 
\[
\mu_+ = K\frac{R-\alpha}{RN}( (N-2)\alpha -R).
\]
Hence, if $R<(N-2)\alpha$, then the equilibrium is linearly unstable. If $R\geq (N-2)\alpha$, then the corresponding block has no positive eigenvalue, and the equilibrium becomes non-hyperbolic with a zero eigenvalue and non-positive eigenvalues.  To remove $R$ from the condition $R\geq (N-2)\alpha$, we consider
\[
f(u) = Nu^2 - u-(N-1)\alpha
\]
whose positive root is $R$. Hence,
\[
R<(N-2)\alpha \quad \Longleftrightarrow \quad f((N-2)\alpha)>0.
\]
Hence, instability condition becomes
\[
\alpha > \frac{2N-3}{N(N-2)^2}=:\alpha_c(N).
\]
For $N=3$, one has $\alpha_c(3)=1$. Since $0<\alpha<1$, the instability criterion $\alpha>\alpha_c(3)$ cannot be satisfied. Hence, this mode does not yield linear instability.   For $N=5$, since $\alpha_c(5) = \frac{7}{45}$, if $\alpha>\frac{7}{45}$, then the state is linearly unstable.

Unlike the case $\bbs^m$, $m\ge2$, the circle imposes an additional parity constraint on the latitude class. Indeed, since the latitude phases are only $\pm\phi_L$, the imaginary part of the self-consistency condition requires the two latitude populations to be equal. Hence $n_L$ must be even. The resulting stability classification of stationary equilibria on $\bbs^1$ is summarized in Table~\ref{tab:S1-equilibria}. Apart from the pure balanced   equilibrium and the parameter-dependent exceptional family $(n_+,n_-)=(1,0)$, all stationary equilibria are linearly unstable.

\begin{table}[ht]
\centering
\caption{Stationary equilibria and their linear stability on $\bbs^1$ in the regime $\kp_1>0$ and $\kp_1+\kp_2<0$.}
\label{tab:S1-equilibria}
\renewcommand{\arraystretch}{1.3}
\begin{tabular}{c c c}
\toprule
\textbf{Equilibrium class} & \textbf{Existence condition} & \textbf{Conclusion} \\
\midrule


Complete synchronization
&
$\textup{$n_+=N$ } $
&
Unstable
\\

Pure bipolar equilibria
&
$n_L=0$
&
Unstable
\\

Pure balanced latitude
&
$n_L=N$, $N$ even
&
\textit{Stable}
\\

\quad Internal polar case
&
$n_L>0$ and $\textup{$n_+\geq2$ or $n_-\geq2$}$
&
Unstable
\\

\quad Two-sided singleton case
&
$n_L>0$, $(n_+,n_-)=(1,1)$
&
Unstable
\\

\quad Negative singleton case
&
$n_L>0$,  $(n_+,n_-)=(0,1)$
&
Unstable
\\

\quad Positive singleton case
&
$(n_+,n_-)=(1,0)$, $N$ odd
&
\makecell{Unstable if $\alpha>\alpha_c(N)$}
\\

\bottomrule
\end{tabular}
\end{table}

The fact that stable stationary equilibria occur only in restricted cases is one of the reasons why generic coherent states on $\bbs^1$ are more appropriately described as phase-locked equilibria rather than stationary equilibria. In particular, finite population imbalance is compensated by a common angular velocity.

%
%
%

\subsection{Phase-locked states} \label{sec:6.2} 
 
So far, we have shown that except for the pure balanced equilibrium  $(n_+,n_-,n_L) = (0,0,N)$ which exists only for even $N$, all stationary equilibria are linearly unstable, with one exceptional mixed family $(n_+,n_-,n_L) = (1,0,N-1)$ which exists only for odd $N$ and can be linearly stable for $\alpha <\alpha_c(N)$. Hence, we would say that the stationary equilibrium is restrictive; in other words, requiring a configuration to be a stationary equilibrium imposes a much stronger constraint. 

In contrast, a phase-locked state can be understood as a relative equilibrium that relaxes this rigidity by allowing an additional degree of freedom, namely the common angular velocity, say, $\Omega$. 

This additional freedom is particularly crucial when the stationary balance condition cannot be satisfied, for instance, due to a finite population imbalance. In such cases, the imbalance is not absorbed by a static configuration but is instead compensated by a collective rotation.  

To this end, the locked states satisfy
\[
\dot \theta_i = \Omega,\quad i\in [N].
\]
Hence, $\dot \Psi = \Omega$ and $\phi_i = \theta_i - \Psi$ converge to definite values. In this case, we have
\[
KR\sin \phi_i (R\cos \phi_i - \alpha) = \Omega. 
\]
Then, the self consistency equation \eqref{D-60} gives
\begin{equation} \label{D-78}
\sum_{i=1}^N \sin \phi_i = 0,\quad \frac1N \sum_{i=1}^N \cos \phi_i = R.
\end{equation}

 \begin{figure}[H]
\centering
\mbox{
\subfigure{  
\includegraphics[width=1\textwidth]{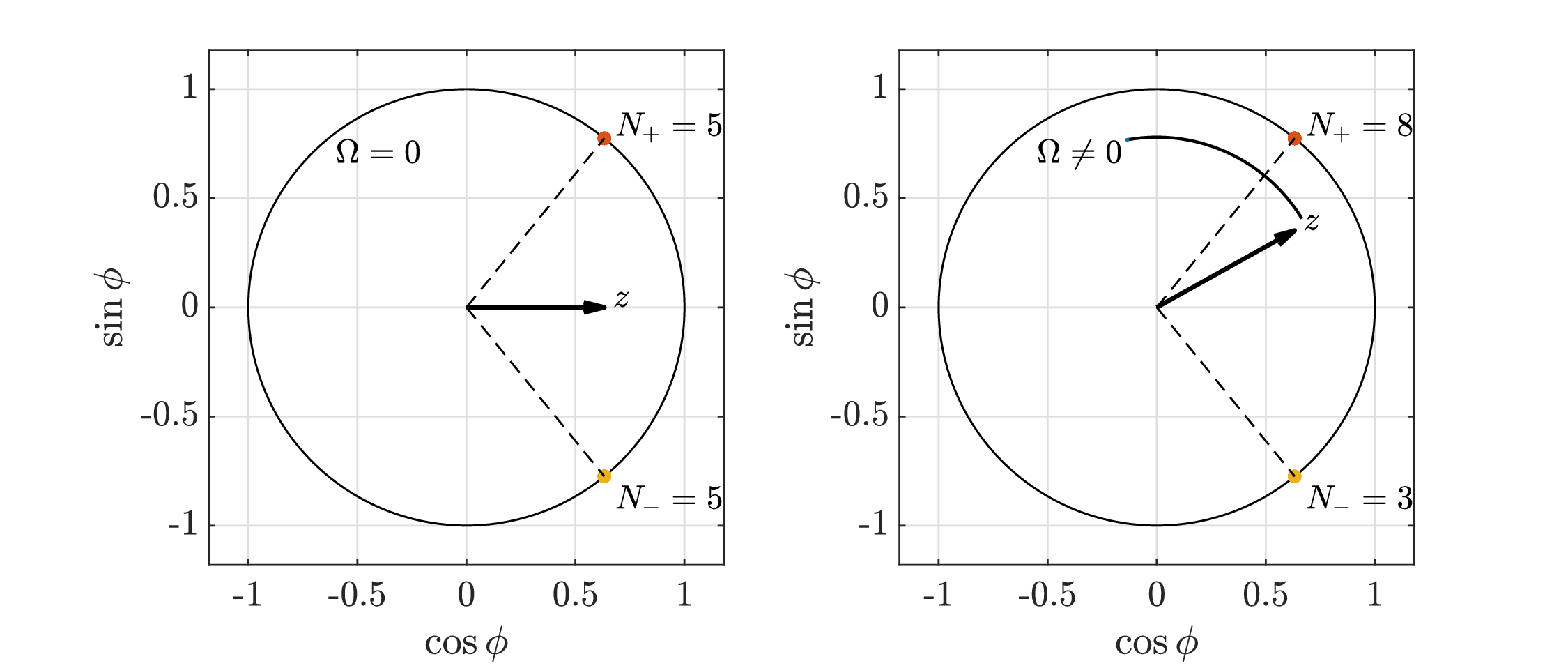}\label{fig:3-1}}
}
\caption{  Parity obstruction on $\bbs^1$: (Left) An equal occupation of the two balanced roots is possible for even $N$ (here, $N=10$) which produces a stationary state with $R^2=\alpha$. (Right) For odd $N$ (here, $N=11$), the population imbalance is converted into a common rotation.  }  \label{fig5-1}
\end{figure}

\begin{remark}
Phase-locked states and their stability have been studied for the Kuramoto-type model, for instance, \cite{HR,LH}. 
\end{remark}

\subsection{Why do two-cluster $R_{N,k}$-branches dominate the dynamics?} \label{sec:6.3} 
In Theorem \ref{T5.2}, we provide a sufficient initial condition under which the order parameter $R(t)^2$ converges to  $R_{N,k}^2$ which can be realized as a two-cluster phase-locked state. In this subsection, we justify why $R_{N,k}^2$ can be observed in most numerical simulations.

Note that the phase-locked state satisfies
\begin{equation} \label{E-36}
f_R(\phi) = -KR(\alpha - R\cos\phi) \sin \phi = \Omega
\end{equation}
for some common angular velocity $\Omega$. Now, we show that there are four roots for equation \eqref{E-36}. Consider the equation $f_R(\phi_i)=\Omega$ with $\eqref{D-78}_1$:
\[
-K R (\alpha - R\cos \phi_i)\sin \phi_i = \Omega,\quad \sum_{i=1}^N \sin \phi_i=0.
\]
Since $\Omega\neq0$,  no locking root can satisfy $\sin\phi_i=0$, because $f_R(\phi_i)=0$ whenever $\sin \phi=0$.  Thus for each $i\in [N]$, we have either $\sin\phi_i<0$ or $\sin \phi_i>0$. Without loss of generality, assume that $\Omega>0$. Suppose that $\sin \phi_i>0$. Then, we should have $R\cos\phi_i>\alpha$ and $R>\alpha$. Define 
\[
\phi_L := \arccos \frac{\alpha}{R}\in \left(0,\frac\pi2\right).
\]
We see that  $\phi_i \in I_+:= (0,\phi_L)$. Similarly, for the case of $\sin \phi_i<0$, we should have $R\cos\phi_i<\alpha$ and $\phi_i \in I_-:= (-\pi,-\phi_L)$. Hence, there exists at least one zero in each of the intervals $I_+$ and $I_-$. By recalling 
\[
g(\phi) = f_R(\phi)-\Omega, 
\]
we see that at the boundary of $I_+ = (0,\phi_L)$, 
\[
g(0) = f_R(0)-\Omega = -\Omega<0,\quad g(\phi_L) = -\Omega<0.
\]
Denote $s_1$ as the leftmost root in $I_+$. Generically, $s_1$ is a simple root, i.e., $g'(s_1)\neq0$ and in fact, $g'(s_1)>0$. Due to the continuity of a solution, we see that there exists $s_2\in (s_1,\phi_L)$ such that $g(s_2)=0$. Thus, in $I_+$, there are at least two zeros. By   exactly the same argument, there exist at least two zeros in $I_-$. Thus, the self-consistency condition yields at least two roots in each of $I_+$ and $I_-$. Since $f_R-\Omega$ is a nonzero trigonometric polynomial of degree two, it has at most four roots in one period. Hence, there are exactly four simple locking roots.

Note that if a root is simple, then the stability of the root is determined by the sign of its derivative. In addition, when crossing the root, the sign of $g$ changes and stability also changes. For instance, if $g$ changes from $+$ to $-$, then the root is stable. For the next root, $g$ changes from $-$ to $+$ and hence the next root is unstable. Thus,  if there are four roots, then there are two stable roots and two unstable roots.

Consequently, generic and robust locked attractors are expected to avoid unstable roots and is selected by the stable roots: unstable roots cannot persist under perturbations, whereas stable roots can be robustly occupied. Since there are two stable roots, the resulting robust locked attractors are naturally two-cluster states. This provides a mechanism explaining why the $R_{N,k}$ branches appear so prominently.

 \begin{figure}[h]
\centering
\mbox{
\subfigure{  
\includegraphics[width=1\textwidth]{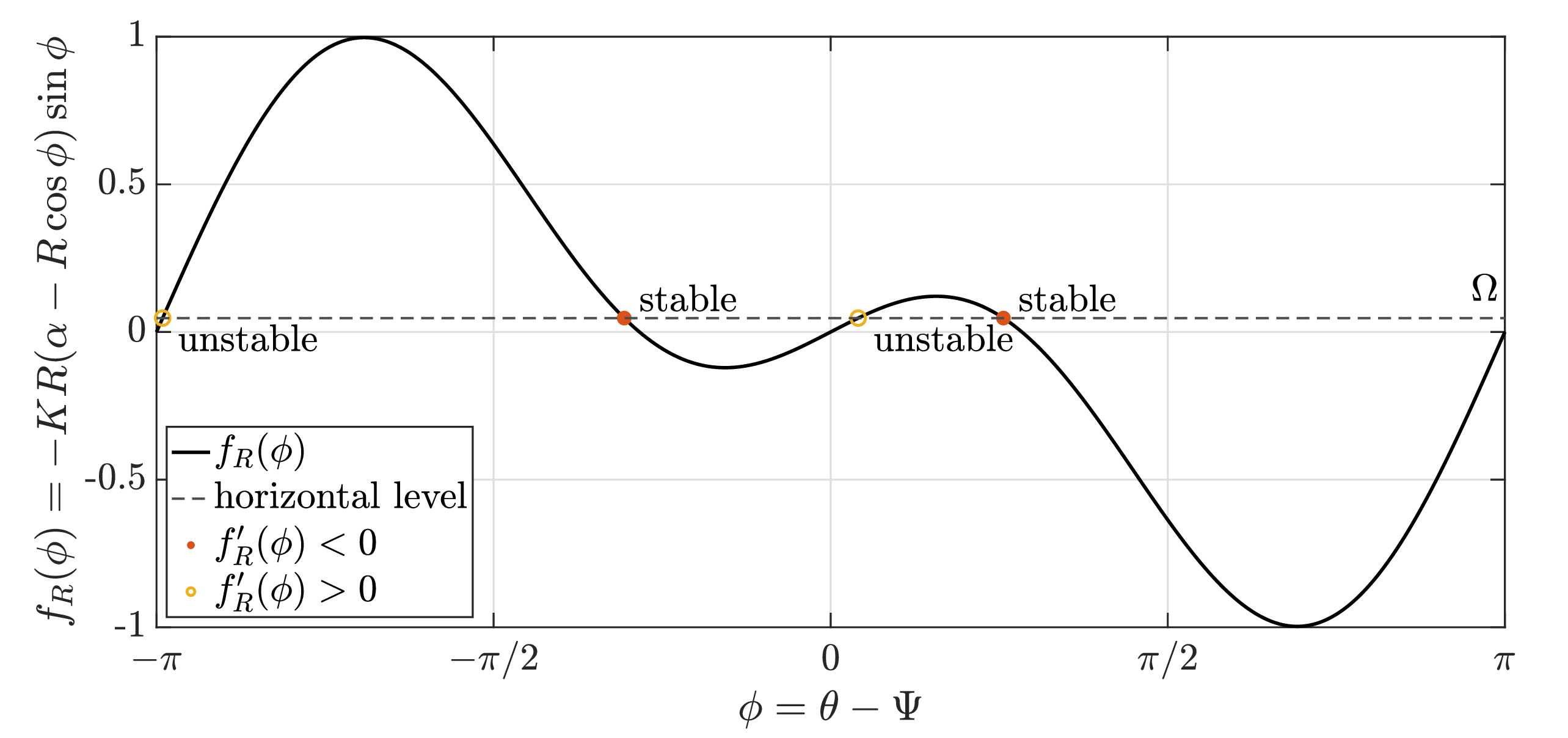}}
}
\caption{ Root selection for phase-locked states: The intersections of $f_R(\phi)$ with the horizontal level $\Omega$ determine the admissible phase-locked states. Filled and unfilled markers indicate roots with negative and positive slopes, respectively, which correspond to internally stable and internally unstable clusters. }  \label{fig6-2}
\end{figure}


\subsection{Why are the $R_{N,k}$ branches not equally likely?} \label{sec:smallk}

Recall
\[
R_{N,k}^2 =\alpha + \frac{  2 (1-\alpha)k^2}{N^2+k^2},\quad  0\leq k\leq N-2,\quad k\equiv N \textup{ (mod } 2).
\]
Here, $k$ and $N$ have same parity. Algebraically, all $k$ are possible; however, they are not equal. Our goal of this subsection is to show that small $k$ is stable, whereas large $k$ is unstable. 

\subsubsection{Internal stability filter}
Write two locking roots $\phi_+$ and $\phi_-$ with their cluster sizes $N_+$ and $N_-$, respectively, so that $N_+ + N_-=N$. Write also $\eta:= \frac{N_+-N_-}{N}$:
\[
P_+:= \{ i : \phi_i = \phi_+\}, \quad P_-:= \{ i : \phi_i = \phi_-\},\quad N_+:= |P_+|,\quad N_-:= |P_-|.
\]
Suppose that several oscillators occupy the same cluster root $\phi_\pm$:
\[
\phi_i = \begin{cases}
\phi_+,\quad i\in P_+, \\
\phi_-,\quad i\in P_- .
\end{cases}
\] 
To establish the internal stability of the $P_+$-cluster, we restrict our attention to zero-sum perturbations
\[
\phi_i \to \phi_i+ \veps_i,\quad i\in P_+,\quad \sum_{i\in P_+}\veps_i=0.
\]
This condition removes the coherent displacement of the cluster and isolates perturbations that only affect its internal phase. Moreover, the contribution of such perturbations to the complex-valued order parameter vanishes at first order. Hence, both the mean-field amplitude and the mean phase remain unchanged to linear order. Precisely, we perturb the oscillators in $P_+$:
\[
\theta_i(t) = \theta_i^*(t) + \veps_i(t),\quad \dot \theta_i^* = \Omega,\quad \theta_i^* = \theta_+^*.
\]
Then, the contribution of the perturbation to the complex-valued order parameter is calculated as
\begin{align*}
\frac{1}{N}\sum_{i\in P_+} e^{\mi (\theta_+^* + \veps_i)} = \frac{e^{\mi\theta_+^*}}{N}\left(  N_+ + \mi \sum_{i\in P_+} \veps_i    \right) + O(\|\veps\|^2) = \frac{N_+}{N}e^{\mi \theta_+^*} + O(\|\veps\|^2).
\end{align*}
This shows that the order parameter is preserved at the first-order level. Next, we denote
\[
R^\veps := R+ \delta R ,\quad \psi^\veps := \psi + \delta \psi.
\]
Then, the order parameter becomes
\[
Z^\veps = (R+\delta R) e^{\mi (\psi + \delta\psi)} = Re^{\mi \psi} + e^{\mi\psi}(\delta R + \mi R \delta \psi) + \textup{higher order terms}
\]
which gives the first variation of the order parameter
\[
\delta Z = e^{\mi \psi}(\delta R+ \mi R\delta \psi) .
\]
Since   $\delta Z =0$ and $R>0$,  we have
\[
\delta R = 0,\quad \delta \psi=0.
\]
Note that
\[
\phi_i^\veps = \theta_i^\veps - \psi^\veps = (\theta_i^* - \psi^*) + \veps _i - \delta \psi = \phi_+ + \veps_i - \delta \psi.
\]
Now, we recall the governing equation 
\[
\dot \theta_i = f_{R^\veps} (\phi_i^\veps) = f_{R+\delta R} (\phi_+ + \veps_i-\delta \psi).
\]
In fact, since $f$ is a function of two variables, we introduce $F=F(R,\phi)$ 
\[
F(R,\phi):= -K R (\alpha - R\cos\phi) \sin \phi.
\]
By the Taylor expansion at $(R,\phi_+)$, we observe
\begin{align*}
F(R+\delta R ,\phi_++\veps_i - \delta \psi) & = F(R,\phi_+) + \partial_R F(R,\phi_+) \delta R + \partial_\phi F(R,\phi_+) (\veps_i - \delta \psi) + O(\|\veps\|^2) \\
& = f_R(\phi_+) + \partial_R f_R(\phi_+) \delta R + f_R'(\phi_+) (\veps _i - \delta \psi) + O(\|\veps\|^2) \\
& = f_R(\phi_+) + f_R'(\phi_+) \veps_i + O(\|\veps\|^2).
\end{align*}
Thus, in the first-order expansion of the perturbed equation, we derive
\[
\Omega+ \dot \veps_i = \dot \theta_i  =  f_R(\phi_+) + f_R'(\phi_+) \veps_i = \Omega + f_R'(\phi_+) \veps_i
\]
which gives
\[
\dot \veps_i = f'_R(\phi_+)\veps_i.
\]
Hence, the internal eigenvalue is
\[
\lambda_+ = f_R'(\phi_+).
\]
Here, we only have two clusters $\phi_+$ and $\phi_-$. Then $\sum \sin \phi_j=0$ gives
\[
\frac{N_+}{N} \sin \phi_+ + \frac{N_-}{N} \sin \phi_-=0.
\]
Hence, the signs of $\sin \phi_+$ and $\sin \phi_-$ are different. After choosing representatives in $(-\pi,\pi]$, the signs of $\phi_+$ and $\phi_-$ are different. By a similar argument, the internal eigenvalue for $P_-$ satisfies
\[
\lambda_- = f_R'(\phi_-).
\]
Now, we calculate $f_R'(\phi)$. Recall 
\[
f_R(\phi) = -KR(\alpha - R\cos \phi)\sin \phi
\]
and hence
\[
f_R'(\phi) = -KR(R\sin^2\phi + (\alpha - R\cos\phi)\cos\phi). 
\]
We calculate 
\[
q_+ := R\cos \phi_+,\quad q_- =: R\cos \phi_-.
\]
For simplicity, we write 
\[
\phi_+ = :a>0,\quad \phi_- =: -b<0.
\]
 Then, by definition, 
\[
Re^{\mi \psi} = \frac1N \sum_{j=1}^N e^{\mi \theta_j} = e^{\mi \psi} (w_+ e^{\mi a} + w_- e^{-\mi b}),\quad w_+:= \frac{N_+}{N},\quad w_-:= \frac{N_-}{N}
\]
which gives
\[
R = w_+ e^{\mi a} + w_-e^{-\mi b}.
\]
Hence, we find
\[
w_+ \sin a = w_- \sin b,\quad R = w_+ \cos a + w_- \cos b.
\]
Multiplying $e^{-\mi a}$ and taking real parts give
\[
q_+ = w_+ + w_-\cos(a+b).
\]
Similarly, we have
\[
q_- = w_- + w_+ \cos(a+b).
\]
On the other hand, squaring gives
\[
R^2 = |w_+ e^{\mi a} + w_- e^{-\mi b}|^2 = w_+ ^2 + w_-^2 + 2w_+ w_- \cos(a+b) = \frac{1+\eta^2}{2} + \frac{1-\eta^2}{2} \cos(a+b).
\]
Hence, we have
\[
\cos(a+b) = \frac{2R^2- 1-\eta^2}{1-\eta^2}
\]
which yields
\begin{align*}
q_+ &= w_+ + w_- \cos(a+b)  = \frac{1+\eta}{2} + \frac{1-\eta}{2} \frac{2R^2-1-\eta^2}{1-\eta^2} \\
&= \frac{1+\eta}{2} 
+ \frac{2R^2-1-\eta^2}{2(1+\eta)} 
 = \frac{ (1+\eta)^2 + 2R^2-1-\eta^2}{2(1+\eta)}  = \frac{R^2+\eta}{1+\eta}.
\end{align*}
Similarly, we see
\[
q_- = \frac{R^2-\eta}{1-\eta}.
\]
Now, we recall the locked condition
\[
f_R(a) = f_R(-b),\quad \textup{or, equivalently,} \quad -KR(\alpha-q_+) \sin a  = KR(\alpha - q_-)\sin b
\]
which is
\[
(q_+-\alpha) \sin a = (\alpha - q_-)\sin b.
\]
It follows from the self-consistency that 
\[
(1+\eta) \sin a = (1-\eta) \sin b.
\]
Hence, we have
\[
R^2 = \alpha + \frac{2(1-\alpha)\eta^2}{1+\eta^2},\quad \eta = \frac kN.
\]
Then, we observe
\[
f_R'(\phi) = -KR(R\sin^2\phi  + (\alpha - R\cos \phi)\cos\phi) = -K(R^2 + \alpha q-2q^2).
\]
Hence, the corresponding projections are
\begin{align*}
q_+ &= \frac{\alpha + (1-\alpha)\eta + \eta^2}{1+\eta^2} = \frac{R^2+\eta}{1+\eta}, \\
q_- & = \frac{\alpha - (1-\alpha)\eta + \eta^2}{1+\eta^2}.
\end{align*}
Then, by straightforward calculation, we have 

\begin{align*}
\lambda_+ &= -K(R^2 + \alpha q_+ - 2q_+^2)  =-K (1-\alpha)\frac{P_\alpha(\eta)}{(1+\eta^2)^2}, 
\end{align*}
where an auxiliary polynomial is defined as
\[
P_\alpha(\eta) := \alpha(1-3\eta + \eta^2 + \eta^3)-4\eta^3.
\]
Hence, the majority cluster is internally stable only if $P_\alpha(\eta)>0$. On the other hand, if $P_\alpha(\eta)<0$, then $\lambda_+>0$ and the branch is linearly unstable. 

Similarly for $\lambda_-$, we have
\[
\lambda_- = -K(1-\alpha) \frac{   \alpha(1+3\eta + \eta^2 - \eta^3) + 4\eta^3   }{(1+\eta^2)^2}.
\]
Since $0\leq \eta<1$ and $0<\alpha<1$, we have $1+3\eta + \eta^2 - \eta^3>0$ and hence $\lambda_-<0$. Hence, the minority cluster is always internally stable. The majority cluster is the one that may lose stability. 

For $P_\alpha(\eta)$, we observe
\[
P_\alpha(0)= \alpha>0,\quad P_\alpha(1)  =-4<0.
\]
In addition, $P_\alpha(\cdot)$ is strictly decreasing on $[0,1]$ for $0<\alpha<1$. Hence, there exists a unique threshold $\eta_c(\alpha)\in (0,1)$ such that $P_\alpha(\eta_c)=0$. Hence, if $\eta= \frac kN  <\eta_c$, then both clusters are internally stable and if $\eta=\frac kN>\eta_c$, then two-cluster branch is linearly unstable. This explains why the $R_{N,k}$ branches with large imbalance (or large $k$) are not equally likely; many of them are indeed internally unstable. In other words, large imbalance branches lose internal stability because the majority cluster eigenvalue becomes positive for large imbalance. 

Note also that we have
\[
\frac{dR^2}{d(\eta^2)} = \frac{2(1-\alpha)}{(1+\eta^2)^2}>0.
\]
Hence, the smaller $\eta$ is, the closer $R^2$ is to the balanced value $\alpha$. In addition, the internal stability condition is easy to satisfy for small $\eta$. In fact, as $\eta\to0$, $P_\alpha(\eta)\to\alpha>0$. Thus, small-imbalance branches have negative internal eigenvalues and are dynamically robust.

  \begin{figure}[H]
\centering
\mbox{
\subfigure{  
\includegraphics[width=1\textwidth]{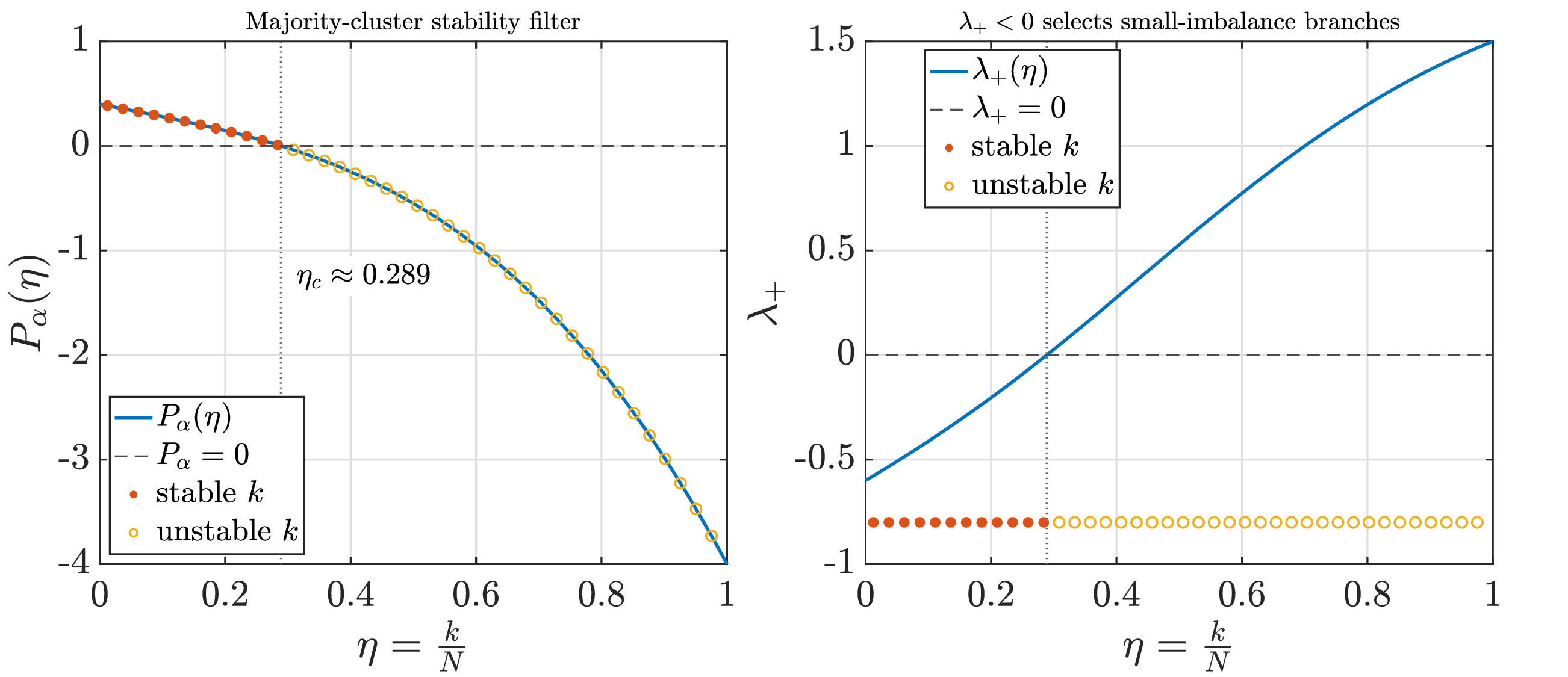}\label{fig:3-1}}
}
\caption{ Stability figures: The algebraic branch $R_{N,k}$ exists for many $k$ but sufficiently large population imbalance leads to instability. Consequently, only branches in the regime $\frac kN<\eta_c$ satisfy the internal-stability condition.   Here, for $\alpha=0.4$, the critical threshold $\eta_c$ is approximately $\eta_c \approx 0.289$. $P_\alpha(\eta)>0$, equivalently $\lambda_+<0$, holds for $\eta<\eta_c$. Filled and unfilled markers indicate the admissible finite-$N$ branches in the stable and unstable regimes, respectively.    }   \label{fig6-3}
\end{figure}

%
%
For even $N$, the two-cluster solutions form a discrete family indexed by $\eta  = \frac{k}{N}$ for $k=0,2,\cdots,N-2$ (even $N$). The branch $k=0$ corresponds to an equal population split and is a stationary balanced equilibrium whereas every $k>0$ branch has unequal cluster populations which give a nonzero common angular velocity. We see that the coupling parameters determine which of these branches are dynamically admissible. More precisely, the $k$-th branch can be internally stable only if $P_\alpha(\eta)>0$. Among the stable branches, the asymptotically selected value of $k$ depends on the basin of attraction containing the initial phase configuration. Since the phase roots evolve self-consistently with the mean-field, the selected branch cannot be chosen from a single initial observable, such as the initial order parameter. Thus, the even-$N$ dynamics exhibits multistability and basin-dependent selection between the stationary balanced branch $k=0$ and rotating phase-locked branches $k>0$.

\subsection{Population imbalance is compensated by rotation}

On the circle, the residual finite-population imbalance has only one way to resolve itself: it becomes a common angular drift. However, in higher dimensions, the same imbalance can be absorbed geometrically by redistributing the transverse components by producing a stationary balanced state instead of a rotating locked state. Thus, the circle resolves finite-population imbalance dynamically through collective rotation, whereas higher-dimensional spheres can absorb the same imbalance geometrically through continuous redistribution in the transverse directions.

Below, we find an explicit formula for $\Omega$ in terms of $K,\alpha$ and $\eta$. Recall that
\[
\Omega = f_R(\phi_+) = KR(q_+ - \alpha)\sin \phi_+.
\]
Since $q_+ = R\cos\phi_+$ and $\phi_+ \in (0,\pi)$, we observe
\[
R^2\sin^2 \phi_+ = R^2 - R^2\cos^2\phi_+ = R^2 - q_+^2.
\]
Hence, $\Omega$ is written as
\[
\Omega = K(q_+-\alpha) \sqrt{R^2 - q_+^2}.
\]
Now, it suffices to represent $R$ and $q_+$ in terms of $\alpha$ and $\eta$. For $q_+-\alpha$, we find
\[
q_+ - \alpha = \frac{\alpha + (1-\alpha)\eta + \eta^2}{1+\eta^2}-\alpha = \frac{(1-\alpha)\eta(1+\eta)}{1+\eta^2}.
\]
On the other hand, for $R^2 - q_+^2$, we see
\begin{align*}
R^2 - q_+^2 &= \frac{\alpha + (2-\alpha)\eta^2}{1+\eta^2} - \left( \frac{\alpha + (1-\alpha)\eta + \eta^2}{1+\eta^2}\right)^2  \\
&= \frac{(1-\alpha)(\alpha+\eta^2)(1-\eta)^2}{(1+\eta^2)^2}.
\end{align*}
Hence, we obtain
\[
\Omega = K(1-\alpha)^\frac32 \frac{(1-\eta^2)\sqrt{\alpha + \eta^2}}{(1+\eta^2)^2}\eta.
\]
For small $\eta$, we use the Taylor expansion to find 
\[
\sqrt{\alpha+\eta^2} = \sqrt\alpha \left( 1 + \frac{\eta^2}{2\alpha} + O(\eta^4)\right),\quad \frac{1-\eta^2}{(1+\eta^2)^2} = 1-3\eta^2 + O(\eta^4)
\]
which gives
\[
\Omega = K\sqrt\alpha(1-\alpha)^\frac32 \eta + O(\eta^3).
\]

\subsection{Locked states are more restrictive in higher dimensions}
Of course, locked states can also be defined on $\bbs^m$ with $m\geq2$. Precisely, such a state satisfies
\[
\dot x_i = \Omega x_i,\quad \Omega\in \mathfrak{so}(m+1).
\]
However for $m\geq2$, this condition is significantly more restrictive. It requires all particle velocities to be generated by a single $\Omega$ and hence all pairwise distances must be preserved:
\[
\frac{d}{dt} \langle x_i,x_j\rangle=0,\quad i,j\in [N].
\]
This imposes many simultaneous algebraic constraints. Hence, for a generic genuinely high-dimensional configuration, these constraints are difficult to be satisfied.

On the other hand on $\bbs^1$, by contrast, the tangent space is one-dimensional. Thus, every tangent velocity is automatically described by a scalar angular velocity. It is therefore enough to match a common scalar velocity in order to obtain a locked state. 

We need to mention that if all particles on $\bbs^m$ with $m\geq2$ lie in a common two-dimensional plane, then that plane is invariant under the dynamics. In such a case, the dynamics reduces to an embedded $\bbs^1$ system, and $\bbs^1$-type phase locked states may appear. However, for genuinely high-dimensional random initial data, such coplanar configurations are nongeneric. 

\subsection{Comparison with  the classical Kuramoto model}
So far, we have shown that within the two-cluster family, the circle $\bbs^1$ exhibits a finite-size selection between stationary equilibria and rotating phase-locked states. Recall that the exact formula for $\Omega$
\[
\Omega = K(1-\alpha)^\frac32 \frac{(1-\eta^2)\sqrt{\alpha + \eta^2}}{(1+\eta^2)^2}\eta,\quad \eta = \frac kN
\]
shows that the branch is stationary if and only if the two populations are exactly balanced, i.e., $\eta=0$. On the other hand, for $\eta\neq0$, the population imbalance is compensated dynamically by collective rotation. Consequently, the even-$N$ branch $k=0$ is a stationary balanced equilibrium, whereas every even branch with $k\geq1$ and every odd branch is a rotating phase-locked state. This mechanism is called population-imbalance-induced phase-locking or finite-size parity selection between equilibria and phase-locked state.

In fact, this mechanism is reminiscent of phase-locking in the classical Kuramoto model
\[
\dot\theta_i = \nu_i + \frac\kp N \sum_{k=1}^N \sin(\theta_k - \theta_i)
\]
but its origin is fundamentally different. By summing the equation above over $i\in [N]$, we see that any phase-locked state has common frequency
\[
\Omega = \bar\nu := \frac1N \sum_{i=1}^N \nu_i.
\]
Thus, the collective drift in the Kuramoto model is given by the mean natural frequency and disappears in the mean-zero rotating frame. In particular, even when the natural frequencies are non-identical, a phase-locked state becomes stationary whenever $\bar \nu=0$. In the present model, by contrast, all oscillators are intrinsically identical, but a  common angular velocity is generated by an imbalanced cluster population $\Omega_\eta$. Hence, the drift is self-generated by finite-population asymmetry and the nonlinear mean-field coupling, rather than imposed by intrinsic frequency heterogeneity.

\subsection{Summary}
On $\bbs^1$, finite-$N$ imbalance cannot be eliminated by redistributing transverse components, because the fixed-projection level set consists of only two points.   For the pure two-cluster balanced mechanism, a finite population imbalance obstructs stationary balance and is instead converted into a common angular drift.  However, in higher dimensions, the transverse space has enough degrees of freedom to absorb the same imbalance geometrically. As a result, the dynamics can converge to a stationary balanced equilibrium rather than to a rotating locked state.

 \section{Numerical simulations} \label{sec:7}
 \setcounter{equation}{0}
 
 In this section, we provide numerical results that support our theoretical findings and suggest qualitative insights. For numerical implementation, we use the fourth-order Runge--Kutta method with a time step $\Delta t= 10^{-3}$. For coupling strengths, we choose
 \[
 \kp_1=0.4,\quad \kp_2=-1,\quad \kp_1+\kp_2=-0.6<0,\quad  \alpha = \frac{\kp_1}{-\kp_2} =0.4.
 \]

\subsection{Case of $\bbs^m$ with $m\geq2$} 
In this subsection, we consider $\bbs^m$ with $m\geq2$. For simulations, we define the total velocity functional
\[
V(\dot x(t)) := \sqrt{\frac1N \sum_{i=1}^N \|\dot x_i(t)\|^2} 
\]
and recall $q_i(t) =\langle x_i(t),x_c(t)\rangle$.

 \begin{figure}[h]
\centering
\mbox{
\subfigure{  
\includegraphics[width=1\textwidth]{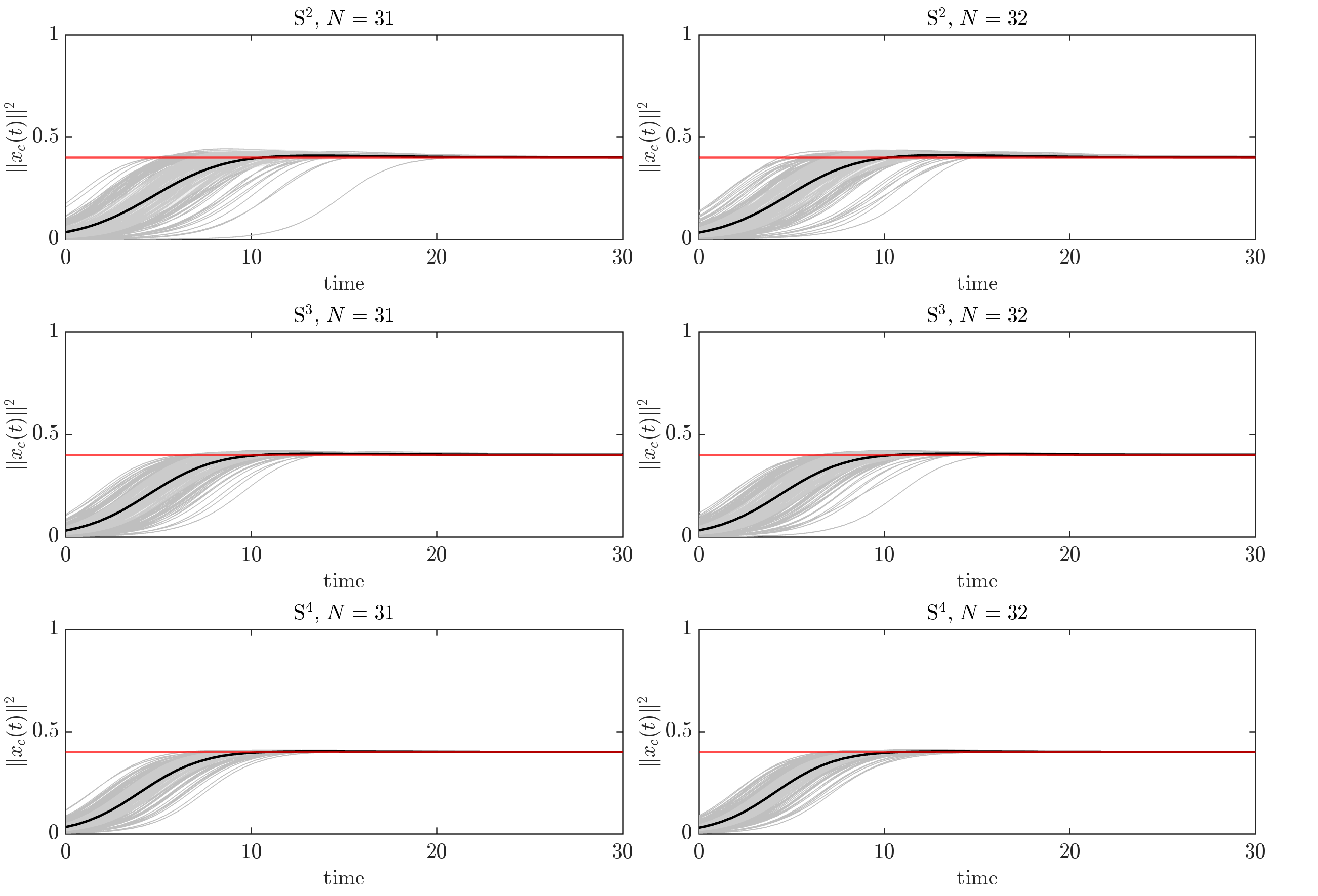}\label{fig:3-1}}
}
\caption{ The left and right panels correspond to $N=31$ (odd) and $N=32$ (even), respectively. The trajectories from 200 independent trials are shown in gray, whereas their ensemble average is represented by the solid black curve. The horizontal red line marks the value $\alpha = 0.4$.  The simulations show that $\|x_c(t)\|^2$ converges to $\alpha$ for the sampled random initial data, for both parities of $N$ and for all dimensions tested in the simulations.   }  \label{fig5-1}
\end{figure}
 
 \begin{figure}[h]
\centering
\mbox{
\subfigure{  
\includegraphics[width=1\textwidth]{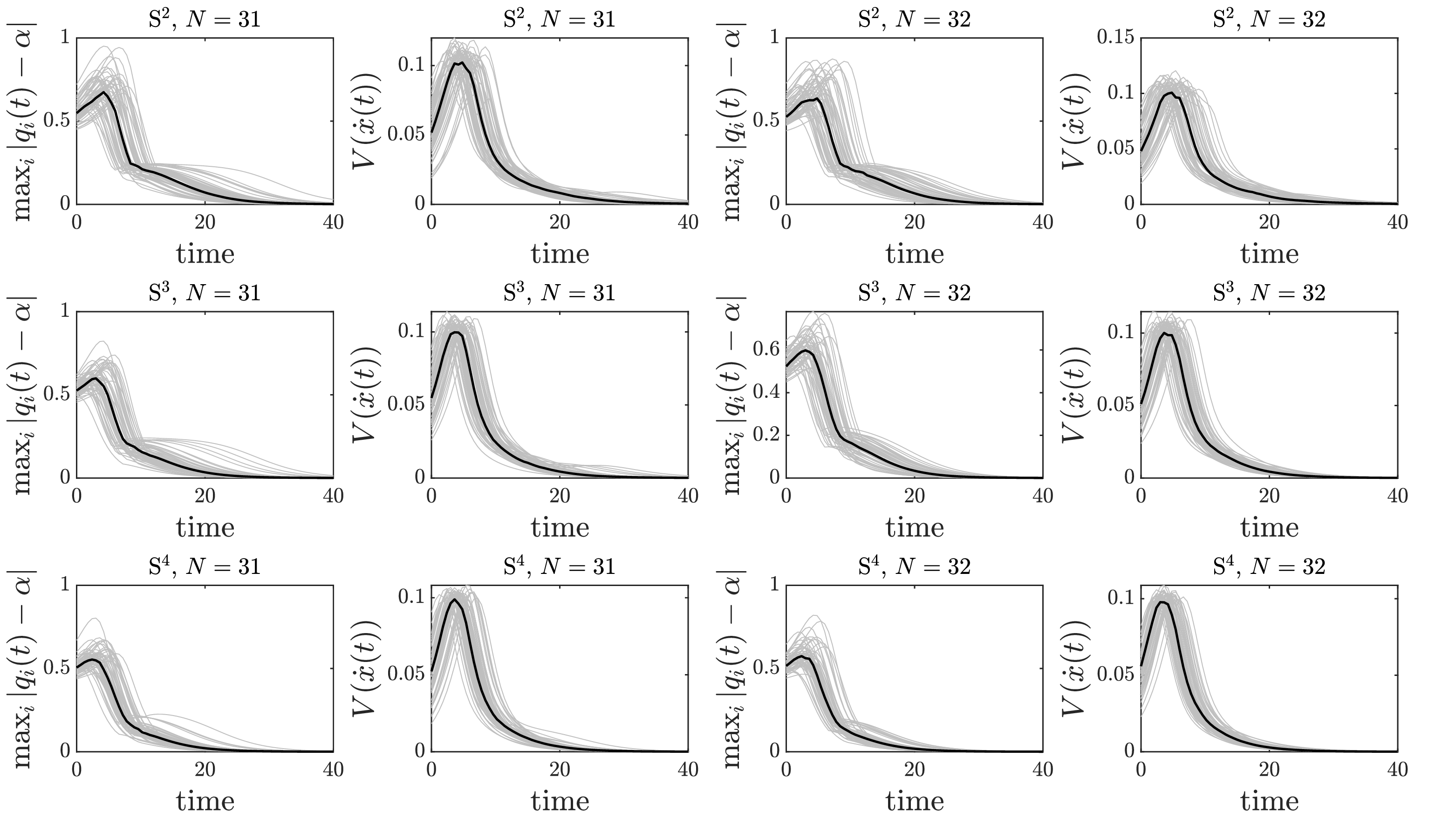}\label{fig:3-1}}
}
\caption{  The two left panels correspond to $N=31$ (odd), whereas the two right panels correspond to $N=32$ (even). The quantities $\max |q_i(t)-\alpha|$ and $V(\dot x(t))$ are plotted. The trajectories from 200 independent trials are shown in gray and their ensemble averages are presented by the solid black curves. The numerical results show that both $\max |q_i(t)-\alpha|$ and $V(\dot x(t))$ converge to zero for the sampled random initial data, for both parities of $N$ and for all dimensions tested in the simulations. }  \label{fig5-1}
\end{figure}

\subsection{Case of $\bbs^1$}
In this subsection, we consider $\bbs^1$. 
For simulations, we define
 \[
 V(\dot \theta(t)) :=  \left( \frac1N \sum_{i=1}^N (\dot \theta_i(t) - \bar{\dot \theta}(t))^2\right)^\frac12       ,\quad \bar{\dot\theta}(t) := \frac1N \sum_{i=1}^N \dot \theta_i(t)
 \]
 which measures the difference of the velocities around their mean. In particular, $V(\dot \theta(t))=0$ if and only if all oscillators have the same velocity, although the common limiting velocity need not vanish.

%

 \begin{figure}[H]
\centering
\mbox{
\subfigure{  
\includegraphics[width=0.9\textwidth]{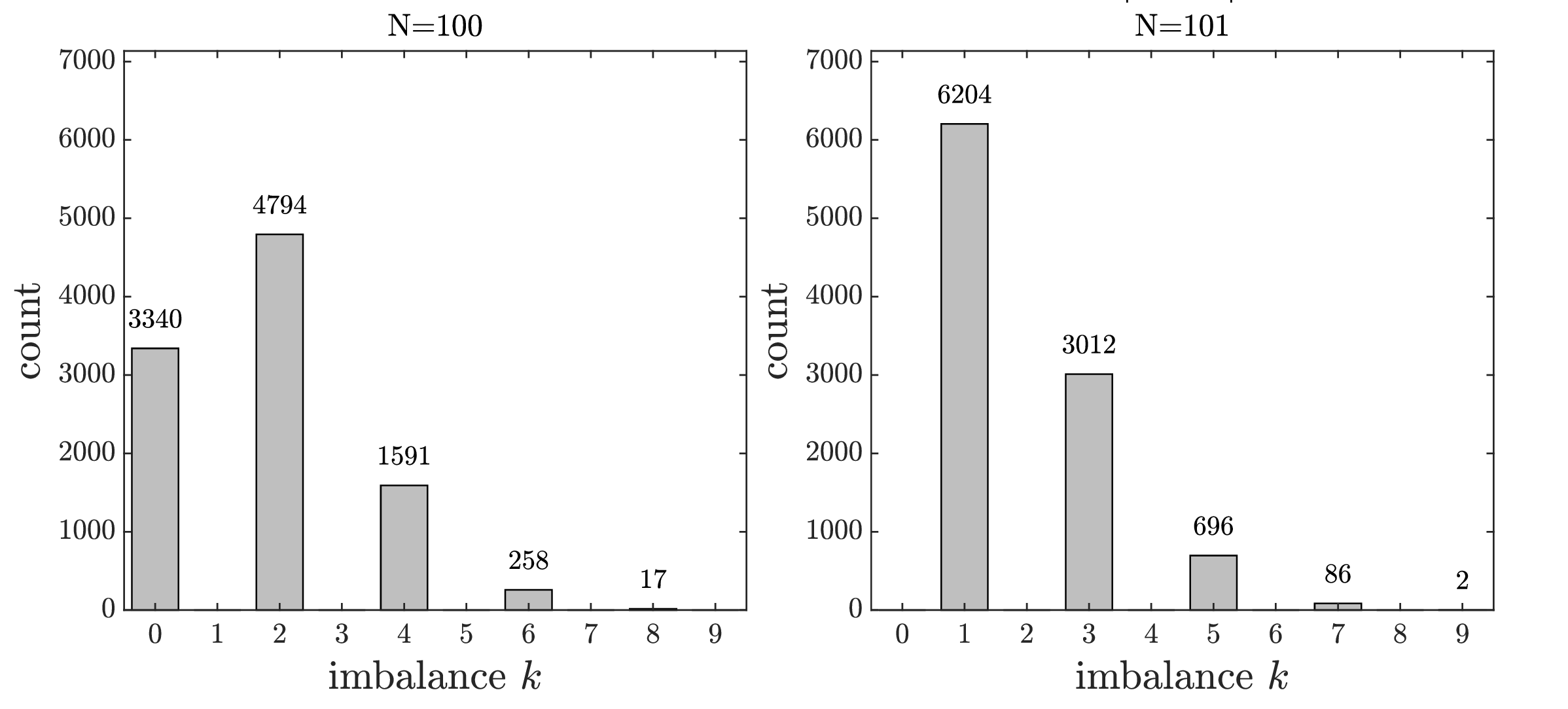}\label{fig:3-1}}
}
\caption{ For each of the two cases with $N=100$ and $N=101$, we performed 10,000 independent simulations. By writing the convergent value of the order parameter as $R_\infty^2 = R_{N,k}^2$, the histograms show the number of realizations associated with each imbalance $k$. Large-imbalance branches are rarely observed and the empirical distribution is concentrated on branches with small $k$.  }  \label{fig5-1}
\end{figure}

 \begin{figure}[H]
\centering
\mbox{
\subfigure{  
\includegraphics[width=0.9\textwidth]{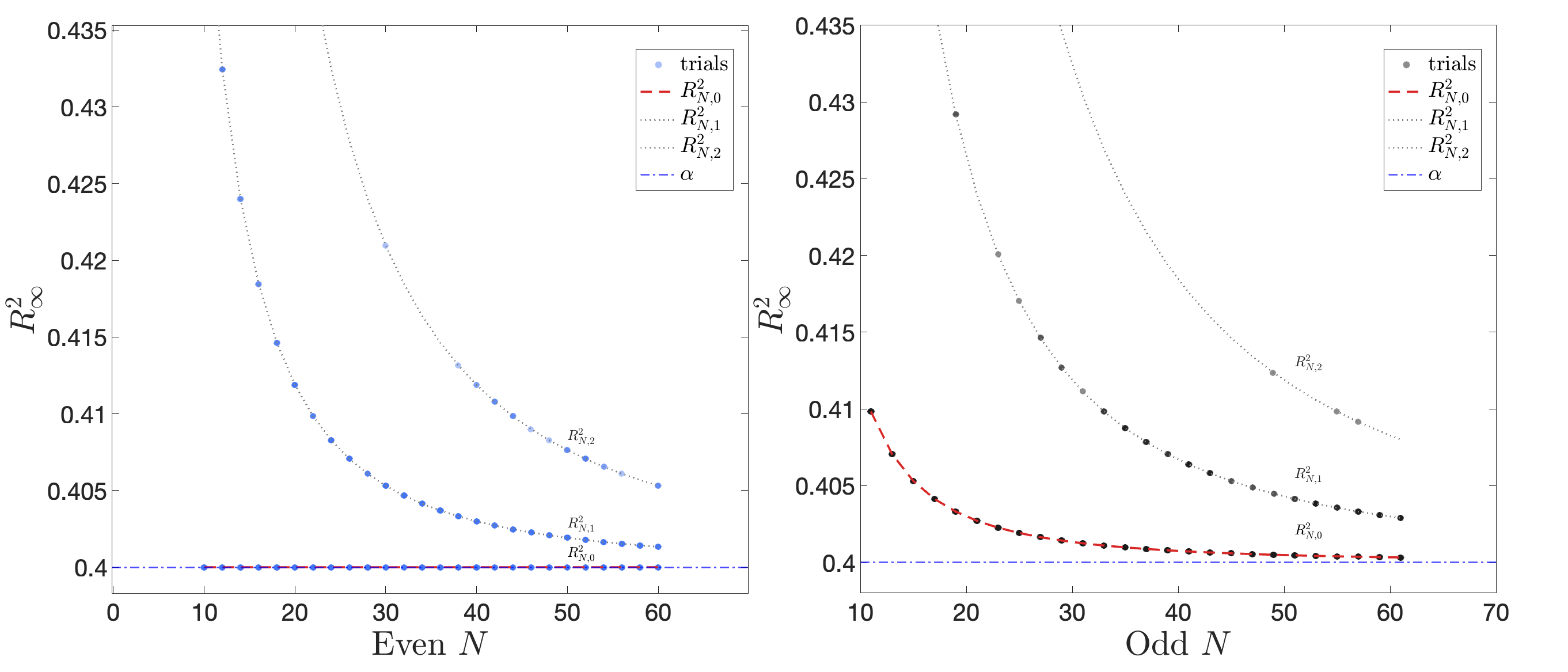}\label{fig:3-1}}
}
\caption{ The left and right panels show the convergent values of the squared order parameter for even and odd values of $N$, respectively. The markers represent the values obtained from numerical simulations and the curves indicate the theoretical branches: for each fixed imbalance $k$, the finite-size correction satisfies
\[
R_{N,k}^2-\alpha = O(N^{-2}),\quad N\to\infty.
\] The horizontal blue line denotes the limiting value $\alpha=0.4$. The numerical data lie on the predicted discrete branches which demonstrates the quantized selection of $R_\infty^2$. For even $N$, the balanced branch $R_{N,0}^2=\alpha$ is admissible. For odd $N$, exact population balance  $R_{N,0}^2=\alpha$ is not possible, so the lowest admissible branch remains slightly above $\alpha$, although it approaches $\alpha$ as $N$ increases. }  \label{fig5-1}
\end{figure}


 \begin{figure}[h]
\centering
\mbox{
\subfigure{  
\includegraphics[width=0.9\textwidth]{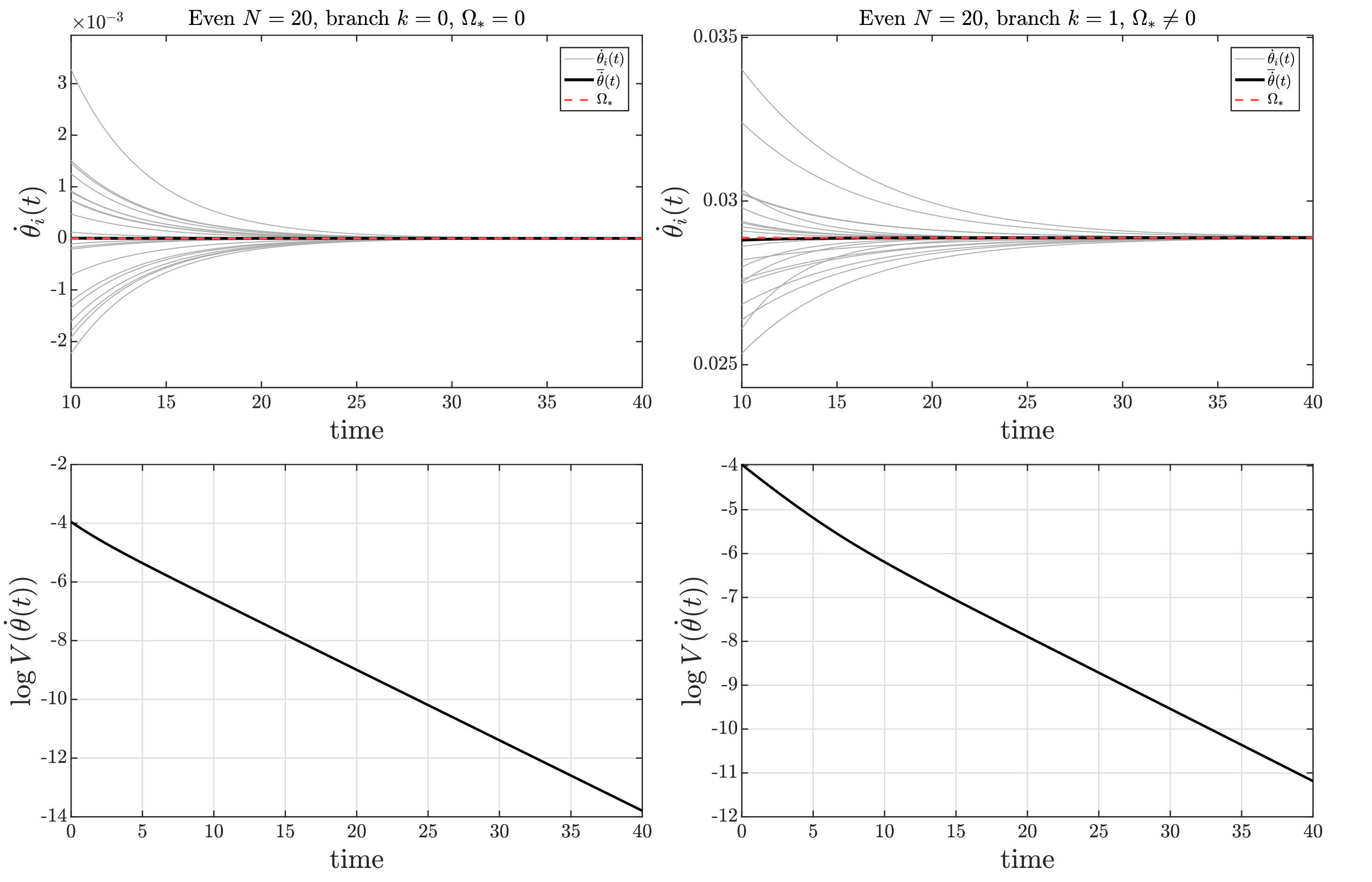}\label{fig:3-1}}
}
\caption{ For even $N$ (here, $N=20$), both convergence to an equilibrium and convergence to a phase-locked state may be observed. On the balanced branch $k=0$, the exponential decay of $\dot\theta_i(t)$ to zero shows convergence to an equilibrium. On the imbalanced branch $k=2$, by  contrast, all angular velocities converge to a common nonzero value $\Omega_*\neq0$ which means convergence to a rotating phase-locked state.}  \label{fig5-1}
\end{figure}

 \begin{figure}[H]
\centering
\mbox{
\subfigure{  
\includegraphics[width=0.9\textwidth]{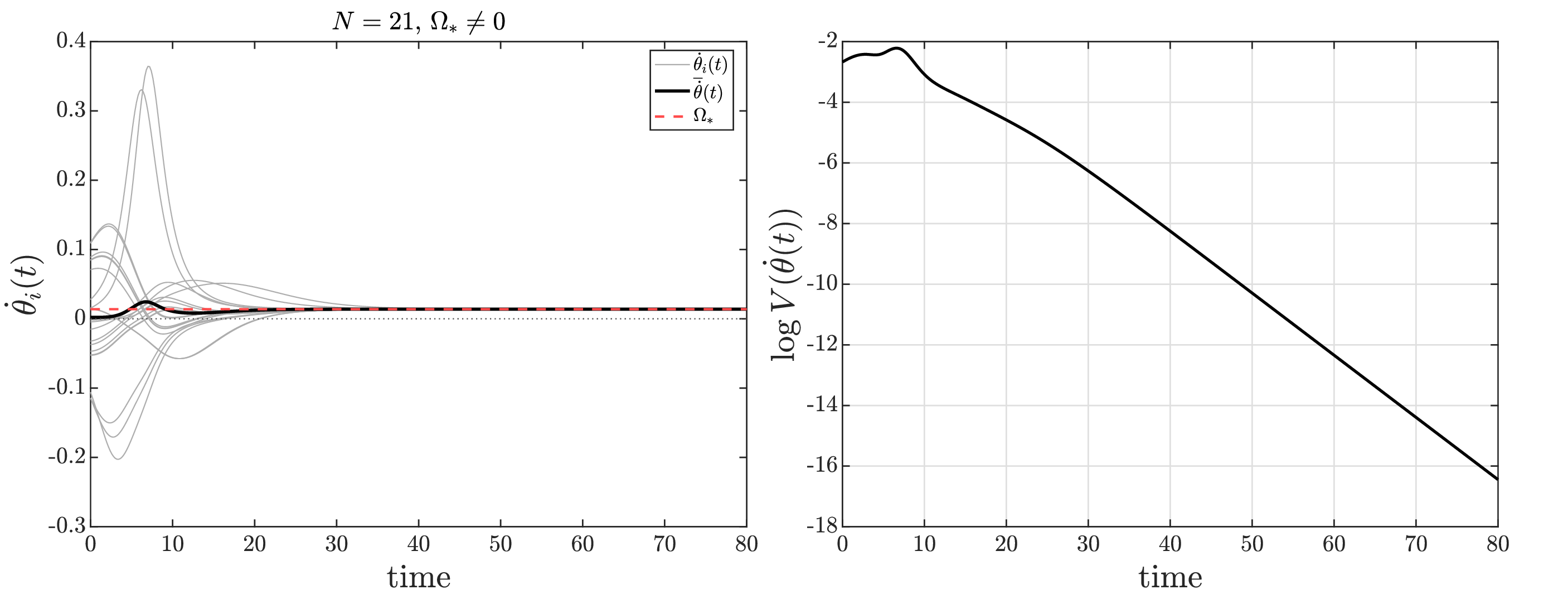}\label{fig:3-1}}
}
\caption{For odd $N$ (here, $N=21$), the balanced branch $k=0$ is not admissible. Consequently, the convergence of all $\dot\theta_i(t)$ to the same nonzero value $\Omega_*\neq0$ confirms convergence to a rotating phase-locked state rather than to an equilibrium. }  \label{fig5-1}
\end{figure}

%
%

\appendix
\section{Proof of Lemma \ref{L3.1}} \label{sec:app.A}
\setcounter{equation}{0}
Let $G^* = (g_{ij}^*)$ with $g_{ij}^*=(1-\alpha)(1-\cos(\theta_i-\theta_j))$ and $\theta_i = \frac{2\pi(i-1)}{N}$. Define 
\[
\mathbf{1} := (1,\cdots,1)^\top,\quad \mathbf{a} := (\cos\theta_1,\cdots,\cos\theta_N)^\top,\quad \mathbf{b}:= (\sin\theta_1,\cdots,\sin\theta_N)^\top.
\]
Then, we have
\[
G^* = (1-\alpha) (\mathbf{1}\mathbf{1}^\top - \mathbf{a}\mathbf{a}^\top - \mathbf{b}\mathbf{b}^\top). \] 
We observe
\[
\sum_{i=1}^N \cos\theta_i=0=\sum_{i=1}^N \sin \theta_i,\quad \sum_{i=1}^N \cos^2\theta_i=\frac N2=\sum_{i=1}^N \sin^2 \theta_i,\quad \sum_{i=1}^N \cos\theta_i\sin\theta_i=0.
\]
Since we have $\mathbf{1}\perp \mathbf{a}$ and $\mathbf{1}\perp \mathbf{b}$,
\[
\|\mathbf{a}\|^2 = \frac N2 = \|\mathbf{b}\|^2,\quad \mathbf{a}\perp \mathbf{b}.
\]
Below, we find eigenvalues for $G^*$ corresponding to eigenvectors $\mathbf{1},\mathbf{a}$ and $\mathbf{b}$.

For $\mathbf{1}$, we observe
\begin{align*}
G^*\mathbf{1} &= (1-\alpha) (\mathbf{1}\mathbf{1}^\top - \mathbf{a}\mathbf{a}^\top - \mathbf{b}\mathbf{b}^\top)\mathbf{1}  = N(1-\alpha)\mathbf{1}
\end{align*}
which shows that the eigenvalue of $\mathbf{1}$ is $N(1-\alpha)$. 

For $\mathbf{a}$, we observe
\begin{align*}
G^*\mathbf{a} &= (1-\alpha) (\mathbf{1}\mathbf{1}^\top - \mathbf{a}\mathbf{a}^\top - \mathbf{b}\mathbf{b}^\top)\mathbf{a}  = -\frac N2(1-\alpha)\mathbf{a}
\end{align*}
which shows that the eigenvalue of $\mathbf{a}$ is $-\frac N2(1-\alpha)$. 

Similarly, we find
\[
G^*\mathbf{b} = -\frac N2(1-\alpha)\mathbf{b}
\]
which shows that the eigenvalue of $\mathbf{b}$ is also $-\frac N2(1-\alpha)$. 

Lastly, for a vector $v$ which is perpendicular to $\mathbf{1},\mathbf{a},\mathbf{b}$, we have
\[
G^*v =0.
\]
Hence, all eigenvalues of $G^*$ are
\[
N(1-\alpha),   \quad -\frac N2(1-\alpha), \quad 0.
\]
We observe
\[
\sum_{j=1}^N g_{ij}^* = \sum_{j=1}^N (1-\langle x_i^*,x_j^*\rangle) = N(1-\alpha).
\]
Then, we have for $u=(u_1,\cdots,u_N)\in \bbr^N$, 
\begin{align*}
\mathcal Q_{*}(u) &= \frac{1}{2N} \sum_{i,j=1}^N g_{ij}^*(u_i + u_j)^2 = \frac1N \sum_{i=1}^N u_i^2 \left(\sum_{j=1}^N g_{ij}^*\right) + \frac1N \sum_{i,j=1}^N g_{ij}^* u_i u_j \\
&= (1-\alpha) \sum_{i=1}^N u_i^2 + \frac1N \langle u,G^*u\rangle  = \left\langle u, \left[(1-\alpha)I + \frac1N G^*\right]u     \right\rangle  \\
&\geq  \left( 1-\alpha + \frac{\lambda_{\min}(G^*)}{N} \right) \|u\|^2.
\end{align*}
 Hence, the eigenvalues of the quadratic form matrix $(1-\alpha)I + \frac1N G^*$ are
\[
2(1-\alpha), \,\, \frac{1-\alpha}{2}, \,\, 1-\alpha, 
\]
where the minimal value is 
\[
\chi_* = \frac{1-\alpha}{2}.
\]
Therefore, we have
\[
\frac{1}{2N} \sum_{i,j=1}^N g_{ij}^* (u_i + u_j)^2 \geq \frac{1-\alpha}{2}\sum_{i=1}^N u_i^2.
\]

\begin{remark} \label{rem:A.1} 
Define  the coercivity constant 
\[
\chi_* := 1-\alpha + \frac{\lambda_{min}(G^*)}{N}.
\]
Then, $\chi_*$ need not be strictly positive in the one-dimensional case $m=1$. Indeed, consider the balanced equilibria on $\bbs^1$ with $N=2s$ in which $s$ oscillators are located at $\phi_*$ and the remaining $s$ oscillators at $-\phi_*$ where $\cos\phi_* = \sqrt \alpha$. Then, $g_{ij}^* = 0$ if $i,j$ belong to the same cluster and $g_{ij}^* = 2(1-\alpha)$ if $i,j$ belong to different clusters.  
Hence, 
\[
G^* = 2(1-\alpha) \begin{pmatrix} 0 & J_s \\ J_s & 0 \end{pmatrix}, 
\]
where $J_s$ denotes the $s\times s$ matrix whose entries are all one. Then, the eigenvalues of $G^*$ are
\[
N(1-\alpha), \,\, 0, \,  -N(1-\alpha)
\]
which gives $\chi_*=0$. Thus, in contrast to the case of $m\geq2$, the quadratic form $\mathcal Q_*$ may fail to be coercive on $\bbs^1$. 

This spectral degeneration provides an algebraic manifestation of the geometric obstruction discussed in Proposition \ref{prop:fixed-projection}: on $\bbs^1$, the transverse latitude consists of only two points and the coercive gap available in higher dimensions may collapse. 

\end{remark}

\section{Proof of Theorem \ref{T5.2}} \label{sec:app.B}
\setcounter{equation}{0}
Since the proof is rather lengthy, we split into several steps. \newline

\noindent $\bullet$ (Step 1: Mean-zero dynamics)  Define the mean-zero perturbation $u(t) = (u_1(t),\cdots,u_N(t))$ by  
\[
u_i(t) := \theta_i(t) - \frac1N\sum_{k=1}^N \theta_k(t) - e_i,\quad e_i =  \begin{cases}
a_* = \frac{\ell}{N}\Delta_*,\quad i\in A, \\
b_* = -\frac{m}{N}\Delta_*,\quad i\in B, 
\end{cases}
\]
where $e_i$ was defined in Section \ref{sec:5.2}. Then, our goal is to show 
\[
\lim_{t\to\infty} |u_i(t)|=0,\quad i\in [N].
\]
Thus, once $u_i(t)\to0$ is established, the convergence $R(t)^2\to R_*^2$ follows immediately:
\[
R(t)^2 = \left| \frac1N \sum_{i=1}^N e^{\mi \theta_i(t)}\right|^2 = \left|e^{\mi c(t) } \frac1N \sum_{i=1}^N  e^{\mi (e_i+u_i(t))}\right|^2 = \left| \frac1N \sum_{i=1}^N  e^{\mi (e_i+u_i(t))}\right|^2 \to R_*^2 =R_{N,k}^2
\]
as in \eqref{F-12}. In what follows, we consider $u_i$ instead of $\theta_i$. Since the common phase does not affect the dynamics of the order parameter $R(t)$, we consider the mean-zero projection:
\[
P:= I - \frac1N \mathbf{1}\mathbf{1}^\top.
\]
For $e=(e_1,\cdots,e_N)\in \bbr^N$ and $u=(u_1,\cdots,u_N)\in\bbr^N$, we denote $Y(t):= e + u(t)$. Then, we have $Y(t) = P\Theta(t)$ and if we write $\dot \Theta = F(\Theta)$, then $Y(t)$ satisfies
\[
\dot Y(t) = P \dot \Theta(t) = PF(\Theta(t)) = PF(Y(t)), 
\]
where $F$ is invariant under the common phase shift:
\[
F(\Theta + \beta \mathbf{1}) = F(\Theta).
\]
Then, since $e$ is a two-cluster rotating wave profile with $F(e) = \Omega_*\mathbf{1}$, then we have $PF(e) =0$. Thus, we find
\[
\dot u = PF(e+u) = P [ F(e+u) - F(e)].
\]

\vspace{0.5cm}

\noindent $\bullet$ (Step 2: Local cluster norm) We define the local mean vector for $u_i$:
\[
\bar u_A(t):= \frac1m \sum_{i\in A} u_i(t),\quad \bar u_B(t):= \frac1\ell \sum_{j\in B} u_j(t).
\]
In addition, we define the maximal differences between each group and the local means:
\begin{align*}
&R_A(u(t)) = \max_{i\in A} |u_i(t) - \bar u_A(t)|,\quad R_B(u(t)) := \max_{j\in B} |u_j(t) -\bar u_B(t)|, \\
& H(u(t))   := |\bar u_A(t) - \bar u_B(t)|,\quad \Gamma(u(t))  := \max\big\{R_A(u(t)),R_B(u(t)),H(u(t))\big\}.
\end{align*}
Since the mean of $u$ is zero, we have $m\bar u_A  + \ell \bar u_B=0$ and hence
\[
|\bar u_A| = |q(\bar u_A - \bar u_B)|\leq H,\quad |\bar u_B| = |-p(\bar u_A - \bar u_B)|\leq H.
\]
Hence, if $i\in A$ and $j\in B$,  then  we have
\begin{align*}
|u_i | &\leq |u_i - \bar u_A| + |\bar u_A|\leq R_A + H \leq 2\Gamma, \\
|u_j| & \leq |u_j - \bar u_B| + |\bar u_B| \leq R_B+ H \leq 2\Gamma, \\
|u_i - u_j|&\leq |u_i - \bar u_A| +|\bar u_A  - \bar u_B| + |\bar u_B - u_j| \leq 3\Gamma
\end{align*}
which gives 
\[
\|u\|_\infty\leq 2\Gamma, \quad |u_i- u_j|\leq 3\Gamma. 
\]
For initial data, we recall
\[
\theta_i(0) = \psi + e_i + r_i,\quad r_i = \begin{cases}
\xi_i,\quad i\in A, \\
\eta_i,\quad i\in B.
\end{cases}
\]
Then, we use $\sum_{i=1}^N  e_i=0$ and $\bar r:= \frac1N \sum_{i=1}^N r_i$ to find 
\begin{align*}
u_i(0) = \theta_i(0) - e_i - c(0) = (\psi + e_i + r_i) - e_i - \frac1N \sum_{i=1}^N (\psi+e_i + r_i) = r_i - \bar r.
\end{align*}
In addition, we observe
\begin{align*}
R_A(0) &= \max_{i\in A} \left| \xi_i - \frac1m \sum_{k\in A}\xi_k\right| \leq 2\veps, \\
R_B(0) & =\max_{i\in B} \left| \eta_i - \frac1\ell \sum_{k\in B}\eta_k\right| \leq 2\veps, \\ 
H(0) & = \left| \frac1m \sum_{i\in A}\xi_i - \frac1\ell \sum_{k\in B} \eta_k     \right|\leq 2\veps
\end{align*}
which yields
\[
\Gamma(0) \leq 2\veps <\delta.
\]
$\bullet$ (Step 3: Nonlinear remainder estimate) We use the mean-value theorem to find
\[
F(e+u) - F(e) = DF(e)u + \mathcal R(u),\quad \mathcal R(u) := \int_0^1 [DF(e+su) - DF(e)] u ds.
\]
Here, we write $w:= su$ for $0\leq s\leq 1$, a point between $e$ and $e+u$. Now, we define a temporal set
\[
T_* := \sup \{ T>0: \Gamma(u(t))\leq\delta,\quad t\in [0,T)\}.
\]
Below, we estimate on $[0,T_*)$ where $\Gamma(u(t))\leq \delta$. Then, we have
\[
\|u\|_\infty \leq 2\delta,\quad \|w\|_\infty \leq 2\delta.
\]
Denote
\[
M(w) := DF(e+w) - DF(e),\quad w=su.
\]
We use for $\ell\neq i$, 
\begin{align*}
 \partial_{\theta_\ell} F_i(\Theta) &= \frac{\kp_1}{N}\cos(\theta_\ell - \theta_i) + \frac{\kp_2}{N^2} \sum_{k=1}^N \cos(\theta_\ell + \theta_k - 2\theta_i), \\
  \partial_{\theta_k} F_i(e+w)  & = \frac{\kp_1}{N}\cos(e_k +w_k  - e_i-w_i) + \frac{\kp_2}{N^2} \sum_{\ell=1}^N \cos(e_\ell+w_\ell + e_k +w_k - 2e_i-2w_i)
\end{align*}
to estimate
\begin{align*}
|(M(w))_{i\ell}|  &= |\partial_{\theta_\ell} F_i(e+w) - \partial_{\theta_\ell} F_i(e)| \\
& \leq   \frac{\kp_1}{N}    \mathcal I_{11} + \frac{|\kp_2|}{N^2} \sum_{k=1}^N  \mathcal I_{12}.
\end{align*}
For $\mathcal I_{11}$, we observe
\begin{align*}
\mathcal I_{11} &= |\cos(e_k +w_k  - e_i-w_i) - \cos(e_k - e_i)     |\leq 2\sin (2\delta), 
\end{align*}
where the argument difference is less than or equal to $4\delta$ and  we used
\[
|\cos(x+\eta) - \cos x | = \left| 2\sin \left(x+\frac\eta2\right) \sin \frac\eta2\right| \leq 2\sin\frac{|\eta|}{2}.
\]
Similarly for $\mathcal I_{12}$, we observe
\begin{align*}
\mathcal I_{12} \leq 2\sin (4\delta), 
\end{align*}
where the argument difference is less than or equal to $8\delta$. Hence, we have
\[
| (M(w))_{i\ell}|\leq \frac{2\kp_1\sin (2\delta) + 2K \sin (4\delta)}{N},\quad i\neq \ell.
\]
Since both $DF(e+w)$ and $DF(e)$ have zero row sum, so does $M(w)$:
\[
(M(w)u)_i = \sum_{\ell \neq i} (M(w))_{i\ell} (u_\ell - u_i).
\]
We use $|u_\ell - u_i|\leq 3\Gamma$ to find
\[
|(M(w)u)_i| \leq \sum_{\ell \neq i} \frac{2\kp_1\sin (2\delta) + 2K \sin (4\delta)}{N} \cdot 3\Gamma \leq 3 (2\kp_1\sin (2\delta) + 2K \sin (4\delta)) \Gamma
\]
which gives
\[
\|M(w)u\|_\infty \leq 3 (2\kp_1\sin (2\delta) + 2K \sin (4\delta)) \Gamma.
\]
Hence, we find
\[
\| \mathcal R(u)\|_\infty = \left\|\int_0^1 M(su) uds \right\|_\infty \leq 3 (2\kp_1\sin (2\delta) + 2K \sin (4\delta)) \Gamma= : 3a_\delta \Gamma.
\]
$\bullet$ (Step 4: Block structure of $DF(e)$) For simplicity, we write 
\[
J:= DF(e),\quad c_*:=\cos\Delta_*,\quad c_2:=\cos(2\Delta_*)
\]
and recall
\[
\partial_{\theta_k} F_i(\theta)  = \frac{\kp_1}{N}\cos(\theta_k - \theta_i) + \frac{\kp_2}{N^2} \sum_{k=1}^N \cos(\theta_\ell + \theta_k - 2\theta_i).
\]
Below, we only calculate the off-diagonal term, since the row sum of $J$ is zero. Precisely, since $F_i$ is invariant under a common phase shift
\[
F_i(\Theta + \beta \mathbf{1} )= F_i(\Theta),
\]
we differentiate the above relation with respect to $\beta$ to find
\[
0 = \frac{d}{d\beta } F_i(\Theta + \beta \mathbf{1}) = \sum_{k=1}^N \partial_{\theta_k} F_i(\Theta)=  \sum_{k=1}^N J_{ik}
\]
which gives the diagonal entry:
\[
J_{ii} =-\sum_{k\neq i} J_{ik}.
\]
In addition, for later use, we write
\[
(Ju)_i = J_{ii}u_i + \sum_{k\neq i} J_{ik}u_k = \sum_{k\neq i} J_{ik}(u_k - u_i).
\]
We consider four cases. 

$\diamond$ (Case A: $i,k\in A$, $i\neq k$) In this case, we have
\[
e_i = e_k = a_*.
\]
Hence, the pairwise term becomes
\[
\frac{\kp_1}{N} \cos(e_k - e_i ) = \frac{\kp_1}{N}
\]
and the three-body term becomes
\begin{align*}
\frac{\kp_2}{N^2} \sum_{\ell=1}^N   \cos (e_\ell + e_k - 2e_i) &= \frac{\kp_2}{N^2} \sum_{\ell \in A}   \cos (e_\ell + e_k - 2e_i) + \frac{\kp_2}{N^2} \sum_{\ell \in B}     \cos (e_\ell + e_k - 2e_i) \\
& =  \frac{\kp_2}{N^2} \sum_{\ell \in A}    \cos (a_* + a_* - 2a_* ) + \frac{\kp_2}{N^2} \sum_{\ell \in B}    \cos (a_* + b_* - 2a_*)  \\
& = \frac{\kp_2}{N^2} ( m + \ell \cos \Delta_*) = \frac{\kp_2}{N^2}( m+ \ell c_*) = \frac{\kp_2}{N} (p+ qc_*).
\end{align*}
Hence, 
\[
J_{\textup{off}}^{AA}: = \frac1N ( \kp_1+\kp_2(p+qc_*)).
\]

$\diamond$ (Case B: $i,k\in B$, $i\neq k$) In this case, we have
\[
e_i = e_k = b_*.
\]
Similar to Case A, we have
\[
J_{\textup{off}}^{BB} := \frac1N ( \kp_1+\kp_2(q+pc_*)).
\]

$\diamond$ (Case C: $(i,k)\in A\times B$) In this case, we have
\[
e_i = a_*,\quad e_k = b_*.
\]
Then, the pairwise term becomes
\[
\frac{\kp_1}{N}\cos(e_k - e_i) = \frac{\kp_1}{N} \cos(b_*- a_*) = \frac{\kp_1}{N}c_*
\]
and the three-body part becomes
\begin{align*}
\frac{\kp_2}{N^2} \sum_{\ell=1}^N   \cos (e_\ell + e_k - 2e_i) &= \frac{\kp_2}{N^2} \sum_{\ell \in A}   \cos (e_\ell + e_k - 2e_i) + \frac{\kp_2}{N^2} \sum_{\ell \in B}     \cos (e_\ell + e_k - 2e_i) \\
&= \frac{\kp_2}{N^2} (mc_* + \ell c_2) = \frac{\kp_2}{N}(pc_* + qc_2).
\end{align*}
Hence, we have
\[
J^{AB}:= \frac1N ( \kp_1c_* + \kp_2(pc_* + qc_2)).
\]

$\diamond$ (Case D: $(i,k)\in B\times A$) In this case, we have
\[
e_i = b_*,\quad e_k = a_*.
\]
Similar to Case C, we have
\[
J^{BA}:= \frac1N ( \kp_1c_* + \kp_2(qc_* + pc_2)).
\]

$\bullet$ (Step 5: Dissipative representation for $Ju$) We here calculate $Ju$. For $i\in A$, 
\begin{align*}
(Ju)_i &= J_{\textup{off}}^{AA} \sum_{i\neq k, k\in A} (u_k - u_i) + J^{AB} \sum_{j \in B} (u_j - u_i) \\
& = -J_{\textup{off}}^{AA} m (u_i - \bar u_A) +  J^{AB}\ell( \bar u_B - u_i) \\
& = -(m J_{\textup{off}}^{AA} + \ell J^{AB})(u_i - \bar u_A) - \ell J^{AB} (\bar u_A - \bar u_B) \\
& = -\Lambda_A (u_i - \bar u_A) -  \ell J^{AB} (\bar u_A - \bar u_B).
\end{align*}
Here, $\Lambda_A$ introduced in \eqref{fourconstants} is defined as
\[
\Lambda_A = mJ_{\textup{off}}^{AA} + \ell J^{AB}.
\]
Similarly for $j\in B$, we find
\[
(Ju)_j = - \Lambda_B(u_j  - \bar u_B) + m J^{BA} (\bar u_A - \bar u_B), 
\]
where $\Lambda_B$ introduced in \eqref{fourconstants} is defined as
\[
\Lambda_B = \ell J_{\textup{off}}^{BB} + m J^{BA}.
\]
Then, we also obtain
\[
\overline{(Ju)}_A:= \frac1m \sum_{i\in A} (Ju)_i = - \ell J^{AB} (\bar u_A - \bar u_B),\quad \overline{(Ju)}_B:= \frac1\ell \sum_{j\in B} (Ju)_j = mJ^{BA} (\bar u_A - \bar u_B).
\]
Hence, 
\begin{align*}
(Ju)_i - \overline{(Ju)}_A & = - \Lambda_A(u_i - \bar u_A), \quad i\in A, \\
(Ju)_j - \overline{(Ju)}_B & = - \Lambda_B(u_j - \bar u_B), \quad j\in B, \\
\overline{(Ju)}_A - \overline{(Ju)}_B & = - (\ell J^{AB} + m J^{BA})(\bar u_A- \bar u_B) =-\Lambda_\Delta ( \bar u_A - \bar u_B), 
\end{align*}
where $\Lambda_\Delta$ in \eqref{fourconstants} is defined as
\[
\Lambda_\Delta = \ell J^{AB} + m J^{BA}.
\]
$\bullet$ (Step 6: Differential inequalities for $R_A, R_B$ and $H$) For simplicity, we write
\[
Z:= Ju + \mathcal R(u) = DF(e)u + \mathcal R(u)
\]
to rewrite
\[
\dot u = PZ 
\]
whose componentwise form is
\[
\dot u_i = Z_i - \bar Z,\quad \bar Z = \frac1N \sum_{k=1}^N Z_k,\quad \bar Z_A := \frac1m \sum_{i\in A} Z_i,\quad \bar Z_B := \frac1\ell \sum_{j\in B} Z_j.
\]

$\diamond$ (Case 1: $R_A$) For $i\in A$, we denote
\[
v_i^A:= u_i - \bar u_A.
\]
Then, we observe
\begin{align*}
\dot v_i^A &= \dot u_i - \bar{\dot u}_A = Z_i - \bar Z_A \\
& = \Big[ (Ju)_i - \overline{(Ju)}_A    \Big]  + [\mathcal R_i - \bar{\mathcal R}_A] = - \Lambda_A v_i^A + [\mathcal R_i - \bar{\mathcal R}_A].
\end{align*}
For the last term, we observe
\[
| \mathcal R_i - \bar{\mathcal R}_A| \leq 2\|\mathcal R\|_\infty \leq 6a_\delta \Gamma. 
\]
Hence, we have
\[
\dot R_A \leq -\Lambda_A R_A + 6a_\delta \Gamma.
\]

$\diamond$ (Case 2: $R_B$) For $j\in B$, we denote
\[
v_j^B:= u_j - \bar u_B.
\]
Then similar to Case 1, we observe
\[
\dot v_j^B = - \Lambda_B v_j^B + (\mathcal R_j - \bar{\mathcal R}_B).
\]
Hence, we have
\[
\dot R_B \leq - \Lambda_B R_B + 6a_\delta \Gamma.
\]

$\diamond$ (Case 3: $H$) We denote
\[
h:= \bar u_A - \bar u_B.
\]
Then, we have
\begin{align*}
\dot h &= \frac{d}{dt} (\bar u_A - \bar u_B) = \bar Z_A - \bar Z_B \\
& = \Big[ \overline{(Ju)}_A -  \overline{(Ju)}_B       \Big]    + [\bar{\mathcal R}_A - \bar{\mathcal R}_B] \\
& = - \Lambda_\Delta h + [\bar{\mathcal R}_A - \bar{\mathcal R}_B].
\end{align*}
Since the last term is also bounded by $6a_\delta \Gamma$, we find
\[
\dot H \leq - \Lambda_\Delta H + 6a_\delta \Gamma.
\]
Hence, $\Gamma=\max\{R_A,R_B,H\}$ satisfies
\[
\dot \Gamma \leq - (\Lambda_0 -6a_\delta)\Gamma = -\lambda_\delta \Gamma,\quad t\in [0,T_*)
\]
whenever $\Gamma(t) \leq \delta$ on $t\in [0,T_*)$. \newline

\noindent $\bullet$ (Step 7: Bootstrap argument) Since we assume $\Gamma(0)\leq 2\veps<\delta$, we have
\[
\Gamma(t) \leq \Gamma(0)e^{-\lambda_\delta t} \leq 2\veps e^{-\lambda_\delta t} <\delta,\quad t\in [0,T_*).
\]
However by the definition of $T_*$, we should have
\[
\Gamma(T_*) = \delta
\]
which contradicts. Hence, $T_*=\infty$ and 
\[
\Gamma(t) \leq \Gamma(0) e^{-\lambda_\delta t}.
\]
Hence, $\Gamma (t)\to0$ and $u_i(t) \to 0$ for $i\in [N]$. \newline

\noindent $\bullet$ (Step 8: Convergence of the angular velocities): Recall that
\[
\Theta(t) = c(t)\mathbf{1}  + e + u(t),\quad c(t) := \frac1N \sum_{k=1}^N \theta_k(t).
\]
Since the vector field is invariant under the common phase, we have 
\[
F(\Theta(t)) = F(c(t)\mathbf{1} + e + u(t)) = F(e+u(t)).
\]
Thus, we recall $F(e) = \Omega_*\mathbf{1}$ to find
\[
\dot \Theta(t) -\Omega_*\mathbf{1} = F(e+u(t)) - F(e).
\]
Since $F$ is smooth, there exists $L_\delta>0$ near $e$ such that
\[
\|\dot \Theta(t) - \Omega_*\mathbf{1}\|_\infty = \|F(e+u) - F(e)\|_\infty \leq L_\delta \|u\|_\infty.
\]

\noindent $\bullet$ (Step 9: Convergence of $R(t)^2$) We recall
\[
R(t) = \left| \frac1N \sum_{i=1}^N e^{\mi(e_i + u_i(t))}     \right|,\quad R_* = \left|   \frac1N \sum_{i=1}^N e^{\mi e_i}   \right|.
\]
Hence, we observe
\begin{align*}
|R(t)^2 - R_*^2| \leq 2|R(t) - R_*|\leq 2\left| \frac1N \sum_{i=1}^N e^{\mi e_i}(e^{\mi u_i(t)}-1)    \right| \leq 2\|u(t)\|_\infty \leq 4\Gamma(t) \to0.
\end{align*}

\section{Complementary dynamics} \label{sec:app.C}
\setcounter{equation}{0}

In this Appendix, although these results are not directly related to the main theme of this paper, we include three complementary dynamical regimes to provide a more complete picture and to extend the analysis initiated in our previous work \cite{HK}.

\subsection{A basin for complete desynchronization} 
\begin{theorem}
Suppose that 
\[
\kp_1<0,\quad R_* := \begin{cases}
\vspace{0.2cm} 1,\quad \kp_2=0, \\
\displaystyle \min \left\{1,\frac{-\kp_1}{|\kp_2|}\right\}\quad \kp_2\neq0.
\end{cases}
\]
If initial data   satisfy 
\[
X^0 \notin \mathcal M_{>\frac N2},\quad |R(0)|<R_*,
\]
where
\[
\mathcal M_{>\frac N2} := \bigcup_{I \subseteq[N], |I|>\frac N2}  \{  X\in (\bbs^m)^N: x_i = x_j,\quad \forall i,j\in I \}, 
\]
then complete desynchronization occurs.
\end{theorem}

\begin{remark}
$X^0 \notin \mathcal M_{>\frac N2} $ is equivalent to 
\[
\max_{y\in\bbsm} |\{i : x_i^0=y\}| \leq \frac N2. 
\]
In other words, the initial configuration does not contain a strict majority synchronized cluster, or we assume that no point on the sphere is initially occupied by more than half of the agents. 
This condition is automatically satisfied if the initial positions are pairwise distinct.  
\end{remark}

\begin{proof}
We observe
\[
\frac12\frac{d}{dt} |x_c|^2 = \frac1N \sum_{i=1}^N (\kp_1+\kp_2\langle x_i,x_c\rangle) (|x_c|^2-\langle x_i,x_c\rangle^2).
\]
and write $q_i = \langle x_i,x_c\rangle $  and $g_{ij} = 1-\langle x_i,x_j\rangle$ 
\[
\dot g_{ij} = \kp_2(q_i - q_j)^2 - \Big(  \kp_1(q_i+q_j) + \kp_2(q_i^2 + q_j^2)\Big) g_{ij}.
\]
In addition, since
\[
\kp_1+\kp_2\langle x_i,x_c\rangle \leq \kp_1+|\kp_2||x_c(t)|,
\]
 as long as $|x_c(t)|<R_*$,  the order parameter is non-increasing. Here, we assume $|x_c(0)|<R_*$ and we have
\[
|x_c(t)|<R_*,\quad t>0. 
\]
We write
\[
c:= R_\infty:= \lim_{t\to\infty} R(t)\geq0
\]
and define
\[
\delta_0 := -(\kp_1+|\kp_2|R_0)>0.
\]
Since we have
\[
\kp_1+\kp_2\langle x_i,x_c\rangle \leq \kp_1+|\kp_2| R(t) \leq \kp_1+|\kp_2|R_0 = -\delta_0, \quad i\in [N],\quad t>0,
\]
we have
\[
\frac12\frac{d}{dt} |x_c|^2 \leq -\frac{\delta_0}{ N} \sum_{i=1}^N (|x_c|^2 - \langle x_i,x_c\rangle^2)
\]
which gives from   Barbalat's lemma 
\[
\lim_{t\to\infty} \langle x_i,x_c\rangle^2 = c^2,\quad \lim_{t\to\infty} |\langle x_i,x_c\rangle| = c.
\]
If $c=0$, then complete desynchronization already follows. Hence, in what follows, we assume $c>0$. 
For sufficiently large $t$, $\langle x_i,x_c\rangle$ cannot cross zero, and so each $\langle x_i,x_c\rangle$ has a fixed eventual sign. Hence, we have
\[
\lim_{t\to\infty} \langle x_i,x_c\rangle =  \sigma_i c.
\]
Our goal is to show that
\[
R_\infty>0  \quad \Longleftrightarrow \quad  X^0 \in \mathcal M_{>\frac N2}. 
\]
$(\Longleftarrow)$: If $X^0 \in \mathcal M_{>\frac N2}$, then $R(t)$ cannot converge to zero. Precisely, if there exists $I\subseteq \{1,\cdots,N\}$ with $|I|=n >\frac N2$ and $x_i^0 = x_j^0$ for $i,j \in I$, then we have
\[
R(t) \geq \frac{2n-N}{N}>0.
\]
First, we know that if $x_i^0 = x_j^0$, then $x_i(t) = x_j(t)$ for $t\geq0$. Thus, for $i\in I$, we write 
\[
x_i(t)= y_i(t),\quad i \in I.
\]
Then, we observe
\[
Nx_c(t) = \sum_{k=1}^N x_k(t) = ny(t) + \sum_{k\neq i} x_k(t)
\]
which gives
\[
N|x_c(t)| \geq  \left| n|y(t)| - \left\|  \sum_{k\neq I} x_k(t)  \right\|     \right| \geq n -(N-n) = 2n-N.
\]
Hence, $R(t)$ satisfies
\[
R(t) \geq \frac{2n-N}{N}>0.
\]
$(\Longrightarrow)$: Assume $R_\infty = c>0$. Then, we will show that the set $I_+:=\{ i: \sigma_i =1\}$ has strict majority, i.e., 
\[
|I_+| >\frac N2.
\]
Since $R(t)\to c>0$ and $q_i(t) \to \sigma_i c$ with $\sigma \in \{\pm1\}$, we have
\[
c^2 = \frac cN \sum_{i=1}^N \sigma_i,\quad \textup{or, equivalently,} \quad c = \frac1N \sum_{i=1}^N \sigma_i.
\]
Denote
\[
N_+ := |\{  i : \sigma_i = 1\}|,\quad N_- := |\{  i : \sigma_i = -1\}|.
\]
Then, we have
\[
N_+ + N_- = N,\quad c = \frac{N_+-N_-}{N}.
\]
Since $c>0$, we have
\[
N_+>\frac N2.
\]
Define
\[
\mathcal M_{>\frac N2} := \bigcup_{I \subseteq 1,\cdots,N, |I|>\frac N2}  \{x_i^0 = x_j^0,\quad i,j\in I\}. 
\]
Conversely, if $\kp_1<0$ and $R_\infty>0$, then the initial data must belong to $\mathcal M_{>\frac N2}$. In other words, 
\[
R_\infty >0 ~~ \Longrightarrow ~~ X^0 \in \mathcal M_{>\frac N2}.
\]
Previously, we have shown that $|I_+|>\frac N2$. Take any pair $i,j \in I_+$. Then, 
\[
x_i(t) - x_j(t) \to 0,\quad g_{ij}\to0.
\]
We split it into two cases. First, suppose that $\kp_2\geq0$. Since 
\[
\kp_2(q_i-q_j)^2\geq0 
\]
and
\[
-[\kp_1(q_i+q_j) + \kp_2(q_i^2 + q_j^2) ]  \to -2c(\kp_1+\kp_2c ) >0.
\]
Hence, there exist $T>0$ and $\lambda>0$ such that
\[
\dot g_{ij} \geq \lambda g_{ij},\quad t\geq T.
\]
Second, suppose that $\kp_2<0$. We observe
\[
(q_i - q_j)^2 = \langle x_i-x_j,x_c\rangle^2 \leq 2|x_c|^2g_{ij}.
\]
Hence, 
\[
\dot g_{ij} \geq E_{ij}(t) g_{ij},\quad E_{ij}:= 2\kp_2|x_c|^2 - \kp_1(q_i+q_j) - \kp_2(q_i^2 + q_j^2) \to -2\kp_1 c>0.
\]
Hence, again, there exist $T>0$ and $\lambda>0$ such that
\[
\dot g_{ij} \geq \lambda g_{ij},\quad t\geq T.
\]
This gives
\[
g_{ij}(t) \geq g_{ij}(T)e^{\lambda(t-T)},\quad t\geq T.
\]
Since we already know that $g_{ij}$ converges to zero, we should have $g_{ij}(T)=0$. But this only comes from $g_{ij}^0=0$, i.e., $x_i^0=x_j^0$. Hence, $X^0 \in \mathcal M_{>\frac N2}$. 

This shows that if we assume $X^0 \notin \mathcal M_{>\frac N2}$ and $R(0)<R_*$, then $c=0$ and  complete desynchronization emerges. 

\end{proof}

\subsection{Non-desynchronizability}

\begin{theorem}
Suppose that $\kp_1>0$ and $x_c^0\neq0$. Then, system \eqref{main}--\eqref{init} cannot exhibit complete desynchronization. 
\end{theorem}

\begin{proof}
We observe 
\[
\frac12\frac{d}{dt} |x_c|^2 = \frac1N \sum_{i=1}^N (\kp_1+\kp_2\langle x_i,x_c\rangle )(|x_c|^2 - \langle x_i,x_c\rangle^2) .
\]
Choose
\[
\rho:= \begin{cases}
\displaystyle\frac{\kp_1}{2|\kp_2|},\quad \kp_2\neq0, \\
\infty,\quad \kp_2=0.
\end{cases}
\]
Then, whenever $|x_c(t)|<\rho$, then we have
\[
\kp_1+\kp_2\langle x_i,x_c\rangle \geq \kp_1-|\kp_2||x_c(t)| \geq \frac{\kp_1}{2}>0
\]
which gives
\[
\frac12\frac{d}{dt} |x_c|^2  \geq \frac{\kp_1}{2N}\sum_{i=1}^N (|x_c|^2 - \langle x_i,x_c\rangle^2)\geq0.
\]
Suppose to the contrary that $|x_c(t)|$ converges to zero as $t\to\infty$. Then, there exists $T>0$ such that
\[
|x_c(t)|<\rho,\quad t\geq T. 
\]
Hence, $|x_c(t)|^2$ is non-decreasing on $[T,\infty)$. Since we assume $x_c(0)\neq0$, this contradicts.

\end{proof}

\subsection{Complete bipolar synchronization}
Define the signed diameter: for $\sigma_i \in \{1,-1\}$, 
\[
\mathcal D_\sigma(t) := \max_{1\leq i,j\leq N} |1-\sigma_i\sigma_j\langle x_i(t),x_j(t)\rangle|.
\]
\begin{theorem}
Suppose that
\[
\kp_1<0,\quad \kp_1+\kp_2>0, \quad \beta:= \frac{-\kp_1}{\kp_2}\in (0,1), 
\]
Fix a bipartition of $[N]$ determined by signs $\sigma_i \in \{\pm1\}$ and let 
\[
N_+:=\#\{i:\sigma_i=1\},\quad N_-:=\#\{i:\sigma_i=-1\}.
\]
Without loss of generality, assume $N_+>N_-\geq1$ and set 
\[
c:= \frac{N_+-N_-}{N}.
\]
If $c>\beta$, then there exists an explicit neighborhood of the corresponding bipolar configuration such that every solution starting in this neighborhood converges to complete bipolar synchronization. More precisely, there exists $\veps_B>0$ such that if the initial signed diameter is sufficiently small, then there exists $p_\infty\in \bbs^m$ for which 
\[
\lim_{t\to\infty} x_i(t) = \sigma_i p_\infty,\quad i\in [N],\quad \lim_{t\to\infty} \|x_c(t)\| = c = \frac{N_+-N_-}{N}.
\]
Moreover, the signed diameter decays exponentially. 
\end{theorem}

\begin{proof}
Let $\sigma_i \in \{1,-1\}$ and imagine that the agent with $\sigma_i=1$ converges to $p$ while the agent with $\sigma_i=-1$ converges to $-p$. Denote
\[
N_+ := |\{  i : \sigma_i = 1\}|,\quad N_- := |\{  i : \sigma_i = -1\}|.
\]
Without loss of generality, we assume $N_+>N_-$ and denote
\[
c:= \frac{N_+-N_-}{N}\in (0,1).
\]
We need to assume
\[
\frac{N_+-N_-}{N}>\frac{-\kp_1}{\kp_2}, \quad \textup{i.e.,}\quad c>\beta.
\]
Denote signed vector
\[
y_i := \sigma _i x_i
\]
and our goal is to show that all $y_i$ converge  to $p$. Denote
\[
x_c = m = \frac1N \sum_{j=1}^N x_j = \frac1N \sum_{j=1}^N \sigma_j y_j,\quad w_i := \langle y_i,m\rangle. 
\]
Then, $y_i$ satisfies
\[
\dot y_i = \sigma_i \dot x_i = (\sigma_i \kp_1 +\kp_2w_i ) (m - w_i y_i),\quad b_i := \sigma_i \kp_1+\kp_2w_i.
\]
At the target bipolar state $w_i = c$, the corresponding coefficients are 
\[
b_i = \begin{cases}
\kp_1 +\kp_2c = \kp_2(c-\beta),\quad \sigma_i=1, \\
-\kp_1 + \kp_2c = \kp_2(c+\beta),\quad \sigma_i=-1.
\end{cases}
\]
Hence, if $c>\beta$, then $b_i>0$ for all $i$. Define the signed pairwise distance
\[
H_{ij}:= 1-\langle y_i,y_j\rangle = 1-\sigma_i \sigma_j \langle x_i,x_j\rangle
\]
which satisfies
\[
\dot H_{ij} = (b_i - b_j)(w_i - w_j) - (b_i w_i + b_j w_j)H_{ij}.
\]
If complete bipolar synchronization occurs, then $m\to cp$ and $w_i= \langle y_i,m\rangle \to c$. Hence,  we denote
\[
e_i := w_i - c .
\]
Also, we write $w_i$ in terms of $H_{ij}$:
\begin{align*}
w_i = \langle y_i,m\rangle = \frac1N \sum_{k=1}^N \sigma_k \langle y_i,y_k\rangle = \frac1N \sum_{k=1}^N \sigma_k ( 1-H_{ik}) = c- \frac1N \sum_{k=1}^N \sigma_k H_{ik}
\end{align*}
which gives
\[
e_i = -\frac1N \sum_{k=1}^N \sigma_k H_{ik}.
\]
For simplicity, if we denote
\[
\Phi_s(w) = w(w-\beta s),\quad s\in \{1,-1\},
\]
then the dynamics of $H_{ij}$ becomes
\[
\dot H_{ij} = \kp_2(e_i - e_j)^2 - \kp_2\beta(\sigma_i - \sigma_j)(e_i - e_j) - \kp_2[ \Phi_{\sigma_i}(w_i ) + \Phi_{\sigma_j}(w_j)]H_{ij}.
\]
Denote
\[
P:= \{ i:\sigma_i=1\},\quad M:= \{i : \sigma_i=-1\}, \quad |P|=N_+,\quad |M|=N_-.
\]
Define the diameters
\[
D_+ := \max_{i,j\in P} H_{ij},\quad D_-:= \max_{i,j\in M} H_{ij}, \quad D_0 := \max_{i\in P, a\in M} H_{ia},
\]
and
\[
\mathcal H_\sigma := A D_+ + BD_- + D_0, 
\]
where $A,B\geq1$ are constants that will be determined later. If there is no pair or singleton, then the diameter becomes zero. If $i\in P$, then $\sigma_i=1$ and 
\[
e_i = -\frac1N \sum_{k=1}^N \sigma_k H_{ik} = -\frac1N \left(  \sum_{k\in P} H_{ik} - \sum_{a\in M} H_{ia}  \right) = \frac1N \left(  \sum_{a\in M} H_{ia} - \sum_{k\in P} H_{ik}  \right) ,\quad i\in P.
\]
Similarly, if $a\in M$, 
\[
e_a = \frac1N  \left(  \sum_{b\in M} H_{ab} - \sum_{k\in P} H_{ak}  \right) ,\quad a\in M. 
\]
We a priori assume, for a sufficiently small $\veps_0>0$, 
\[
\max_{i,j} H_{ij} \leq \veps_0.
\]
Then, 
\[
|e_i|\leq \veps_0,\quad w_i = c+e_i \in [c-\veps_0,c+\veps_0]. 
\]
Now, we choose $\veps_0 := \frac{c-\beta}{2}$. Then, we observe
\[
c-\veps_0 = \frac{c+\beta}{2}>0,\quad c-\veps_0 - \beta = \frac{c-\beta}{2} >0.
\]
Thus, for $s=1$, 
\[
\Phi_+(w) = w(w-\beta) \geq (c-\veps_0)(c-\veps_0-\beta) = \frac{c+\beta}{2}\frac{c-\beta}{2} = \frac{c^2-\beta^2}{4}>0.
\]
On the other hand for $s=-1$, 
\[
\Phi_-(w) = w(w+\beta) \geq (c-\veps_0)(c-\veps_0+\beta) >0.
\]
Thus, for any pair $(i,j)$, we have a positive lower bound
\[
\Phi_{\sigma_i}(w_i ) + \Phi_{\sigma_j}(w_j) \geq 2\cdot \frac{c^2-\beta^2}{4} = \frac{c^2-\beta^2}{2}.
\]
Hence, we have
\[
\kp_2 (\Phi_{\sigma_i}(w_i ) + \Phi_{\sigma_j}(w_j)) \geq \lambda :=\frac{ \kp_2}{ 2} (c^2-\beta^2) .
\]
Suppose that $\mathcal H_\sigma \leq \veps_0$. Since $A,B\geq1$, 
\[
D_+,D_-,D_0\leq \mathcal H_\sigma. 
\]
$\bullet$ (Inequality for $D_+$): Since $i,j\in P$, we have $\sigma_i = \sigma_j=1$. Hence, we have
\[
\dot H_{ij} = \kp_2(e_i - e_j)^2 -  \kp_2(\Phi_{\sigma_i}(w_i ) + \Phi_{\sigma_j}(w_j))H_{ij}.
\]
We observe
\[
|e_i - e_j|^2 \leq 4\mathcal H_\sigma^2. 
\]
Thus, we have
\[
\dot H_{ij} \leq -\lambda H_{ij} + 4\kp_2\mathcal H_\sigma^2
\]
which gives
\[
D_+' (t) \leq -\lambda D_+ + 4\kp_2\mathcal H_\sigma^2.
\]
$\bullet$ (Inequality for $D_-$): If $i,j \in M$, we have $\sigma_i = \sigma_j=-1$. Hence, by a similar argument, we have
\[
D_-' (t) \leq -\lambda D_- + 4\kp_2\mathcal H_\sigma^2.
\]
$\bullet$ (Inequality for $D_0$): For $i\in P$ and $a\in M$, we have $\sigma_i - \sigma_a=2$. Then, we have
\[
\dot H_{ia} = \kp_2(e_i - e_a)^2 - 2\kp_2\beta(e_i - e_a) - \kp_2( \Phi_+(w_i) + \Phi_-(w_a))H_{ia}. 
\] 
We need to focus on the second term. For $i\in P$, 
\[
e_i = \frac1N \left( \sum_{b\in M} H_{ib} - \sum_{k\in P} H_{ik}\right) \geq - \frac1N \sum_{k\in P} H_{ik} \geq - \frac{N_+-1}{N} D_+.
\]
Similarly for $a\in M$, 
\[
e_a = \frac1N\left(  \sum_{b\in M} H_{ab} - \sum_{k\in P} H_{ak}  \right) \leq \frac1N \sum_{b\in M} H_{ab} \leq \frac{N_--1}{N}D_-.
\]
Hence, we have
\[
e_i - e_a \geq -\frac{N_+-1}{N}D_+ - \frac{N_--1}{N} D_-
\]
and
\[
-2\kp_2\beta(e_i - e_a) \leq 2\kp_2\beta \left( \frac{N_+-1}{N}D_+ + \frac{N_--1}{N} D_- \right) =:a_+D_+ + a_-D_{-},
\]
where
\[
a_+ := \frac{2\kp_2\beta(N_+-1)}{N},\quad a_- := \frac{2\kp_2\beta(N_--1)}{N}.
\]
Hence, we have
\[
\dot D_0(t) \leq -\lambda D_0 + a_+D_+ + a_-D_- + 4\kp_2\mathcal H_\sigma^2.
\]
Since
\[
\mathcal H_\sigma := AD_+ + BD_- + D_0, 
\]
where 
\[
A:= \max\left\{  1,\frac{2a_+}{\lambda}   \right\},\quad B:= \max\left\{ 1,\frac{2a_-}{\lambda} \right\}, 
\]
we have
\[
A\lambda - a_+ \geq \frac{A\lambda}{2},\quad B\lambda - a_-\geq \frac{B\lambda}{2}.
\]
Now, we add all the inequalities to find
\[
\mathcal H_\sigma ' \leq -\frac{\lambda}{2}\mathcal H_\sigma + 4\kp_2(A+B+1)\mathcal H_\sigma^2.
\]
Note that this inequality holds when $\mathcal H_\sigma \leq \veps_0$. Now, we set
\[
\veps_B := \min \left\{ \veps_0,\frac{ \lambda}{16\kp_2(A+B+1)} \right\}.
\]
If we assume that $\mathcal H_\sigma (0) <\veps_B$, then we have 
\[
\mathcal H_\sigma(t) <\veps_0,\quad t>0,
\]
and 
\[
\mathcal H_\sigma ' \leq -\frac\lambda4\mathcal H_\sigma, \quad t>0.
\]
Hence, $\mathcal H_\sigma(t)$ decays exponentially. In addition, we observe
\[
x_c - cy_i = \frac1N \sum_{k=1}^N \sigma_k ( y_k - y_i)
\]
which gives
\[
\|x_c - cy_i\| \leq \sqrt{2\mathcal H_\sigma(t)}.
\]
Thus, 
\[
\|\dot y_i(t)\| \leq (|\kp_1|+\kp_2) \sqrt{2\mathcal H_\sigma(t)}
\]
which implies
\[
\int_0^\infty \|\dot y_i(t)\|dt <\infty.
\]
This gives 
\[
\lim_{t\to\infty} x_i(t)  = \sigma_i p_\infty, \quad \lim_{t\to\infty} \|x_c(t)\| = \frac{N_+-N_-}{N}.
\]

\end{proof}

\begin{remark}
The condition $c>\beta$ is not an existence condition for bipolar equilibria; rather, it is the threshold which ensures local attraction of the majority cluster. 
\end{remark}

  \begin{figure}[H]
\centering
\mbox{
\subfigure{  
\includegraphics[width=1\textwidth]{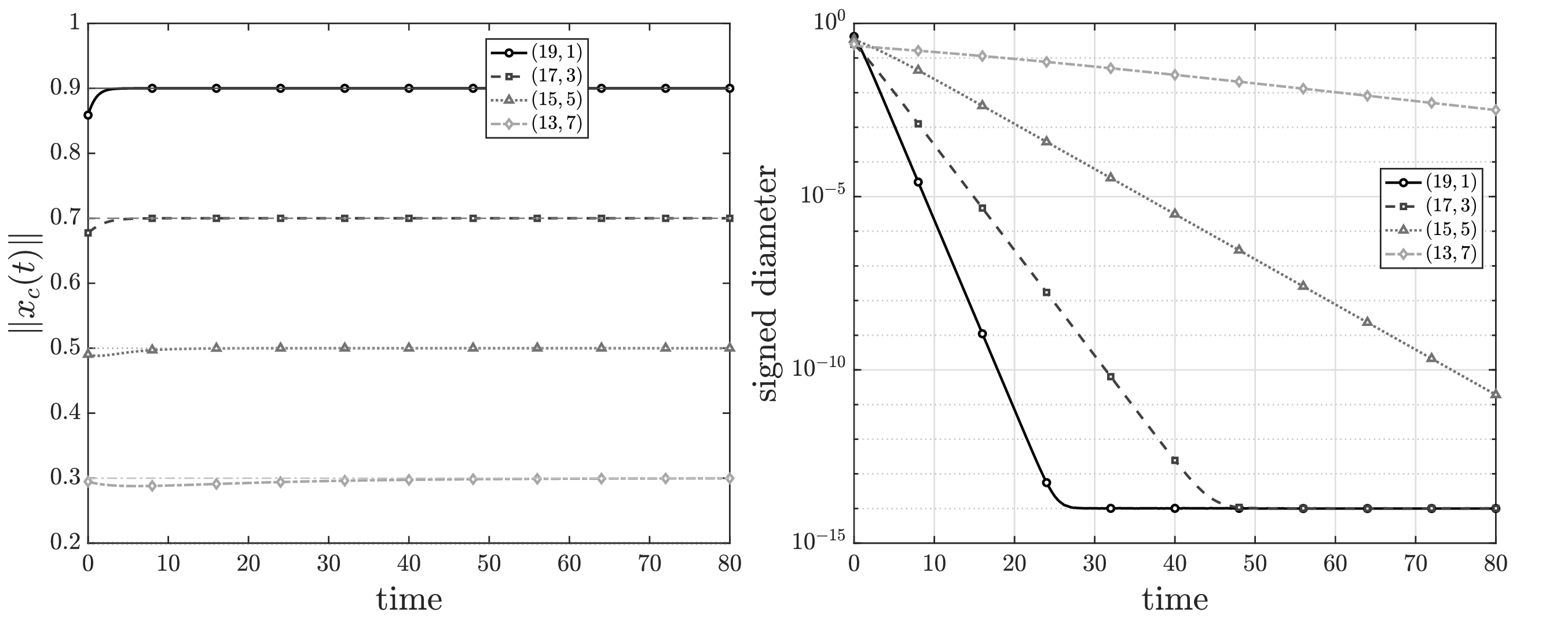}\label{fig:3-1}}
}
\caption{ $(\kp_1,\kp_2,N)=(-0.2,1,20)$: The left panel shows the temporal evolution of the   order parameter for four bipartite partitions $(N_+,N_-)=(19,1), (17,3), (15,5), (13,7)$. In each case, the particles converge to two antipodal clusters $x_i\to\pm p$ and hence the order parameter converges to $\frac{|N_+-N_-|}{N}$. The corresponding limiting values are 0.9, 0.7, 0.5, 0.3 which are larger than $\frac{-\kp_1}{\kp_2}=0.2$. The right panel shows the signed diameter on a logarithmic scale whose exponential decay confirms convergence to complete bipolar synchronization in agreement with the sufficient stability condition $c>\beta=0.2$.     }  \label{fig5-1}
\end{figure}

\begin{remark}
Explicit classes of initial configurations leading to bi-cluster formation have been constructed for related agent-based system, such as the Cucker-Smale model \cite{CHHJK}. In addition, the formation of two subgroups under competing attractive and repulsive interactions is also reminiscent of bi-cluster flocking in \cite{FHJ}. 
\end{remark}

 \subsection{Epilogue}
 So far, when $\kp_1>0$ and $\kp_1+\kp_2<0$,  we identify the balanced value $R^2=\alpha$, construct an explicit basin of attraction leading to it and show that non-balanced nonzero-mean equilibria are linearly unstable.   In Appendix \ref{sec:app.C}, we study emergent dynamics in slightly different settings. First, when $\kp_1<0$, regardless of the sign of $\kp_2$,  we   provide a sufficient condition leading to the complete desynchronization where the order parameter converges to 0. Second, we see that if $\kp_1>0$, regardless of the sign of $\kp_2$, then the order parameter cannot converge to zero. In other words, complete desynchronization cannot happen. Hence,  the sign of $\kp_1$ determines whether complete desynchronization can occur from nontrivial initial data. More precisely, complete desynchronization can occur when $\kp_1<0$, whereas it can never happen when $\kp_1>0$. 
 
 Lastly, when $\kp_1<0$ and $\kp_1+\kp_2>0$, we provide an explicit basin of attraction leading to complete bipolar synchronization where the agents split into two groups, each converging to one of the two poles of $\bbs^m$. 
 
 Although our main focus is on the regime $\kp_1>0$ and $\kp_1+\kp_2<0$, we also include additional results to enhance the completeness of the paper and improve the results in the previous work \cite{HK}. The results for general $\kp_1$ and $\kp_2$ are summarized in the figures below.

 \begin{figure}[h]
\centering
\mbox{
\subfigure{  
\includegraphics[width=1\textwidth]{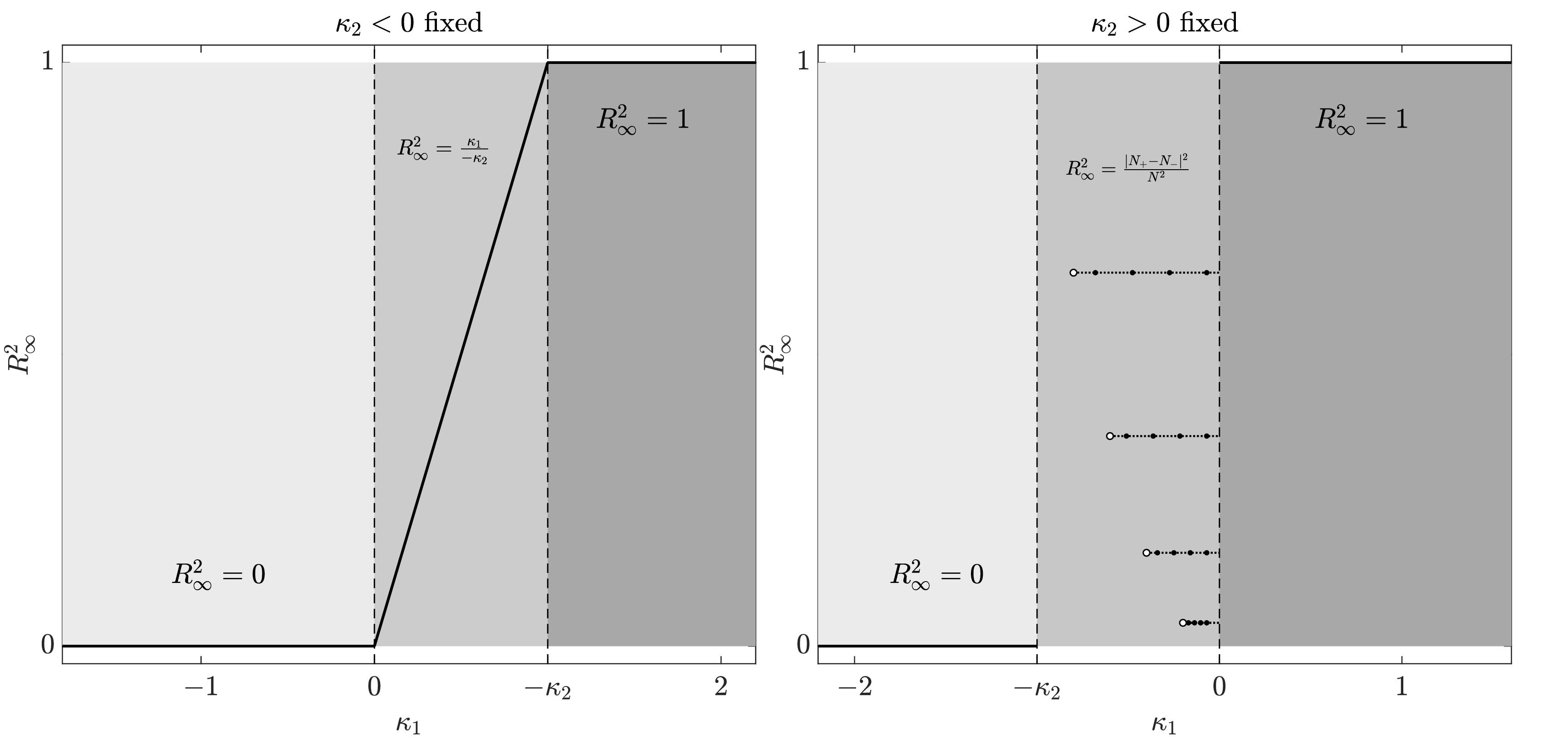}\label{fig:3-1}}
}
\caption{  (Summary of the dimension-dependent asymptotic regimes): In the regime $\kp_1>0,\kp_2<0,\kp_1+\kp_2<0$, the dynamics admits balanced states with the continuously selected value $R_\infty^2 = \alpha$. In contrast, in the regime $\kp_1<0,\kp_2>0,\kp_1+\kp_2>0$, complete bipolar synchronization gives discrete order-parameter levels determined by the population imbalance. Thus, the two competing parameter regimes exhibit qualitatively different mechanisms of the order-parameter selection. }  \label{fig5-1}
\end{figure}

 \begin{figure}[h]
\centering
\mbox{
\subfigure{  
\includegraphics[width=0.5\textwidth]{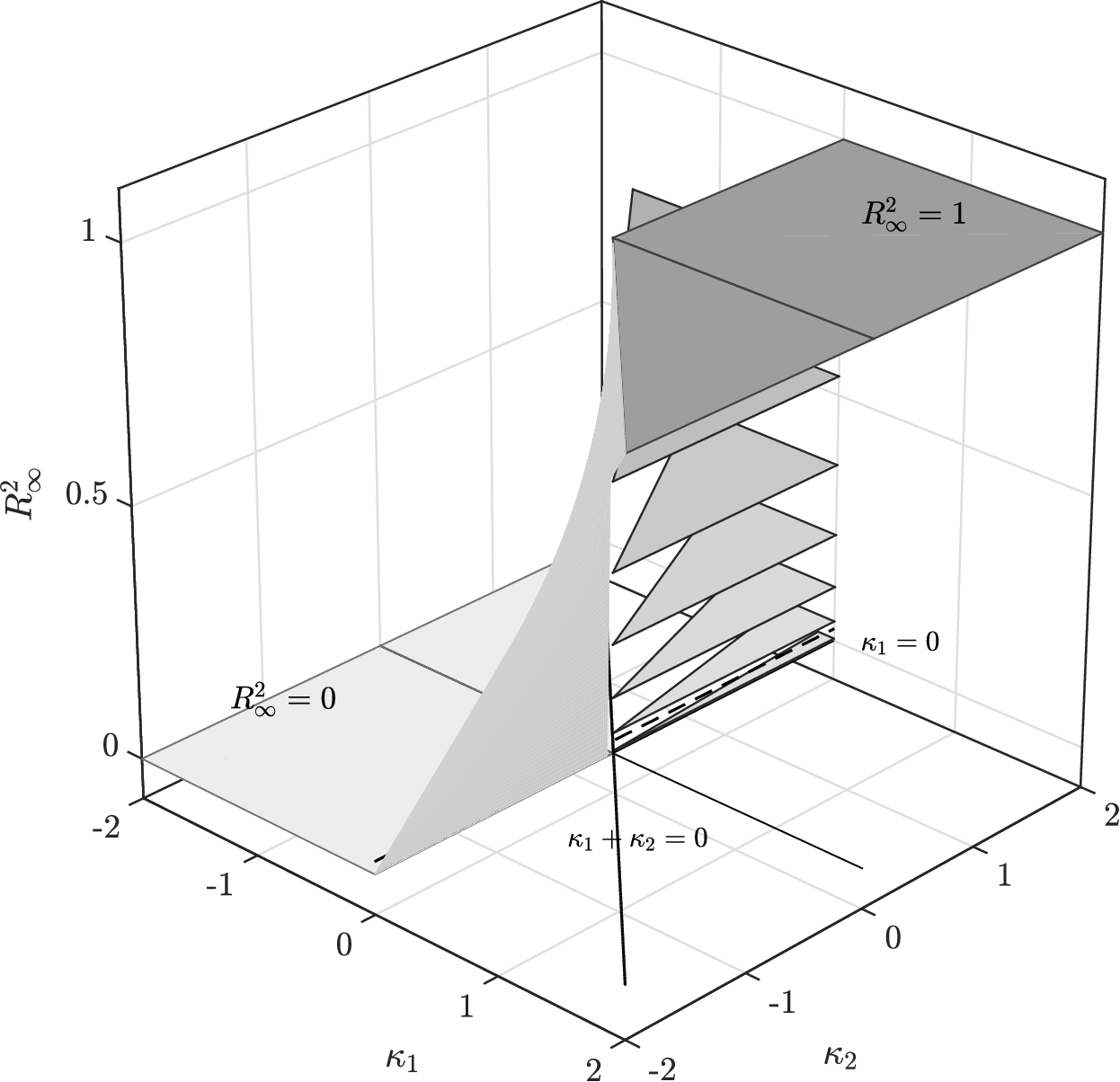}\label{fig:3-1}}
\subfigure{  
\includegraphics[width=0.5\textwidth]{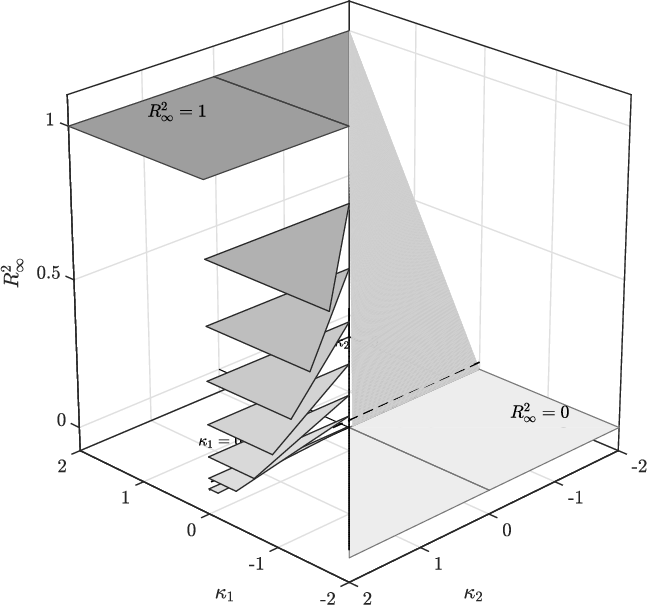}\label{fig:3-1}}
}
\caption{ Three-dimensional phase diagram of the admissible asymptotic branches of $R_\infty^2$ in the $(\kp_1,\kp_2)$-parameter space shown from two different viewing angles.   The horizontal sheets $R_\infty^2=0$ and $R_\infty^2=1$ represent complete desynchronization and complete synchronization, respectively. In the regime $\kp_1>0,\kp_2<0,\kp_1+\kp_2<0$, the inclined surface $R_\infty^2=\alpha$ represents the balanced latitude branch.  For $\kp_2>0$, complete bipolar synchronization gives the discrete quantized levels $R_\infty^2 = \frac{(N_+-N_-)^2}{N^2}$ where only the portions satisfying $\frac{|N-2k|}{N}>\beta=\frac{-\kp_1}{\kp_2}$ are stable. The planes $\kp_1=0$ and $\kp_1+\kp_2=0$ mark the critical boundaries for desynchronizability and synchronizability, respectively, while the coexistence of several sheets indicates multistability and basin-dependent selection.  }  \label{fig5-1}
\end{figure}

 \begin{figure}[H]
\centering
\mbox{
\subfigure{  
\includegraphics[width=0.5\textwidth]{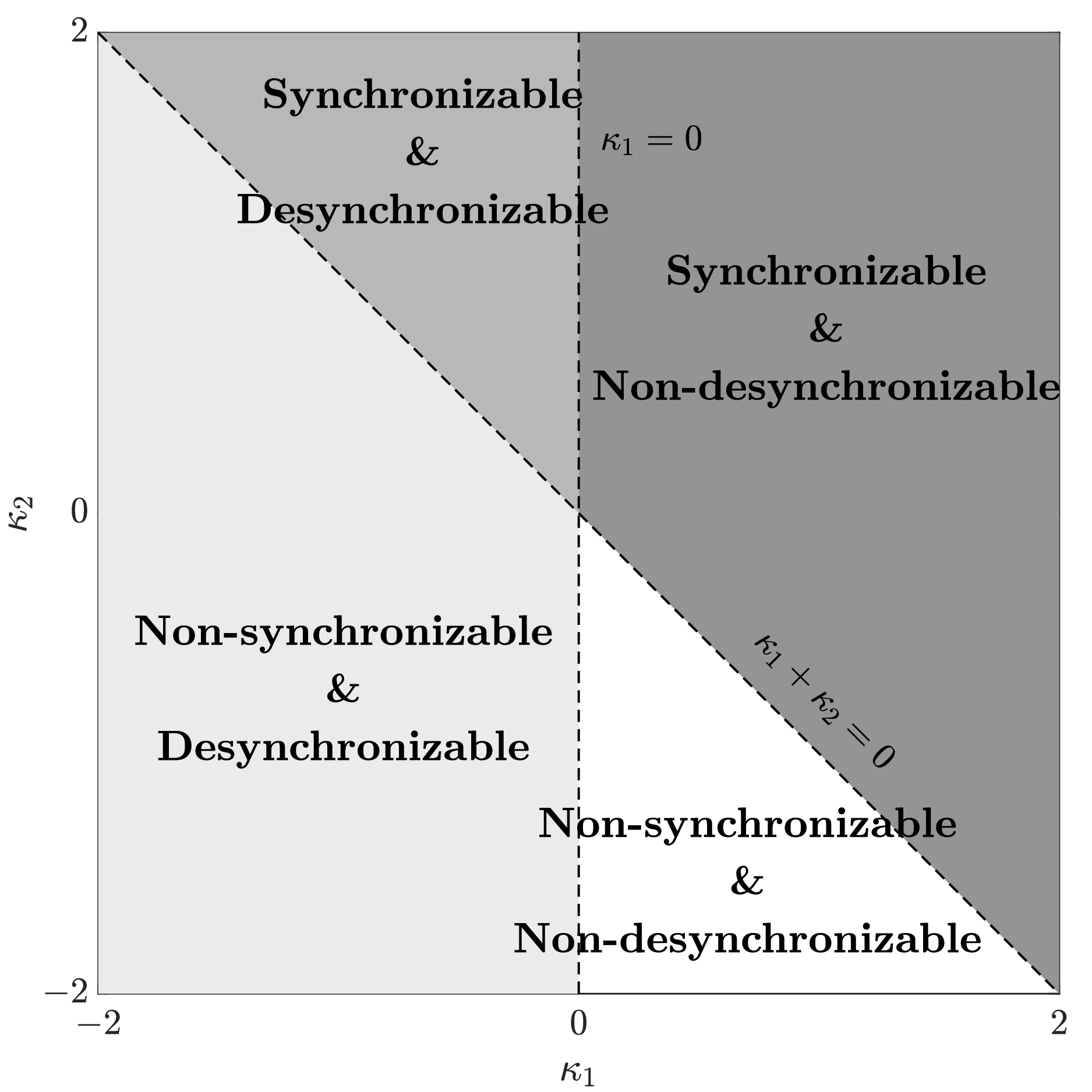}\label{fig:3-1}}
\subfigure{  
\includegraphics[width=0.5\textwidth]{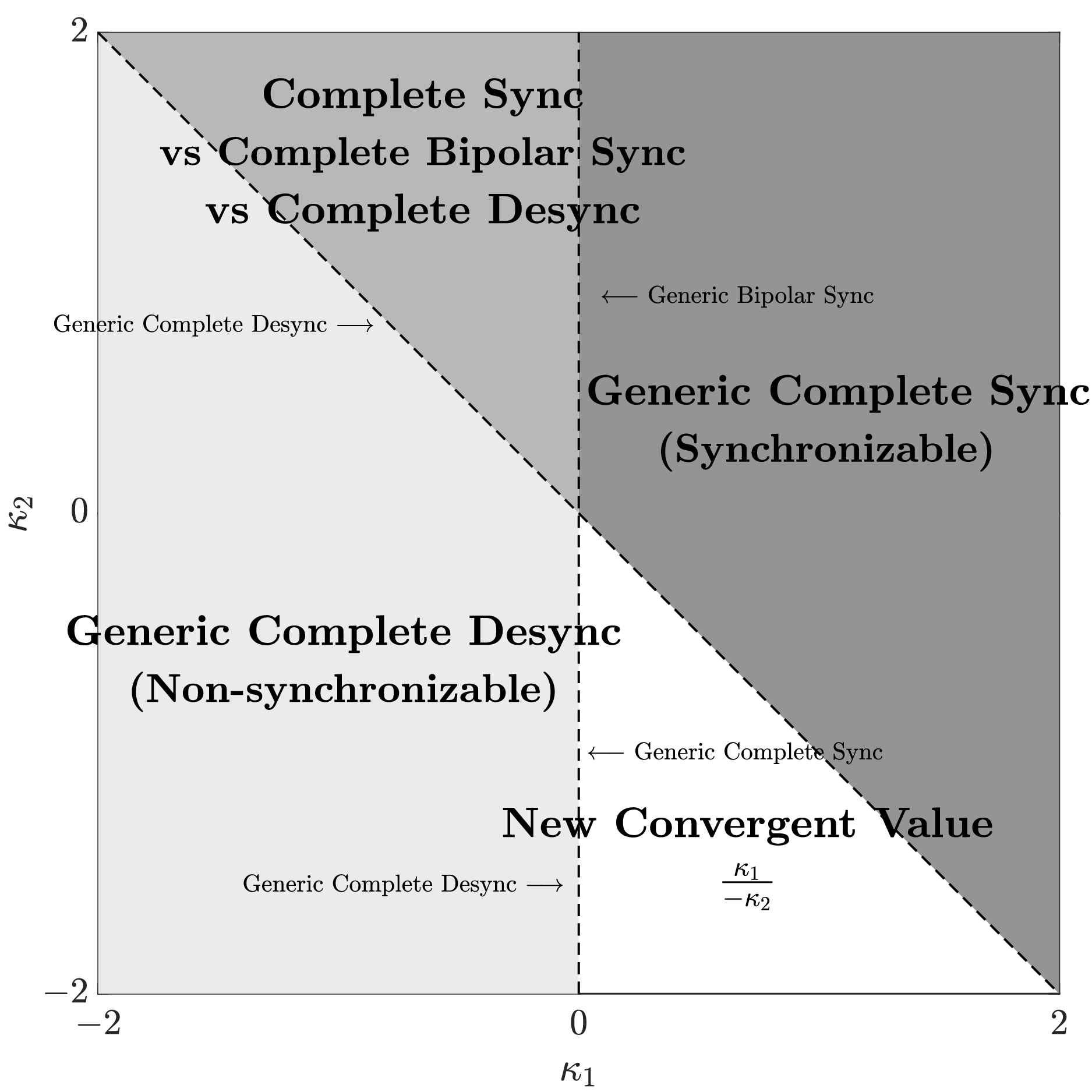}\label{fig:3-1}}
}
\caption{ Classification of the $(\kp_1,\kp_2)$-parameter plane according to the complementary dynamical regimes studied in this paper. The left panel indicates the regions associated with synchronizability and desynchronizability, with the lines $\kp_1+\kp_2=0$ and $\kp_1=0$ marking the corresponding critical boundaries. The right panel supplements this classification with the asymptotic behaviors established in the present work. }  \label{fig5-1}
\end{figure}

\end{document}